\documentclass[11pt]{amsart}
\usepackage[a4paper,left=25mm,top=20mm,right=20mm,bottom=15mm]{geometry}

\usepackage{latexsym}
\usepackage{amsmath}
\usepackage{amssymb}
\usepackage{MnSymbol}
\usepackage{graphicx}

\usepackage{auto-pst-pdf}

\usepackage[colorlinks]{hyperref}
\hypersetup{
	citecolor=blue,
   linkcolor=blue,
}

\newcommand{\citet}{\cite}
\newcommand{\citep}{\cite}
\newcommand{\citealp}{\cite}
\renewcommand{\emph}{\textit}

\newcommand{\empht}{\textit}

\newcommand{\lca}{\ensuremath{\operatorname{lca}}}

\newcommand{\rt}[1]{\ensuremath{\mathsf{#1}}}     
\newcommand{\mc}[1]{\ensuremath{\mathcal{#1}}}
\newcommand{\CoT}{\ensuremath{\operatorname{CoT}}}
\newcommand{\cl}{\ensuremath{\operatorname{cl}}}
\newcommand{\Gen}{\ensuremath{\mathfrak{G}}}
\newcommand{\Spe}{\ensuremath{\mathfrak{S}}}
\newcommand{\symdiff}{\vartriangle}
\newcommand{\al}{\alpha}
\newcommand{\be}{\beta}
\newcommand{\ga}{\gamma}
\newcommand{\de}{\delta}

\renewcommand{\S}{\mathcal{L}}
\newcommand{\ALF}{\texttt{ALF}}

\usepackage{empheq}
\usepackage{xr}
\newtheorem{innercustomthm}{Theorem}
\newenvironment{cthm}[1]
  {\renewcommand\theinnercustomthm{#1}\innercustomthm}
  {\endinnercustomthm}

\newtheorem{innercustompro}{Proposition}
\newenvironment{cpro}[1]
  {\renewcommand\theinnercustompro{#1}\innercustompro}
  {\endinnercustompro}

\newcounter{ILPcounter}
\makeatletter
\newenvironment{ILP}
{%
  \refstepcounter{ILPcounter}
  \start@align\@ne\st@rredtrue\m@ne
  \tag{\textbf{\texttt{\small ILP\ {\theILPcounter}}}}
}{%
  \endalign
}
\makeatother

\providecommand{\keywords}[1]{\textbf{\textit{Keywords: }} #1}

\newtheorem{thm}{Theorem}
\newtheorem{lem}{Lemma}
\newtheorem{proposition}{Proposition}
\newtheorem{defi}{Definition}

\newtheorem{rem}{Remark}

\title[Phylogenomics with Paralogs]{Phylogenomics with Paralogs}

\author[Hellmuth et.\ al]{Marc Hellmuth \and Nicolas Wieseke \and Marcus Lechner  \and Hans-Peter Lenhof \and Martin Middendorf \and Peter F. Stadler}

\thanks{  We thank Jiong Guo, Leo van Iersel, Daniel St\"ockel, and Jakob
  L. Andersen for helpful comments 
  on the cograph 
	editing problem and the
  ILP formulation. 
  This work was funded by the German Research Foundation
  (DFG) (Proj.\ No.\ MI439/14-1). \\
All addresses and author information at \url{http://www.pnas.org/content/112/7/2058}
}



\date{\today}

\begin{document}

\maketitle
\begin{abstract}
  Phylogenomics heavily relies on well-curated sequence data sets that
  consist, for each gene, exclusively of 1:1-orthologous. Paralogs are
  treated as a dangerous nuisance that has to be detected and removed. We
  show here that this severe restriction of the data sets is not
  necessary. Building upon recent advances in mathematical phylogenetics we
  demonstrate that gene duplications convey meaningful phylogenetic
  information and allow the inference of plausible phylogenetic trees,
  provided orthologs and paralogs can be distinguished with a degree of
  certainty.  Starting from tree-free estimates of orthology, cograph
  editing can sufficiently reduce the noise in order to find correct
  event-annotated gene trees. The information of gene trees can then
  directly be translated into constraints on the species trees. While the
  resolution is very poor for individual gene families, we show that
  genome-wide data sets are sufficient to generate fully resolved
  phylogenetic trees, even in the presence of horizontal gene transfer.

We demonstrate that the distribution of paralogs in large
  gene families contains in itself sufficient phylogenetic signal to infer
  fully resolved species phylogenies. This source of phylogenetic
  information is independent of information contained in orthologous
  sequences and is resilient against horizontal gene transfer. An important
  consequence is that phylogenomics data sets need not be restricted to 1:1
  orthologs.
\end{abstract}




\sloppy

\section{Introduction}
Molecular phylogenetics is primarily concerned with the
reconstruc\-tion of evolutionary relationships between species based on
sequence information.  To this end, alignments of protein or DNA sequences
are employed, whose evolutionary history is believed to be congruent to that
of the respective species.  This property can be ensured most easily in the
absence of gene duplications and horizontal gene transfer. 
Phylogenetic studies judiciously select
families of genes that rarely exhibit duplications (such as rRNAs, most
ribosomal proteins, and many of the housekeeping enzymes). In
phylogenomics, elaborate automatic pipelines such as \texttt{HaMStR}
\cite{Ebersberger:09}, are used to filter genome-wide data sets to at least
deplete sequences with detectable paralogs (homologs in the same species).

\begin{figure}[t]
\includegraphics[bb= 33 628 523 776, scale=1]{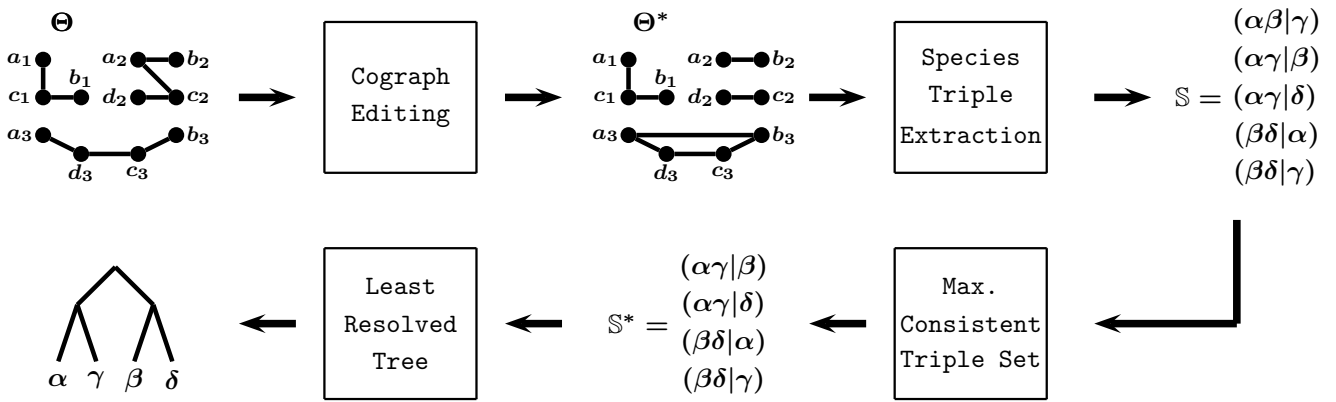}
    \caption{Outline of the computational framework.
	Starting from an estimated
    orthology relation $\Theta$, its graph representation $G_{\Theta}$ is
    edited to obtain the closest cograph $G_{\Theta^*}$, which in turn is
    equivalent to a (not necessarily fully resolved) gene tree $T$ and an
    event labeling $t$.  From $(T,t)$ we extract the set $\mathbb{S}$ of
    all relevant species triples. As the triple set $\mathbb{S}$ need not
    to be consistent, we compute the maximal consistent subset
    $\mathbb{S^*}$ of $\mathbb{S}$. Finally, we construct a least resolved
    species tree from  $\mathbb{S^*}$.} 
  \label{fig:workflow}
\end{figure}

In the presence of gene duplications, however, it becomes necessary to
distinguish between the evolutionary history of genes (\emph{gene trees})
and the evolutionary history of the species (\emph{species trees}) in which
these genes reside. Leaves of a gene tree represent genes. Their inner
nodes represent two kinds of evolutionary events, namely the duplication of
genes within a genome -- giving rise to paralogs -- and speciations, in
which the ancestral gene complement is transmitted to two daughter
lineages. Two genes are (co-)orthologous if their last common ancestor in
the gene tree represents a speciation event, while they are paralogous if
their last common ancestor is a duplication event, see \citet{Fitch2000}
and \citet{GK13} for a more recent discussion on orthology and paralogy
relationships.  Speciation events, in turn, define the inner vertices of a
species tree. However, they depend on both, the gene and the species
phylogeny, as well as the reconciliation between the two. The latter
identifies speciation vertices in the gene tree with a particular
speciation event in the species tree and places the gene duplication events
on the edges of the species tree.  Intriguingly, it is nevertheless
possible in practice to distinguish orthologs with acceptable
accuracy without constructing either gene or species trees
\cite{Altenhoff:09}. Many tools of this type have become available over the
last decade, see \citet{KWMK:11, DAAGD2013} for a recent review. The output
of such methods is an estimate $\Theta$ of the true orthology relation
$\Theta^*$, which can be interpreted as a graph $G_\Theta$ whose vertices
are genes and whose edges connect estimated (co-)orthologs.

Recent advances in mathematical phylogenetics suggest that the
estimated orthology relation $\Theta$ contains information on the structure
of the species tree. To make this connection, we combine here three
abstract mathematical results that are made precise in
\emph{Materials and Methods} below.

(1) Building upon the theory of symbolic ultrametrics \cite{Boeckner:98} we
showed that \emph{in the absence of horizontal gene transfer, the orthology
  relation of each gene family is a cograph} \cite{Hellmuth:13d}. Cographs
can be generated from the single-vertex graph $K_1$ by complementation and
disjoint union \cite{Corneil:81}. This special structure of cographs
imposes very strong constraints that can be used to reduce the noise and
inaccuracies of empirical estimates of orthology from pairwise sequence
comparison. To this end, the initial estimate of $G_{\Theta}$ is modified
to the closest correct orthology relation $G_{\Theta^*}$ in such a way that
a minimal number of edges (i.e., orthology assignments) are introduced or
removed. This amounts to solving the cograph-editing problem
\cite{Liu:11,Liu:12}.

(2) It is well known that \emph{each cograph is equivalently represented by
  its cotree} \cite{Corneil:81}. The cotree is easily computed for a given
cograph.  In our context, the cotree of $G_{\Theta^*}$ is an incompletely
resolved event-labeled gene-tree. That is, in addition to the tree
topology, we know for each internal branch point whether it corresponds to
a speciation or a duplication event.  Even though, adjacent speciations or
adjacent duplications cannot be resolved, the tree faithfully encodes the
relative order of any pair of duplication and speciation
\cite{Hellmuth:13d}. In the presence of horizontal gene transfer $G_\Theta$
may deviate from the structural requirements of a cograph. Still, the
situation can be described in terms of edge-colored graphs whose subgraphs
are cographs \cite{Boeckner:98,Hellmuth:13d}, so that the cograph structure
remains an acceptable approximation.

(3) \emph{Every triple (rooted binary tree on three leaves) in the cotree
  that has leaves from three species and is rooted in a speciation event
  also appears in the underlying species tree}
\cite{hernandez2012event}. Thus, the estimated orthology relation, after
editing to a cograph and conversion to the equivalent event-labeled gene
tree, provide many information on the species tree.  This result allows us
to collect from the cotrees for each gene family partial information on the
underlying species tree. Interestingly, only gene families that harbor
duplications, and thus have a non-trivial cotree, are informative. If no
paralogs exist, then the orthology relation $G_\Theta$ is a clique (i.e.,
every family member is orthologous to every other family member) and the
corresponding cotree is completely unresolved, and hence contains no
triple. On the other hand, full resolution of the species tree is
guaranteed if at least one duplication event between any two adjacent
speciations is observable. The achievable resolution therefore depends on
the frequency of gene duplications and the number of gene families.

Despite the variance reduction due to cograph editing, noise in the data,
as well as the occasional introduction of contradictory triples as a
consequence of horizontal gene transfer is unavoidable. The species triples
collected from the individual gene families thus will not always be
congruent.  A conceptually elegant way to deal with such potentially
conflicting information is provided by the theory of supertrees in the form
of the largest set of consistent triples \cite{Jansson:05,GM-13}. The data
will not always contain a sufficient set of duplication events to achieve
full resolution. To this end we consider trees with the property that
  the contraction of any edge leads to the loss of an input triple. There
  may be exponentially many alternative trees of this type. They can be
  listed efficiently using Semple's algorithms \cite{sem:03}. To reduce the
  solution space further we search for a \emph{least resolved tree} in the
  sense of \cite{Jansson:12}, i.e., a tree that has the minimum number of inner
  vertices. It constitutes one of the best estimates of the phylogeny
  without pretending a higher resolution than actually supported by the
  data. In the Supplemental Material we discuss alternative choices.

The mathematical reasoning summarized above, outlined in \emph{Materials
  and Methods}, and presented in full detail in the Supplemental Material,
directly translates into a computational workflow, Fig.\ 1. 
It entails three NP-hard combinatorial optimization problems: cograph
editing \cite{Liu:12}, maximal consistent triple set \cite{Bryant97,
  Wu2004, Jansson2001} and least resolved supertree \cite{Jansson:12}. We
show here that they are nevertheless tractable in practice by formulating
them as Integer Linear Programs (ILP) that can be solved for both,
artificial benchmark data sets and real-life data sets, comprising
genome-scale protein sets for dozens of species, even in the presence of
horizontal gene transfer.


\section{Preliminaries} 

Here, we summarize the definitions and notations required to outline the
mathematical framework, presented in Section \emph{Theory} and \emph{ILP
  Formulation}

\empht{Phylogenetic Trees:} We consider a set $\Gen$ of at least three genes
from a non-empty set $\Spe$ of species. We denote genes by lowercase Roman
and species by lowercase Greek letters. We assume that for each gene its
species of origin is known. This is encoded by the surjective map
$\sigma:\Gen\to \Spe$ with $a \mapsto \sigma(a)$.  A \textit{phylogenetic
  tree (on $L$)} is a rooted tree $T=(V,E)$ with leaf set $L\subseteq V$
such that no inner vertex $v\in V^0:= V\setminus L$ has outdegree one and
whose root $\rho_T\in V$ has indegree zero.  A phylogenetic tree $T$ is
called \emph{binary} if each inner vertex has outdegree two.  A
phylogenetic tree on $\Gen$, resp., on $\Spe$, is called \emph{gene tree},
resp., \emph{species tree}. A (inner) vertex $y$ is an ancestor of $x\in
V$, in symbols $x\prec_T y$ if $y\ne x$ lies on the unique path connecting
$x$ with $\rho_T$.  The \emph{most recent common ancestor} $\lca_T(L')$ of
a subset $L'\subseteq L$ is the unique vertex in $T$ that is the least
upper bound of $L'$ under the partial order $\preceq_T$. We write $L(v):=\{
y\in L| y\preceq_T v\}$ for the set of leaves in the subtree $T(v)$ of $T$
rooted in $v$. Thus, $L(\rho_T)=L$ and $T(\rho_T)=T$.  
\smallskip

\empht{Rooted Triples:} 
Rooted triples \cite{Dress:book}, i.e., rooted binary trees on three leaves,
 are a key concept in the theory of
supertrees \cite{sem-ste-03a,Bininda:book}.  A rooted triple
$r={\rt{(xy|z)}}$ with leaf set $L_r=\{x,y,z\}$ is \emph{displayed} by a
phylogenetic tree $T$ on $L$ if (i) $L_r\subseteq L$ and (ii) the path from
$x$ to $y$ does not intersect the path from $z$ to the root $\rho_T$. Thus
$\lca_T(x,y)\prec_T \lca_T(x,y,z)$. A set $R$ of triples is
\emph{(strictly) dense} on a given leaf set $L$ if for each set of three
distinct leaves there is (exactly) one triple $r\in R$. We denote by
$\mathfrak{R}(T)$ the set of all triples that are displayed by the
phylogenetic tree $T$. A set $R$ of triples is \emph{consistent} if there
is a phylogenetic tree $T$ on $L_R:=\cup_{r\in R} L_r$ such that
$R\subseteq\mathfrak{R}(T)$, i.e., $T$ displays (all triples of) $R$.
If no such tree exists, $R$ is said to be \emph{inconsistent}.\\
Given a triple set $R$, the
polynomial-time algorithm \texttt{BUILD} \cite{Aho:81} either constructs a
phylogenetic tree $T$ displaying $R$ or recognizes that $R$ is
inconsistent. The problem of finding a phylogenetic tree with the smallest
possible number of vertices that is consistent with every rooted triple in
$R$, i.e., a \emph{least resolved} tree, is an NP-hard problem
\citep{Jansson:12}. If $R$ is inconsistent, the problem of determining a
maximum consistent subset of an inconsistent set of triples is
NP-hard and also APX-hard, see \cite{Byrka:10a,vanIersel:09}. Polynomial-time
approximation algorithms for this problem and further theoretical results
are reviewed by \cite{Byrka:10}.

\empht{Triple Closure Operations and Inference Rules:}
If $R$ is consistent it is often possible to infer additional consistent
triples. Denote by $\langle R \rangle$ the set of all phylogenetic trees on
$L_R$ that display $R$. The \emph{closure} of a consistent set of triples
$R$ is $\displaystyle \cl(R) = \cap_{T\in \langle R \rangle}
\mathfrak{R}(T)$, see \cite{BS:95,GSS:07,Bryant97,huber2005recovering,BBDS:00}. 
We say $R$ is \emph{closed} if $R=\cl(R)$ and write $R\vdash \rt{(xy|z)}$ iff
$\rt{(xy|z)}\in \cl(R)$. The closure of a given consistent set $R$ can be
computed in in $O(|R|^5)$ time \cite{BS:95}. 
Extending earlier work of Dekker
\citet{Dekker86}, Bryant and Steel \citet{BS:95} derived conditions under
which $R\vdash \rt{(xy|z)} \implies R'\vdash \rt{(xy|z)}$ for some
$R'\subseteq R$. Of particular importance are the following so-called
\emph{2-order} inference rules:\\[0.1cm]
\hspace*{1cm}
$ \{\rt{(ab|c)}, \rt{(ad|c)}\}\vdash \rt{(bd|c)}$\hfill(i)\\
\hspace*{1cm}
$\{\rt{(ab|c)}, \rt{(ad|b)}\}\vdash \rt{(bd|c)},\rt{(ad|c)}$\hfill(ii)\\
\hspace*{1cm} $ \{\rt{(ab|c)}, \rt{(cd|b)}\}\vdash
\rt{(ab|d)},\rt{(cd|a)}$.
\hfill(iii)\\[0.1cm]
Inference rules based on pairs of triples $r_1, r_2 \in R$ can imply new
triples only if $|L_{r_1}\cap L_{r_2}| = 2$. Hence, in a strictly dense
triple set only the three rules above may lead to new triples.

\empht{Cograph:}
Cographs have a simple characterization as $P_4$-free
graphs, that is, no four vertices induce a simple path, although there
are a number of equivalent characterizations, see \citet{Brandstaedt:99}.
 Cographs can be recognized in linear time \citep{Corneil:85,habib2005simple}. 

\empht{Orthology Relation:}
An empirical orthology relation $\Theta \subset \Gen\times\Gen$ is a
symmetric, irreflexive relation that contains all pairs $(x,y)$ of
orthologous genes. Here, we assume that $x,y\in\Gen$ are \emph{paralogs} if
and only if $x\ne y$ and $(x,y)\notin\Theta$. This amounts to ignoring
horizontal gene transfer.  Orthology detection tools often report some
weight or confidence value $w(x,y)$ for $x$ and $y$ to be orthologs from
which $\Theta$ is estimated using a suitable cutoff.  Importantly, $\Theta$
is symmetric, but not transitive, i.e., it does in general not represent a
partition of $\Gen$.

\empht{Event-Labeled Gene Tree:}
Given $\Theta$ we aim to find a gene tree $T$ with an ``event labeling''
$t:V^0\to\{\bullet,\square\}$ at the inner vertices so that, for any two
distinct genes $x,y\in L$, $t(\lca_{T}(x,y))=\bullet$ if $\lca_{T}(x,y)$
corresponds to a speciation and hence $(x,y)\in\Theta$ and
$t(\lca_{T}(x,y))=\square$ if $\lca_{T}(x,y)$ is a duplication vertex and
hence $(x,y)\notin\Theta$. If such a tree $T$ with event-labeling $t$
exists for $\Theta$, we call the pair $(T,t)$ a \emph{symbolic
  representation} of $\Theta$. We write $(T,t;\sigma)$ if in addition the
species assignment map $\sigma$ is given. A detailed and more general
introduction to the theory of symbolic representations is given in the
Supplemental Material.

\empht{Reconciliation Map:}
A phylogenetic tree $S=(W,F)$ on $\Spe$ is a species tree for a gene tree
$T=(V,E)$ on $\Gen$ if there is a reconciliation map $\mu:V\to W\cup F$
that maps genes $a\in \Gen$ to species $\sigma(a)=\al \in \Spe$ such that
the ancestor relation $\preceq_S$ is implied by the ancestor relation
$\preceq_T$. A more formal definition is given in the Supplemental
Material. Inner vertices of $T$ that map to inner vertices of $S$ are
speciations, while vertices of $T$ that map to edges of $S$ are
duplications.

\section{Theory} 

In this section, we summarize the main ideas and concepts behind our new 
methods that are 
based on our results established in \cite{hernandez2012event, Hellmuth:13d}.
We consider the following problem. 
Given an empirical orthology relation $\Theta$ we want to compute
a species tree. To this end, four independent problems as explained
below have to be solved.

\empht{From Estimated Orthologs to Cographs:}
Empirical estimates of the orthology relation $\Theta$ will in general
contain errors in the form of false-positive orthology assignments, as well
as false negatives, e.g., due to insufficient sequence
similarity. Horizontal gene transfer adds to this noise. Hence an empirical
relation $\Theta$ will in general not have a symbolic representation. In
fact, $\Theta$ has a \emph{symbolic representation} $(T,t)$ if and only if
$G_\Theta$ is a cograph \cite{Hellmuth:13d}, from which $(T,t)$ can be
derived in linear time, see also Theorem \ref{A:thm:ortho-cograph} in the
Supplemental Material.  However, the \emph{cograph editing
  problem}, which aims to convert a given graph $G(V,E)$ into a cograph
$G^*=(V,E^*)$ with the minimal number $|E\symdiff E^*|$ of inserted or
deleted edges, is an NP-hard problem \citep{Liu:11, Liu:12}. Here, the
 symbol $\symdiff$ denotes the symmetric difference of two sets.  
In our setting the problem is considerably simplified  by the structure of
the input data. The gene set of every living organism consists of hundreds
or even thousands of non-homologous gene families. Thus, the initial
estimate of $G_{\Theta}$ already partitions into a large number of connected
components.  As shown in Lemma \ref{A:lem:disconnected} in the Supplemental
Material, it suffices to solve the cograph editing for each connected
component separately.

\empht{Extraction of All Species Triples:}
From this edited cograph $G_{\Theta^*}$, we obtain a unique cotree that, in particular, 
is congruent to an incompletely resolved event-labeled gene-tree $(T,t;\sigma)$.
In \citealp{hernandez2012event}, we investigated the conditions for the
existence of a reconciliation map $\mu$ from the gene tree $T$ to 
the species tree $S$. Given
$(T,t;\sigma)$, consider the triple set $\mathbb{G}$ consisting of all
triples $r=\rt{(ab|c)}\in\mathfrak{R}(T)$ so that (i) all genes
$a,b,c$ belong to different species, and (ii) the event at
the most recent common ancestor of $a,b,c$ is a speciation event,
$t(\lca_T(a,b,c))=\bullet$. From $\mathbb{G}$ and $\sigma$, one can
construct the following set of species triples:
\begin{equation*}
 \mathbb{S}= \left\{ \rt{(\al\be|\ga)}   |\, 
   \exists \rt{(ab|c)}\in\mathbb{G} 
   \textrm{\ with\ } 
   \sigma(a)=\al, \sigma(b)=\be,\sigma(c)=\ga 
   \right\}
\end{equation*}
The main result of \citet{hernandez2012event} establishes that there is a
species tree on $\sigma(\Gen)$ for $(T,t,\sigma)$ if and only if the triple
set $\mathbb{S}$ is consistent. In this case, a reconciliation map can
be found in polynomial time. No reconciliation map exists if $\mathbb{S}$
is inconsistent. 

\empht{Maximal Consistent Triple Set:}
In practice, we cannot expect that the set $\mathbb{S}$ will be consistent. 
Therefore, we have to solve an NP-hard problem, namely, 
computing a maximum consistent subset of triples
$\mathbb{S}^* \subset \mathbb{S}$ \cite{Jansson:12}. 
The following result (see \cite{GM-13} and Supplemental Material) plays a
key role for the ILP formulation of triple consistency.
\begin{thm}\small
  A strictly dense triple set $R$ on $L$ with $|L|\geq 3$ is consistent if
  and only if $\cl(R')\subseteq R$ holds for all $R'\subseteq R$ with
  $|R'|= 2$.
  \label{thm:consistIFFpairwise}  
\end{thm}

\empht{Least Resolved Species Tree:} In order to compute an estimate for
the species tree in practice, we finally compute from $\mathbb{S}^*$ a
  least resolved tree $S$ that minimizes the number of inner vertices.
Hence, we have to solve another NP-hard problem
\cite{Byrka:10a,vanIersel:09}.  However, some instances can be solved in
polynomial time, which can be checked efficiently by utilizing the next
result (see Supplemental Material).
\begin{proposition} \small
  If the tree $T$ inferred from the triple set $R$ by means of 
  \texttt{BUILD} is binary, then the closure $\cl(R)$ is strictly dense.
  Moreover,  $T$ is unique and hence, a least resolved tree
  for $R$.
  \label{pro:BinaryClDense} 
\end{proposition}

\section{ILP Formulation} 

Since we have to solve three intertwined NP-complete optimization problems
we cannot realistically hope for an efficient exact algorithm. We therefore
resort to ILP as the method of choice for solving the problem of computing
a least resolved species tree $S$ from an empirical estimate of the
orthology relation $G_\Theta$. We will use binary variables throughout.
Table 1 
summarizes the definition of the ILP variables and
provides a key to the notation used in this section. In the following we
summarize the ILP formulation.  A detailed description  and proofs for  the
correctness and completeness of the constraints can be
found in the Supplemental Material.

	\begin{table}[t]
\caption{The notation used in the ILP formulation.}
 \begin{tabular}{ll}
   \hline 
    \textbf{Sets \& Constants}   & \textbf{Definition} \\ 
    \hline
    $\Gen$ & Set of genes \\
    $\Spe$ & Set of species \\
    $\Theta_{ab}$ & Genes $a,b\in \Gen$ are estimated orthologs: \\
                  & $\Theta_{ab}=1$ iff  $(a,b)\in \Theta$.\\ 
    \hline
    \textbf{Binary Variables}   & \textbf{Definition}\\ 
    \hline
    $E_{xy}$    & Edge set of the cograph $G_{\Theta^*}=(\Gen,E_{\Theta^*})$ \\
                & of the closest relation $\Theta^*$ to $\Theta$: \\
                & $E_{xy}=1$ iff $\{x,y\}\in E_{\Theta^*}$ 
                (thus, iff $(x,y)\in \Theta^*)$.\\[0.05cm]
    $T_{\rt{(\al\be|\ga)}}$ & Rooted (species) triples in obtained set $\mathbb S$: \\ 
                &              $T_{\rt{(\al\be|\ga)}}=1$ iff 
                 $\rt{(\al\be|\ga)}\in \mathbb S$.\\[0.05cm]			
     $T'_{\rt{(\al\be|\ga)}}$, 
     $T^*_{\rt{(\al\be|\ga)}}$ & Rooted (species) triples 
					in auxiliary strict dense \\ & set $\mathbb S'$, resp., maximal 
                    consistent species   
                   triple \\ & set $\mathbb S^*$: 
                 $T^\bullet_{\rt{(\al\be|\ga)}}=1$
	         iff $\rt{(\al\be|\ga)}\in \mathbb S^\bullet$, $\bullet\in \{\prime,\ast\} $.\\[0.05cm]
    $M_{\al p}$ & Set of clusters: $M_{\al p}=1$ iff $\al\in \Spe$ is 
                  contained \\ 
                & in cluster $p\in\{1,\dots, |\Spe|-2\}$. \\[0.05cm]
    $N_{\al\be,p}$ & Cluster $p$ contains both species $\al$ and $\be$: \\
                &   $N_{\al\be, p}=1$ iff 
                    $M_{\al p}=1$ and $M_{\be p}=1$ \\[0.05cm]
    $C_{p,q,\Gamma\Lambda}$ & Compatibility: $C_{p,q,\Gamma\Lambda}=1$ 
                              iff cluster $p$ and $q$ \\ 
                & have gamete $\Gamma\Lambda\in \{01,10,11\}$.
                  \\[0.05cm]
    $Y_p$       & Non-trivial clusters: $Y_p$=1 iff  cluster 
                  $p\neq \emptyset$. \\ 
    \hline 
\end{tabular}
\label{tab:notation}
\end{table}



\empht{From Estimated Orthologs to Cographs:}  
Our first task is to compute a cograph $G_{\Theta^*}$ that is as similar as
possible to $G_\Theta$ (Eq.\ \eqref{ilp:minDiff} and \eqref{ilp:cog}) with
the additional constraint that no pair of genes within the same species is
connected by an edge, since no pair of orthologs can be found in the same
species (Eq.\ \eqref{ilp:forbid_E}).  Binary variables $E_{xy}$ express
(non)edges in $G_{\Theta^*}$ and binary constants $\Theta_{ab}$ (non)pairs
of the input relation $\Theta$. This ILP formulation 
requires $O(|\Gen|^2)$ binary variables and
$O(|\Gen|^4)$ constraints. In practice, the effort is not dominated by the
number of vertices, since the connected components of $G_{\Theta}$ can be
treated independently.

\begin{ILP}  \label{ilp:minDiff}
\min & \sum_{(x,y)\in\Gen \times \Gen} (1-\Theta_{xy}) E_{xy} + 
       \sum_{(x,y)\in \Gen \times \Gen } \Theta_{xy} (1-E_{xy}) 
\end{ILP}
\begin{ILP}\label{ilp:forbid_E}
  E_{xy}=0 \text{ for all } x,y\in \Gen
	\text{ with }  \sigma(x)=\sigma(y)
\end{ILP}
\begin{ILP}  \label{ilp:cog}
  &E_{wx} + E_{xy}+ E_{yz} - E_{xz} - E_{wy} - E_{wz} \leq 2  \\
  &\forall \text{\ ordered tuples\ } (w,x,y,z) \text{\ of distinct\ } 
	w,x,y,z\in\Gen
\end{ILP}

\empht{Extraction of All Species Triples:}
The construction of the species tree $S$ is based upon the set $\mathbb{S}$
of species triples that can be derived from the set of gene triples
$\mathbb{G}$, as explained in the previous section.  Although the problem
of determining such triples is not NP-hard, we give in the Supplemental
Material an ILP formulation due to the sake of completeness.  However, as
any other approach can be used to determine the species triples we omit
here the ILP formulation, but state that it requires $O(|\Spe|^3)$
variables and $O(|\Gen|^3+|\Spe|^4)$ constraints.

\empht{Maximal Consistent Triple Set:}
An ILP approach to find maximal consistent triple sets was proposed in
\cite{chang2011ilp}. It explicitly builds up a binary tree as a way of
checking consistency. Their approach, however, requires $O(|\Spe|^4)$ ILP
variables, which limits the applicability in practice. By Theorem
\ref{thm:consistIFFpairwise}, strictly a dense triple set $R$ is consistent,
if for all two-element subsets $R'\subseteq R$ the closure $\cl(R')$ is
contained in $R$. This observation allows us to avoid the explicit tree
construction and makes is much easier to find a maximal consistent subset
$\mathbb{S}^*\subseteq \mathbb{S}$. Of course, neither $\mathbb{S}^*$ nor
$\mathbb{S}$ need to be strictly dense. However, since $\mathbb{S}^*$ is
consistent, Lemma \ref{A:lem:binstrictdense} (Supplemental Material)
guarantees that there is a strictly dense triple set $\mathbb{S}'$ containing
$\mathbb{S}^*$. Thus we have $\mathbb{S}^* = \mathbb{S}'\cap \mathbb{S}$,
where $\mathbb{S}'$ must be chosen to maximize $|\mathbb{S}'\cap
\mathbb{S}|$. We define binary variables $T'_{\rt{(\al\be|\ga)}}$,
$T^*_{\rt{(\al\be|\ga)}}$, resp., binary constants $T_{\rt{(\al\be|\ga)}}$
to indicate whether $\rt{(\al\be|\ga)}$ is contained in $\mathbb{S}'$, 
$\mathbb{S}^*$, resp., $\mathbb{S}$.
The ILP formulation that uses  $O(|\Spe|^3)$ variables and $O(|\Spe|^4)$ constraints
 is as follows. 
\begin{ILP} 
 \max \sum_{\rt{(\al\be|\ga)}\in \mathbb S} T'_{\rt{(\al\be|\ga)}} 
 \label{ilp:maxdense}
\end{ILP}
\begin{ILP} 
  \label{ilp:sd}
  &T'_{\rt{(\al\be|\ga)}} + T'_{\rt{(\al\ga|\be)}} + 
   T'_{\rt{(\be\ga|\al)}} = 1 
\end{ILP}
\begin{ILP}
  2 T'_{\rt{(\al\be|\ga)}} + 2&T'_{\rt{(\al\de|\be)}} - 
  T'_{\rt{(\be\de|\ga)}} - T'_{\rt{\rt{(\al\de|\ga)}}} \leq 2  
  \label{ilp:eq:infRule2}
\end{ILP}
\begin{ILP}
  0 \leq T'_{\rt{(\al\be|\ga)}} + T_{\rt{(\al\be|\ga)}} - 
  2T^*_{\rt{(\al\be|\ga)}} \leq 1
  \label{eq:tstar}
\end{ILP}
This ILP formulation can easily be adapted to solve a \emph{``weighted''
 maximum consistent subset} problem: Denote by $w\rt{(\al\be|\ga)}$ the
number of connected components in $G_{\Theta^*}$ that contain three
vertices $a,b,c\in \Gen$ with $\rt{(ab|c)}\in \mathbb G$ and
$\sigma(a)=\al,\sigma(b)=\be, \sigma(c)=\ga$.  These weights can simply be
inserted into the objective function Eq.\ \eqref{ilp:maxdense}

\begin{ILP} 
  \max \sum_{\rt{(\al\be|\ga)}\in \mathbb{S}}
  T'_{\rt{(\al\be|\ga)}}*w\rt{(\al\be|\ga)}
  \label{ilp:wmax}
\end{ILP}
to increase the relative importance of species triples in $\mathbb{S}$, if
they are observed in multiple gene families.

\empht{Least Resolved Species Tree:}
We finally have to find a least resolved species tree from the set $\mathbb
S^*$ computed in the previous step. Thus the variables
$T^*_{\rt{(\al\be|\ga)}}$ become the input constants. For the explicit
construction of the tree we use some of the ideas of \citet{chang2011ilp}.
To build an arbitrary tree for the consistent triple set $\mathbb S^*$, one
can use one of the fast implementations of \texttt{BUILD}
\cite{sem-ste-03a}. If this tree is binary, then Proposition
\ref{pro:BinaryClDense} implies that the closure $\cl(\mathbb S^*)$ is
strictly dense and that this tree is a unique and least resolved tree for
$\mathbb S^*$.  Hence, as a preprocessing step \texttt{BUILD} is used in
advance, to test whether the tree for $\mathbb S^*$ is already binary.  If
not, we proceed with the following ILP approach that 
uses $O(|\Spe|^3)$ variables and constraints.
\begin{ILP}	\label{ilp:minY}
 \min\ &\sum_p Y_p 
\end{ILP}
\begin{ILP}
  0 \leq Y_p|\Spe|  - \sum_{\al\in\Spe} M_{\al p}\leq |\Spe|-1 \label{ilp:yp}
\end{ILP}
\begin{ILP}\label{ilp:Nclus}
  0\leq & M_{\al p} + M_{\be p} - 2 N_{\al\be, p} \leq 1
\end{ILP}
\begin{ILP}
  1 - |\Spe|(1- T^*_{\rt{(\al\be|\ga)}}) \leq 
	\sum_p N_{\al\be,p} - \frac{1}{2}  
        N_{\al\ga,p} -\frac{1}{2} N_{\be\ga,p}
  \label{ilp:rep}
\end{ILP}
\begin{ILP}\label{ilp:CM}
C_{p,q,01}\geq &-M_{\al p}+M_{\al q}\\[-0.1cm]
C_{p,q,10}\geq &\ \ \ \ M_{\al p}-M_{\al q}  \notag\\[-0.1cm]
C_{p,q,11}\geq &\ \ \ \ M_{\al p}+M_{\al q}-1 \notag
\end{ILP}
\begin{ILP}
C_{p,q,01} + C_{p,q,10} + C_{p,q,11} \leq 2 \ \forall p,q \label{ilp:comp}
\end{ILP}
Since a phylogenetic tree $S$ is equivalently specified by its \emph{hierarchy}
$\mathcal{C} = \{L(v)\mid v\in V(S)\}$ whose elements are called \emph{clusters}
(see Supplemental Material or \cite{sem-ste-03a}), 
we construct the clusters induced by
all triples of $\mathbb{S}^*$ and check whether they form a hierarchy on
$\Spe$. Following \citep{chang2011ilp}, we define the binary
$|\Spe|\times(|\Spe|-2)$ matrix $M$, whose entries $M_{\al p}=1$ indicates
that species $\al$ is contained in cluster $p$, see Supplemental
Material. The entries $M_{\al p}$ serve as ILP variables.  In contrast to
the work of \citet{chang2011ilp}, we allow \emph{trivial} columns in $M$ in
which all entries are $0$.  Minimizing the number of \emph{non-trivial}
columns then yields a least resolved tree.

For any two distinct species $\al,\be$ and all clusters $p$ we introduce
binary variables $N_{\al\be, p}$ that indicate whether two species
$\al,\be$ are both contained in the same cluster $p$ or not
(Eq.\ \eqref{ilp:Nclus}).  To determine whether a triple $\rt{(\al\be|\ga)}$
is contained in $\mathbb{S}^* \subseteq \mathbb{S}$ and displayed by a
tree, we need the constraint Eq.\ \eqref{ilp:rep}. Following, the ideas of
Chang et al.\ we use the ``three-gamete condition''. Eq.\ \eqref{ilp:CM} and
\eqref{ilp:comp} ensures that $M$ defines a ``partial'' hierarchy (any two
clusters satisfy $p\cap q\in \{p,q, \emptyset\}$) of compatible clusters.
A detailed discussion how these conditions establish that $M$ encodes a
``partial'' hierarchy can be found in the Supplemental Material.

Our aim is to find a least resolved tree that displays all triples of
$\mathbb{S}^*$. We use the $|\Spe|-2$ binary variables $Y_p=1$ to indicate
whether there are non-zero entries in column $p$ (Eq.\ \eqref{ilp:yp}).
Finally, Eq.\ \eqref{ilp:minY} captures that the number of non-trivial
columns in $M$, and thus the number of inner vertices in the respective
tree, is minimized.  In the Supplemental Material we also discuss an
  ILP formulation to find a tree that displays the minimum number of
  additional triples not contained in $\mathbb S^*$ as an alternative to
  minimizing number of interior vertices.

\section{Implementation and Data Sets} 

Details on implementation and test data sets can be found in the
Supplemental Material. Simulated data were computed with and without
  horizontal gene transfer using both, the method described in \cite{HHW+14}
  and the \texttt{Artificial Life Framework} (\ALF) \cite{Dalquen:12}. As
real-life data sets we used the complete protein complements of 11
\emph{Aquificales} and 19 \emph{Enterobacteriales} species. The initial
orthology relations are estimated with \texttt{Proteinortho}
\cite{Lechner:11a}. The ILP formulation of Fig.\ 1 
is implemented in the Software \texttt{ParaPhylo} using \texttt{IBM ILOG
  CPLEX{\texttrademark}} Optimizer 12.6. \texttt{ParaPhylo} is
freely available from \\
\texttt{http://pacosy.informatik.uni-leipzig.de/paraphylo}.

\section{Results and Discussion}

\begin{figure*}[t]
\begin{center} 
\includegraphics[bb=150 280 470 520, scale=1.3]{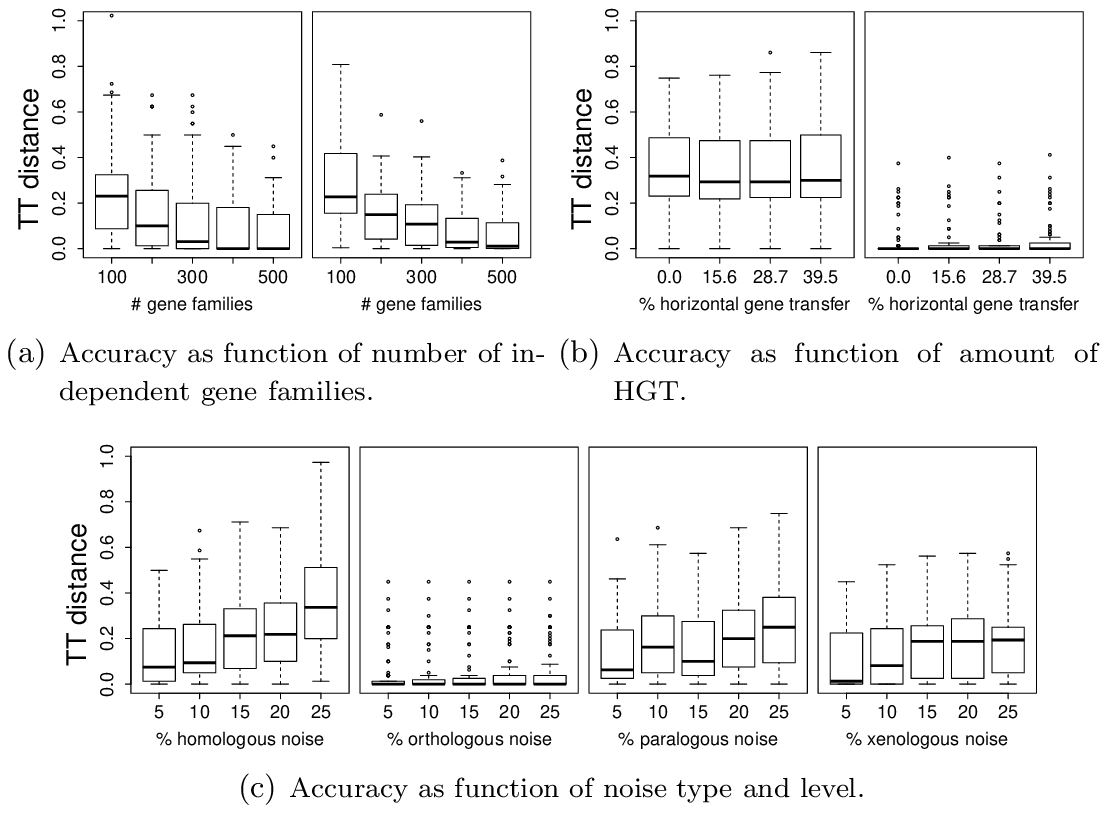} \hspace{1cm} 
\end{center}
\caption{
  Accuracy of reconstructed species trees in simulated data sets.
  \emph{(a)} Dependence on the number of gene families:
  10 (left), and 20 (right) species and 100 to 500 gene families are
  generated using \ALF\ with duplication/loss rate 0.005 and horizontal
  gene transfer rate $0.0$.
  \emph{(b)} Dependence on the intensity of  horizontal gene transfer:
  Orthology estimated with \texttt{Proteinortho} (left), and assuming
  perfect paralogy knowledge (right); $10$ species and $1000$ gene
  families are generated using \ALF\ with duplication/loss rate $0.005$
  and horizontal gene transfer rate ranging from $0.0$ to $0.0075$.
  \emph{(c)} Dependence on the type and intensity ($p=5-25\%$) of noise
  in the raw orthology data $\Theta$:
  $10$ species and 1000 gene families are generated using \ALF\ with
  duplication/loss rate $0.005$ and horizontal gene transfer rate $0.0$.
    Tree distances are measured by the triple metric (TT); all box plots
    summarize 100 independent data sets.
}%
\label{fig:plot}
\end{figure*}

We have shown rigorously that orthology information alone is sufficient to
reconstruct the species tree provided that (i) the orthology is known
without error and unperturbed by horizontal gene transfer, and (ii) 
the input data contains a sufficient number of duplication events. While
this species tree can be inferred in polynomial time for noise-free data,
in a realistic setting three NP-hard optimization problems need to be
solved.

To this end, we use here an exact ILP formulation implementing the workflow
of Fig.\ 1 
to compute species trees from empirically
estimated orthology assignments.  We first use simulated data to
demonstrate that it is indeed feasible in practice to obtain correct gene
trees directly from empirical estimates of orthology.  For 5, 10, 15, and
20 species we obtained prefect, fully resolved reconstructions of 80\%,
56\%, 24\%, and 11\% of the species trees using 500 gene families.  This
comes with no surprise, given the low amount of paralogs in the simulations
(7.5\% to 11.2\%), and the high amount of extremely short branches in the
generated species trees -- on 11.3\% to 17.9\% of the branches, less then
one duplication is expected to occur. Nevertheless, the average TT distance
\emph{was always smaller than 0.09} for more than 300 gene families,
independent from the number of species, Fig.\ 2(a). 
Similar
results for other tree distance measures are compiled in the Supplemental
Material. Thus, deviations from perfect reconstructions are nearly
exclusively explained by a lack of perfect resolution.

In order to evaluate the robustness of the species trees in response to
noise in the input data we used simulated gene families with different
noise models and levels: (i) insertion and deletion of edges in the
orthology graph (homologous noise), (ii) insertion of edges (orthologous
noise), (iii) deletion of edges (paralogous noise), and (iv) modification
of gene/species assignments (xenologous noise). We observe a substantial
dependence of the accuracy of the reconstructed species trees on the noise
model. The results are most resilient against overprediction of orthology
(noise model ii), while missing edges in $\Theta$ have a larger impact, see
Fig.\ 2(c) 
for TT distance, and Supplemental Material for
the other distances. This behavior can be explained by the observation that
many false orthologs (overpredicting orthology) lead to an orthology graph,
whose components are more clique-like and hence, yield few informative
triples. Incorrect species triples thus are reduced, while missing species
triples often can be supplemented through other gene families. On the other
hand, if there are many false paralogs (underpredicting orthology) more
false species triples are introduced, resulting in inaccurate trees.
Xenologous noise (model iv), simulated by changing gene/species
associations with probability $p$ while retaining the original gene tree,
amounts to an extreme model for horizontal transfer. Our model, in
particular in the weighted version, is quite robust for small amounts of
HGT of 5\% to 10\%. Although some incorrect triples are introduced in the
wake of horizontal transfer, they are usually dominated by correct
alternatives observed from multiple gene families, and thus, excluded during
computation of the maximal consistent triple set. Only large scale
concerted horizontal transfer, which may occur in long-term endosymbiotic
associations \cite{Keeling:08}, thus pose a serious problem.

Simulations with \ALF \cite{Dalquen:12} show that our method is resilient
against errors resulting from mis-predicting xenology as orthology, see
Fig.\ 2(b) 
right, even at horizontal gene transfer rates
of $39.5\%$. Assuming perfect paralogy knowledge, i.e., assuming that all
xenologs are mis-predicted as orthologs, the correct trees are
reconstructed essentially independently from the amount of HGT for
$69.75\%$ of the data sets, and the triple distance to the correct tree
remain minute in the remaining cases. This is consistent with noise model
(ii), i.e., a bias towards overpredicting orthology. Tree reconstruction
based directly on the estimated orthology relation computed with
\texttt{Proteinortho} are of course more inaccurate, Fig.\ 
2(b) 
left. Even extreme rates of HGT, however, have no
discernible effect on the quality of the inferred species trees. Our
approach is therefore limited only by quality of initial orthology
prediction tools. 

The fraction $s$ of all triples obtained from the orthology relations
  that are retained in the final tree estimates serves as a quality measure
  similar in flavor e.g.\ to the retention index of
  cladistics. Bootstrapping support values for individual nodes are readily
  computed by resampling either at the level of gene families or at the
  level of triples (see Supplemental Material).

\begin{figure}[t]
\centerline{
	\includegraphics[width=1.\columnwidth]{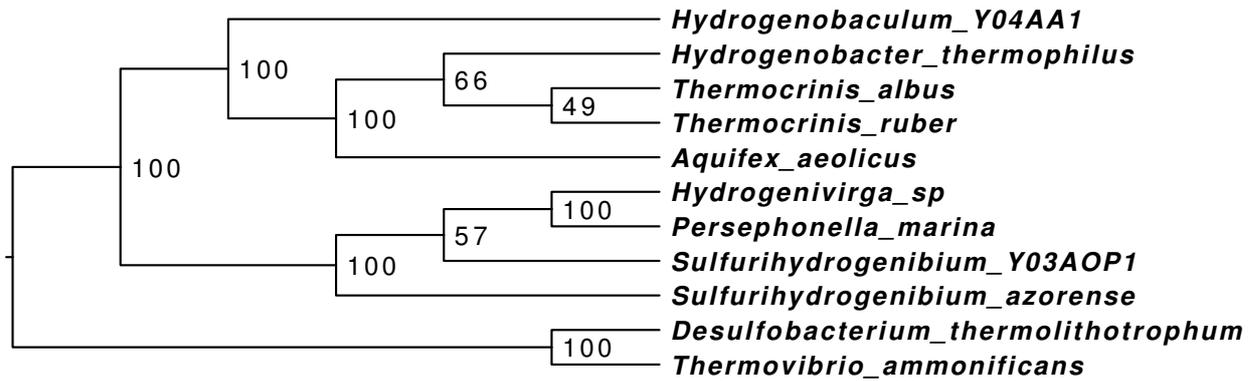}
}
\caption{
  Phylogenetic tree of eleven \emph{Aquificales} species inferred
  from paralogy. Internal node labels indicate triple-based bootstrap
  support.}
\label{fig:simTree}
\end{figure}

For the \emph{Aquificales} data set \texttt{Proteinortho} predicts 2856
gene families, from which 850 contain duplications.  The reconstructed
species tree (see Fig.\ 3
, support $s=0.61$) is almost
identical to the tree presented in \cite{Lechner:14b}.  All species are
clustered correctly according to their taxonomic families.  A slight
difference refers to the two \emph{Sulfurihydrogenibium} species not being
directly clustered.  These two species are very closely related. With only
a few duplicates exclusively found in one of the species, the data was not
sufficient for the approach to resolve this subtree correctly.
Additionally, \emph{Hydrogenivirga sp.} is misplaced next to
\emph{Persephonella marina}.  This does not come as a surprise: Lechner
\emph{et al.} \citet{Lechner:14b} already suspected that the data from this
species was contaminated with material from \emph{Hydrogenothermaceae}.

The second data set comprises the genomes of 19 \emph{Enterobacteriales}
with 8218 gene families of which 15 consists of more than 50 genes and 1342
containing duplications.  Our orthology-based tree shows the expected
groupings of \emph{Escherichia} and \emph{Shigella} species and identifies
the monophyletic groups comprising \emph{Salmonella}, \emph{Klebsiella},
and \emph{Yersinia} species. The topology of the deeper nodes agrees only
in part with the reference tree from \texttt{PATRIC} database
\cite{Wattam:13}, see Supplemental Material for additional information. The
resulting tree has a support of $0.53$, reflecting that a few of the deeper
nodes are poorly supported.

Data sets of around 20 species with a few thousand gene families, each
having up to 50 genes, can be processed in reasonable time,
see Table \ref{tab:runtimeExtended}. 
However, depending on the
amount of noise in the data, the runtime for cograph editing can increase
dramatically even for families with less than 50 genes.

\section{Conclusion}

We have shown here both theoretically and in a practical implementation
that it is possible to access the phylogenetic information implicitly
contained in gene duplications and thus to reconstruct a species phylogeny
from information of paralogy only.  This source of information is strictly
complementary to the sources of information employed in phylogenomics
studies, which are always based on alignments of orthologous sequences. In
fact, 1:1 orthologs -- the preferred data in sequence-based phylogenetics
-- correspond to cographs that are complete and hence have a star as their
cotree and therefore do not contribute \emph{at all} to the phylogenetic
reconstruction in our approach. Access to the phylogenetic information
implicit in (co-)orthology data requires the solution of three NP-complete
combinatorial optimization problems. This is generally the case in
  phylogenetics, however: both the multiple sequence alignment problem and
  the extraction of maximum parsimony, maximum likelihood, or optimal
  Bayesian trees is NP-complete as well. Here we solve the
   computational tasks exactly for moderate-size problems by means of an
ILP formulation.  Using phylogenomic data for \emph{Aquificales} and
\emph{Enterobacteriales} we demonstrated that non-trivial phylogenies can
indeed be re-constructed from tree-free orthology estimates
alone. Just as sequence-based approaches in molecular phylogeny
  crucially depend on the quality of multiple sequence alignments, our
  approach is sensitive to the initial estimate $\Theta$ of the orthology
  relation. Horizontal gene transfer, furthermore, is currently not
  included in the model but rather treated as noise that disturbs the
  phylogenetic signal.  Simulated data indicate that the method is rather
robust and can tolerate surprisingly large levels of noise in the form
  of both, mis-predicted orthology and horizontal gene transfer, provided
 a sufficient number of independent gene families is available as
  input data. Importantly, horizontal gene-transfer can introduce a
bias only when many gene families are simultaneously affected by horizontal
transfer.  Lack of duplications, on the other hand, limits our resolution
at very short time scales, a regime in which sequence-based approaches work
very accurately.

We have used here an exact implementation as ILP to demonstrate the
  potential of the approach without confounding it with computational
  approximations.  Thus, the current implementation does not easily scale
to very large data sets.
Paralleling
the developments in sequence-based phylogenetics, where the
  NP-complete problems of finding a good input alignment and of
  constructing tree(s) maximizing the parsimony score, likelihood, or
  Bayesian posterior probability also cannot be solved exactly for large
  data sets, it will be necessary in practice to settle for heuristic
solutions. In sequence-based phylogenetics, these have improved over
  decades to the point where they are no longer a limiting factor in
  phylogenetic reconstruction.  Several polynomial time heuristics and
approximation algorithms have been devised \emph{already} for the triple
consistency problem \cite{Gasieniec:99,Maemura:07,Byrka:10a,Tazehkand:13}.
The cograph editing problem and the least resolved tree problem, in
contrast, have received comparably little attention so far, but constitute
the most obvious avenues for boosting computational efficiency.  Empirical
observations such as the resilience of our approach against overprediction
of orthologs in the input will certainly be helpful in designing efficient
heuristics.

In the long run, we envision that the species tree $S$, and the symbolic
representation of the event-annotated gene tree $(T,t)$ may serve as
constraints for a refinement of the initial estimate of $\Theta$, solely
making use only of (nearly) unambiguously identified branchings and event
assignments. A series of iterative improvements of estimates for $\Theta$,
$(T,t)$, and $S$, and, more importantly, methods that allow to
  accurately detect paralogs, may not only lead to more accurate trees and
orthology assignments, but could also turn out to be computationally more
efficient.

\section*{APPENDIX}

\section{Theory} 

In this section we give an expanded and more technical account of the
mathematical theory underlying the relationships between orthology
relations, triple sets, and the reconciliation of gene and triple sets. In
particular, we include here the proofs of the key novel results outline in
the main text. The notation in the main text is a subset of the one used
here. Theorems, remarks, and ILP formulations have the same numbers as in
the main text. As a consequence, the numberings in this supplement may not
always be in ascending order.

\subsection{Notation}

For an arbitrary set $X$ we denote with $\binom{X}{n}$ the set of
$n$-elementary subsets of X.  In the remainder of this paper, $L$ will
always denote a finite set of size at least three.  Furthermore, we will
denote with $\Gen$ a set of genes and with $\Spe$ a set of species and
assume that $|\Gen|\ge 3$ and $|\Spe|\ge 1$. Genes contained in $\Gen$ are
denoted by lowercase Roman letters $a,b,c,\ldots$ and species in $\Spe$ by
lower case Greek letters $\al,\be,\ga\ldots$. Furthermore, let
$\sigma:\Gen\to \Spe$ with $x \mapsto \sigma(x)$ be a mapping that assigns
to each gene $x\in \Gen$ its corresponding species $\sigma(x)=\chi\in
\Spe$. With $\sigma(\Gen)$ we denote the image of $\sigma$. W.l.o.g. we can
assume that the map $\sigma$ is surjective, and thus,
$\sigma(\Gen)=\Spe$. We assume that the reader is familiar with graphs and
its terminology, and refer to \citet{Diestel12} as a standard reference.

\subsection{Phylogenetic Trees} 

A tree $T=(V,E)$ is a connected cycle-free graph with vertex set $V(T)=V$
and edge set $E(T)=E$. A vertex of $T$ of degree one is called a
\emph{leaf} of $T$ and all other vertices of $T$ are called \emph{inner}
vertices. An edge of $T$ is an \emph{inner} edge if both of its end
vertices are inner vertices. The sets of inner vertices of $T$ is denoted
by $V^0$.  A tree $T$ is called \emph{binary} if each inner vertex has
outdegree two. A \emph{rooted tree} $T=(V,E)$ is a tree that contains a
distinguished vertex $\rho_T\in V$ called the \emph{root}.

A \emph{phylogenetic tree $T$ (on $L$)} is a rooted tree $T=(V,E)$
with leaf set $L\subseteq V$ such that no inner vertex has in- and
outdegree one and whose root $\rho_T\in V$ has indegree zero.
A phylogenetic tree  on $\Gen$, resp., on $\Spe$, 
is called \emph{gene tree}, resp., \emph{species tree}.

Let $T=(V,E)$ be a phylogenetic tree on $L$ with root $\rho_T$. The
ancestor relation $\preceq_T$ on $V$ is the partial order defined, for
all $x,y\in V$, by $x \preceq_T y$ whenever $y$ lies on the (unique)
path from $x$ to the root. Furthermore, we write $x \prec_T y$ if $x
\preceq_T y$ and $x\ne y$. For a non-empty subset of leaves
$L'\subseteq L$, we define $\lca_T(L')$, or the \emph{most recent
common ancestor of $L'$}, to be the unique vertex in $T$ that is the
least upper bound of $L'$ under the partial order $\preceq_T$. In
case $L'=\{x,y \}$, we put $\lca_T(x,y):=\lca_T(\{x,y\})$ and if
$L'=\{x,y,z \}$, we put $\lca_T(x,y,z):=\lca_T(\{x,y,z\})$. If there
is no danger of ambiguity, we will write $\lca(L')$ rather then	
$\lca_T(L')$. 

For $v\in V$, we denote with $L(v):=\{ y\in L| y\preceq_T v\}$ the
set of leaves in the subtree $T(v)$ of $T$ rooted in $v$. Thus,
$L(\rho_T)=L$ and $T(\rho_T)=T$. 

It is well-known that there is a one-to-one correspondence between
(isomorphism classes of) phylogenetic trees on $L$ and so-called
hierarchies on $L$.  For a finite set $L$, a \emph{hierarchy on $L$} is a
subset $\mathcal C$ of the power set $\mathbb P(L)$ such that 

\begin{enumerate}
\item[(i)] $L\in \mathcal{C}$
\item[(ii)]  $\{x\}\in \mathcal{C}$ for all $x\in L$ and
\item[(iii)] $p\cap q\in \{p, q, \emptyset\}$ for all $p, q\in
  \mathcal{C}$.
\end{enumerate}

The elements of $\mathcal{C}$ are called clusters. 
\ \\ \begin{thm}[\citealp{sem-ste-03a}]
  Let $\mc C$ be a collection of non-empty subsets of $L$.  Then, there is
  a phylogenetic tree $T$ on $L$ with $\mc C = \{L(v)\mid v\in V(T)\}$ if
  and only if $\mc C$ is a hierarchy on $L$.
  \label{A:thm:hierarchy}
\end{thm}\ \\

The following result appears to be well known. We include a simple proof
since we were unable to find a reference for it.
\ \\ \begin{lem}
	The number of clusters $|\mc C|$ in a hierarchy $\mc C$ on $L$
	determined by a phylogenetic tree $T=(V,E)$ on $L$ is bounded
	by $2 |L|-1$. 
	\label{A:lem:nrC}
\end{lem}
\begin{proof} 
  Clearly, the number of clusters $|\mathcal{C}|$ is determined by the number 
  of vertices $|V|$, since each leaf $v\in L$, 
  determines the singleton cluster $\{v\}\in \mathcal{C}$ 
  and each inner node $v$ has at least two children and thus, 
  gives rise to a new cluster $L(v) \in \mathcal{C}$. Hence, 
  $|\mathcal{C}| = |V|$.

  First, consider a binary phylogenetic tree $T=(V,E)$ on $|L|$
  leaves. Then there are $|V|-|L|$ inner vertices, all of out-degree
  two. Hence, $|E| = 2(|V|-|L|) =|V|-1$ and thus $|V|=2|L|-1$. Hence, $T$
  determines $|\mc C| =2|L|-1$ clusters and has in particular $|L|-1$ inner
  vertices.

  Now, its easy to verify by induction on the number of leaves $|L|$ that
  an arbitrary phylogenetic tree $T'=(V',E')$ has $n_0\leq |L|-1$ inner
  vertices and thus, $ |\mathcal{C}'|=|V'| = n_0 +|L| \leq 2|L|-1$
  clusters.  
\end{proof}

\subsection{Rooted Triples}

\subsubsection{Consistent Triple Sets} 

Rooted triples, sometimes also called rooted triplets \citep{Dress:book},
constitute an important concept in the context of supertree reconstruction
\citep{sem-ste-03a,Bininda:book} and will also play a major role here. A
rooted triple $r={\rt{(xy|z)}}$ is \emph{displayed} by a phylogenetic tree
$T$ on $L$ if $x,y,z\in L$ pairwise distinct, and the path from $x$ to $y$ does not intersect
the path from $z$ to the root $\rho_T$ and thus, having $\lca_T(x,y)\prec_T
\lca_T(x,y,z)$.  We denote with $L_r$ the set of the three leaves
$\{x,y,z\}$ contained in the triple $r=\rt{(xy|z)}$, and with
$L_R:=\cup_{r\in R} L_r$ the union of the leaf set of each $r\in R$. For a
given leaf set $L$, a triple set $R$ is said to be \emph{(strict) dense}
if for each $x,y,z\in L$ there is (exactly) one triple $r\in R$ with
$L_r=\{x,y,z\}$.  For a phylogenetic tree $T$, we denote by
$\mathfrak{R}(T)$ the set of all triples that are displayed by $T$. A set
$R$ of triples is \emph{consistent} if there is a phylogenetic tree $T$ on
$L_R$ such that $R\subseteq\mathfrak{R}(T)$, i.e., $T$ displays all triples
$r\in R$.

Not all sets of triples are consistent, of course. Given a triple set $R$
there is a polynomial-time algorithm, referred to in \citep{sem-ste-03a} as
\texttt{BUILD}, that either constructs a phylogenetic tree $T$ displaying
$R$ or recognizes that $R$ is not consistent or \emph{inconsistent}
\citep{Aho:81}. Various practical implementations have been described
starting with \citep{Aho:81}, improved variants are discussed in
\citep{Henzinger:99,Jansson:05}.  The problem of determining a maximum
consistent subset of an inconsistent set of triples, however, is NP-hard
and also APX-hard, see \citep{Byrka:10a,vanIersel:09} and the references
therein. We refer to \citep{Byrka:10} for an overview on the available
practical approaches and further theoretical results.

 For a given consistent triple set $R$, a rooted phylogenetic tree, in
  which the contraction of any edge leads to the loss of an input triple is
  called a \emph{least resolved} tree (for $R$).  Following the idea of
  Janson et al.\ \cite{Jansson:12}, we are mainly concerned with the
  construction of least resolved trees, that have in addition as few inner
  vertices as possible and cover the largest subset of compatible triples
  contained in $R$.  Finding a tree with a minimal number of inner nodes
for a given consistent set of rooted triples is also an NP-hard problem,
see \cite{Jansson:12}. Unless stated explicitly, we use the term
  \emph{least resolved tree} to refer to a tree with a minimum number of
  interior vertices, i.e., the least resolved trees in the sense of
  \cite{Jansson:12}. Alternative choices include the trees that display the
  minimum number of additional triples not contained in $R$.

\subsubsection{Graph Representation of Triples}

There is a quite useful representation of a set of triples $R$ as a graph
also known as \emph{Aho graph}, see \citep{Aho:81, huson2010phylogenetic,
  BS:95}. For given a triple set $R$ and an arbitrary subset $\S\subseteq
L_R$, the graph $[R,\S]$ has vertex set $\S$ and two vertices $x,y\in \S$
are linked by an edge, if there is a triple $\rt{(xy|z)} \in R$ with $z\in
\S$.  Based on connectedness properties of the graph $[R,\S]$ for
particular subsets $\S\subseteq L_R$, the algorithm \texttt{BUILD}
recognizes if $R$ is consistent or not. In particular, this algorithm makes
use of the following well-known theorem.
\bigskip

\begin{thm}[\citealp{Aho:81,BS:95}]
A set of rooted triples $R$ is consistent if and only if for each
subset $\S\subseteq L_R$, $|\S|>1$ the graph $[R,\S]$ is disconnected. 
\label{A:thm:ahograph}
\end{thm}
\bigskip

\begin{lem}[\citealp{huson2010phylogenetic}]
  Let $R$ be a dense set of rooted triples on $L$.  Then for each
  $\S\subseteq L$, the number of connected components of the Aho graph
  $[R,\S]$ is at most two.
\label{A:lem:dense-binary}
\end{lem}\bigskip

The tree computed with \texttt{BUILD} based on the Aho graph for a
  consistent set of rooted triples $R$ is denoted by $\mathrm{Aho}(R)$. 
Lemma~\ref{A:lem:dense-binary} implies that $\mathrm{Aho}(R)$ must be
binary for a consistent dense set of rooted triples.  We will use the Aho
graph and its key properties as a frequent tool in upcoming proofs.

For later reference, we recall \smallskip

\begin{lem}[\citealp{BS:95}]
  If $R'$ is a subset of the triple set $R$ and $L$ is
  a leaf set, then $[R',L]$ is a subgraph of $[R,L]$. 
  \label{A:lem:subgraph}
\end{lem}

\subsubsection{Closure Operations and Inference Rules}

The requirement that a set $R$ of triples is consistent, and thus, that
there is a tree displaying all triples, allows to infer new triples from
the set of all trees displaying all triples of $R$ and to define a
\emph{closure operation} for $R$, which has been extensively studied in the
last decades, see \citep{BS:95, GSS:07,
  Bryant97,huber2005recovering,BBDS:00}. Let $\langle R \rangle$ be the set
of all phylogenetic trees on $L_R$ that display all the triples of $R$.
The closure of a consistent set of rooted triples $R$ is defined as
$$\cl(R) = \bigcap_{T\in \langle R \rangle} \mathfrak{R}(T).$$ 
This operation satisfies the usual three properties of a closure operator,
namely: $R \subseteq \cl(R)$; $\cl(\cl(R))=\cl(R)$ and if $R' \subseteq R$,
then $\cl(R')\subseteq \cl(R)$. We say $R$ is \emph{closed} if $R=\cl(R)$.
Clearly, for any tree $T$ it holds that $\mathfrak{R}(T)$ is closed.  The
brute force computation of the closure of a given consistent set $R$ runs
in $O(|R|^5)$ time \citep{BS:95}: For any three leaves $x,y,z\in L_R$ test
whether exactly one of the sets $R\cup\{\rt{(xy|z)}\}$,
$R\cup\{\rt{(xz|y)}\}$, $R\cup\{\rt{(zy|x)}\}$ is consistent, and if so,
add the respective triple to the closure $\cl(R)$ of $R$.

For a consistent set $R$ of rooted triples we write $R\vdash \rt{(xy|z)}$
if any phylogenetic tree that displays all triples of $R$ also displays
$\rt{(xy|z)}$. In other words, $R\vdash \rt{(xy|z)}$ iff $\rt{(xy|z)}\in
\cl(R)$. In a work of Bryant and Steel \citep{BS:95}, in which the authors
extend and generalize the work of Dekker \citep{Dekker86}, it was shown
under which conditions it is possible to infer triples by using only
subsets $R'\subseteq R$, i.e., under which conditions $R\vdash \rt{(xy|z)}
\implies R'\vdash \rt{(xy|z)}$ for some $R'\subseteq R$. In particular, we
will make frequent use of the following inference rules:
\renewcommand{\theequation}{\roman{equation}}
\begin{align}
  \{\rt{(ab|c)}, \rt{(ad|c)}\} &\vdash \rt{(bd|c)} \label{eq:infRule1} \\
  \{\rt{(ab|c)}, \rt{(ad|b)}\} & \vdash \rt{(bd|c)},\rt{(ad|c)} \label{eq:infRule2} \\
  \{\rt{(ab|c)}, \rt{(cd|b)}\} &\vdash \rt{(ab|d)},\rt{(cd|a)}.\label{eq:infRule3}
\end{align}
\renewcommand{\theequation}{\arabic{equation}}

\begin{rem} 
  It is an easy task to verify, 
  that such inference rules based on two triples $r_1, r_2
  \in R$ can lead only to new triples, whenever 
  $|L_{r_1}\cap L_{r_2}| = 2$. Hence, the latter three
  stated rules are the only ones that lead to new triples for a given
  pair of triples in a strictly dense triple set. 
  \label{A:rem:only}
\end{rem}\bigskip

For later reference and the ILP formulation, we give the following lemma.
\bigskip

\begin{lem}
	\label{lem:suffRule}
	Let $R$ be a strictly dense set of rooted triples. 
	For all  $L'=\{a,b,c,d\} \subseteq L_R$ we have the following statements: 
	
	All triples inferred by rule \eqref{eq:infRule2} applied on triples $r\in R$ with $L_r\subset L'$ are contained in $R$ 
	if and only if all triples inferred by rule \eqref{eq:infRule3} applied on triples $r\in R$ with $L_r\subset L'$
	are contained in $R$.  
	
	 Moreover, if all triples inferred by rule \eqref{eq:infRule2} applied on triples $r\in R$ with $L_r\subset L'$ are contained in $R$  
	 then 	all triples inferred by rule \eqref{eq:infRule1} applied on triples $r\in R$ with $L_r\subset L'$ are contained in $R$.  
\end{lem}
\begin{proof}
	The first statement was established in \citep[Lemma 2]{GM-13}. 
	
	For the second statement assume that for all pairwise distinct
	 $L'=\{a,b,c,d\} \subseteq L_R$ it holds that all triples inferred by rule 
	\eqref{eq:infRule2}, or equivalently, by rule \eqref{eq:infRule3}
	applied on triples $r\in R$ with $L_r\subset L'$ are contained in $R$.
	Assume for contradiction that there are triples 
	$\rt{(ab|c)}, \rt{(ad|c)} \in R$, but $\rt{(bd|c)}\not\in R$.
	Since $R$ is strictly dense, we have either $\rt{(bc|d)}\in R$
	or $\rt{(cd|b)}\in R$. In the first case and since $\rt{(ab|c)}\in R$,
	 rule \eqref{eq:infRule2} implies that $\rt{(ac|d)}\in R$, a contradiction. 
	In the second case and since $\rt{(ab|c)}\in R$,
	rule \eqref{eq:infRule3} implies that $\rt{(cd|a)}\in R$, a contradiction. 
\end{proof}

We are now in the position to prove the following important and helpful
lemmas and theorem. The final theorem basically states that consistent strict
dense triple sets can be characterized by the closure of any two element
subset of $R$. Note, an analogous result was established by \cite{GM-13}.
However, we give here an additional direct and transparent proof.
\bigskip 

\begin{lem}
  Let $R$ be a strictly dense set of triples on $L$ such that for all
  $R'\subseteq R$ with $|R'|= 2$ it holds $\cl(R')\subseteq R$. Let
  $x\in L$ and $L'=L\setminus \{x\}$. Moreover, let $R_{|L'}\subset R$
  denote the subset of all triples $r\in R$ with $L_r\subseteq L'$. Then
  $R_{|L'}$ is strictly dense and for all $R'\subseteq R_{|L'}$ with
  $|R'|= 2$ it holds $\cl(R')\subseteq R_{|L'}$. 
  \label{A:lem:rest}
\end{lem}
\begin{proof}
  Clearly, since $R$ is strictly dense and since $R_{|L'}$ contains all
  triples except the ones containing $x$ it still holds that for all
  $a,b,c\in L'$ there is exactly one triple $r\in R_{|L'}$ with
  $a,b,c\in L_r$. Hence, $R_{|L'}$ is strictly dense. 
  
  Assume for contradiction, that there are triples	 
  $r_1, r_2\in R_{|L'}\subset R$ with $\cl(r_1,r_2)\not\subseteq R_{|L'}$. 
  By construction of  $R_{|L'}$, no triples $r_1, r_2\in R_{|L'}$ 
  can infer a new triple $r_3$ with $x\in L_{r_3}$.
  This immediately implies that $\cl(r_1,r_2)\not\subseteq R$, 
  a contradiction.  
\end{proof}

\begin{lem}
  Let $R$ be a strictly dense set of triples on $L$ with $|L|=4$. If
  for all $R'\subseteq R$ with $|R'|= 2$ holds $\cl(R')\subseteq R$
  then $R$ is consistent. 
  \label{A:lem:L=4}
\end{lem}
\begin{proof}
  By contraposition, assume that $R$ is not consistent. Thus, the Aho graph
  $[R,\S]$ is connected for some $\S\subseteq L$. Since $R$ is strict
  dense, for any $\S\subseteq L$ with $|\S|=2$ or $|\S|=3$ the Aho graph
  $[R,\S]$ is always disconnected.  Hence, $[R,\S]$ for $\S=L$ must be
  connected.  The graph $[R,L]$ has four vertices, say $a,\ b,\ c$ and $d$.
  The fact that $R$ is strictly dense and $|L|=4$ implies that $|R|=4$ and in
  particular, that $[R,L]$ has three or four edges. Hence, the graph
  $[R,L]$ is isomorphic to one of the following graphs $G_0$, $G_1$ or $G_2$.

  The graph $G_0$ is isomorphic to a path $x_1-x_2-x_3-x_4$ on four
  vertices; $G_1$ is isomorphic to a chordless square; and $G_2$ is
  isomorphic to a path $x_1-x_2-x_3-x_4$ on four vertices where the edge
  $\{x_1,x_3\}$ or $\{x_2,x_4\}$ is added. W.l.o.g. assume that for the
  first case $[R,L]\simeq G_0$ has edges $\{a,b\}$, $\{b,c\}$, $\{c,d\}$;
  for the second case $[R,L]\simeq G_1$ has edges $\{a,b\}$, $\{a,c\}$,
  $\{c,d\}$ and $\{b,d\}$ and for the third case assume that $[R,L]\simeq
  G_2$ has edges $\{a,b\}$, $\{a,c\}$, $\{c,d\}$ and $\{a,d\}$.
  
  Let $[R,L]\simeq G_0$. Then there are triples of the form $\rt{(ab|*)}$,
  $\rt{(bc|*)}$, $\rt{(cd|*)}$, where one kind of triple must occur twice,
  since otherwise, $[R,L]$ would have four edges. Assume that this is
  $\rt{(ab|*)}$.  Hence, the triples $\rt{(ab|c)}, \rt{(ab|d)}\in R$ since
  $|R|=4$.  Since $R$ is strictly dense, $\rt{(bc|*)}=\rt{(bc|d)}\in R$,
  which implies that $\rt{(cd|*)} = \rt{(cd|a)}\in R$. Now,
  $R'=\{\rt{(ab|c)},\rt{(bc|d)} \} \vdash \rt{(ac|d)}$.  However, since $R$
  is strictly dense and $\rt{(cd|a)}\in R$ we can conclude that
  $\rt{(ac|d)}\not\in R$, and therefore $ \cl(R')\not\subseteq R$.  The
  case with triples $\rt{(cd|*)}$ occurring twice is treated analogously.
  If triples $\rt{(bc|*)}$ occur twice, we can argue the same way to obtain
  obtain $\rt{(bc|a)}, \rt{(bc|d)}\in R$, $\rt{(ab|*)} = \rt{(ab|d)}$, and
  $\rt{(cd|*)} = \rt{(cd|a)}$. However,
  $R'=\{\rt{(bc|a)},\rt{(cd|a)}\}\vdash \rt{(bd|a)} \notin R$, and thus
  $\cl(R')\not\subseteq R$.
  	
  Let $[R,L]\simeq G_1$. Then there must be triples of the form
  $\rt{(ab|*)}$, $\rt{(ac|*)}$, $\rt{(cd|*)}$, $\rt{(bd|*)}$. Clearly,
  $\rt{(ab|*)}\in \{\rt{(ab|c)}, \rt{(ab|d)}\}$. Note that not both
  $\rt{(ab|c)}$ and $\rt{(ab|d)}$ can be contained in $R$, since then
  $[R,L]\simeq G_0$.  If $\rt{(ab|*)}=\rt{(ab|c)}$ and since $R$ is strict
  dense, $\rt{(ac|*)} = \rt{(ac|d)}$. Again, since $R$ is strictly dense,
  $\rt{(cd|*)} = \rt{(cd|b)}$ and this implies that $\rt{(bd|*)} =
  \rt{(bd|a)}$. However, $R'=\{\rt{(ab|c)},\rt{(ac|d)}\}\vdash \rt{(ab|d)}
  \notin R$, since $R$ is strictly dense and $\rt{(bd|a)}\in R$. Thus,
  $\cl(R')\not\subseteq R$.  If $\rt{(ab|*)}=\rt{(ab|d)}$ and since $R$ is
  strictly dense, we can argue analogously, and obtain, $\rt{(bd|*)} =
  \rt{(bd|c)}$, $\rt{(cd|*)} = \rt{(cd|a)}$ and $\rt{(ac|*)} =
  \rt{(ac|b)}$. However, $R'=\{\rt{(ab|d)},\rt{(bd|c)}\}\vdash \rt{(ad|c)}
  \notin R$, and thus $\cl(R')\not\subseteq R$.

  Let $[R,L]\simeq G_2$. Then there must be triples of the form
  $\rt{(ab|*)}$, $\rt{(ac|*)}$, $\rt{(cd|*)}$, $\rt{(ad|*)}$. Again,
  $\rt{(ab|*)}\in \{\rt{(ab|c)}, \rt{(ab|d)}\}$. By similar arguments as in
  the latter two cases, if $\rt{(ab|*)} = \rt{(ab|c)}$ then we obtain,
  $\rt{(ac|*)} = \rt{(ac|d)}$, $\rt{(ad|*)} = \rt{(ad|b)}$ and $\rt{(cd|*)}
  = \rt{(cd|b)}$. Since $R'=\{\rt{(ab|c)},\rt{(ac|d)}\}\vdash \rt{(bc|d)}
  \notin R$, we can conclude that $\cl(R')\not\subseteq R$.  If
  $\rt{(ab|*)}=\rt{(ab|d)}$ we obtain analogously, $\rt{(ad|*)} =
  \rt{(ad|c)}$, $\rt{(cd|*)} = \rt{(cd|b)}$ and $\rt{(ac|*)} =
  \rt{(ac|b)}$. However, $R'=\{\rt{(ab|d)},\rt{(ad|c)}\}\vdash \rt{(bd|c)}
  \notin R$, and thus $\cl(R')\not\subseteq R$. 
\end{proof}

\begin{cthm}{\ref{thm:consistIFFpairwise}}
  Let $R$ be a strictly dense triple set on $L$ with $|L|\geq 3$. The
  set $R$ is consistent if and only if $\cl(R')\subseteq R$ holds
  for all $R'\subseteq R$ with
  $|R'|= 2$.
\end{cthm}
\begin{proof}
  $\Rightarrow:$
  If $R$ is strictly dense and consistent, then for any triple triple
  $\rt{(ab|c)} \notin R$ holds $R\cup \rt{(ab|c)}$ is inconsistent as either
  $\rt{(ac|b)}$ or $\rt{(bc|a)}$ is already contained in $R$. Hence, for each
  $a,b,c\in L$ exactly one $R\cup\{\rt{(ab|c)}\}$, $R\cup\{\rt{(ac|b)}\}$,
  $R\cup\{\rt{(bc|a)}\}$ is consistent, and this triple is already
  contained in $R$. Hence, $R$ is closed. Therefore, for any subset
  $R'\subseteq R$ holds $\cl(R')\subseteq \cl(R)=R$. In particular, this
  holds for all $R'\subseteq R$ with $|R'|= 2$. 
  
  $\Leftarrow:$ \emph{(Induction on $|L|$.)}\\
  If $|L|=3$ and since $R$ is strictly dense, it holds $|R|=1$ and thus,
  $R$ is always consistent. If $|L|=4$, then Lemma \ref{A:lem:L=4}
  implies that if for any two-element subset
  $R'\subseteq R$ holds that $\cl(R')\subseteq R$, then $R$ is
  consistent. Assume therefore, the assumption is true for 
  all strictly dense triple sets $R$ on $L$ with $|L|=n$.
  
  Let $R$ be a strictly dense triple set on $L$ with $|L|=n+1$ such that
  for each $R'\subseteq R$ with $|R'|= 2$ it holds $\cl(R')\subseteq
  R$. Moreover, let $L'=L\setminus \{x\}$ for some $x\in L$ and
  $R_{|L'}\subset R$ denote the subset of all triples $r\in R$ with
  $L_r\subset L'$. Lemma \ref{A:lem:rest} implies that $R_{|L'}$ is
  strictly dense and for each $R'\subseteq R_{|L'}$ with $|R'|= 2$ we
  have $\cl(R')\subseteq R_{|L'}$. Hence, the induction hypothesis can
  be applied for any such $R_{|L'}$ implying that $R_{|L'}$ is
  consistent. Moreover, since $R_{|L'}$ is strictly dense and
  consistent, for any triple $\rt{(xy|z)}\notin R_{|L'}$ holds that
  $R_{|L'} \cup \rt{(xy|z)}$ is inconsistent. But this implies that
  $R_{|L'}$ is closed, i.e., $\cl(R_{|L'})=R_{|L'}$. Lemma
  \ref{A:lem:dense-binary} implies that the Aho graph $[R_{|L'},\S]$ has
  exactly two connected components $C_1$ and $C_2$ for each
  $\S\subseteq L'$ with $|\S|>1$. In the following we denote with
  $\S_i=V(C_i)$, $i=1,2$ the set of vertices of the connected
  component $C_i$ in $[R_{|L'},\S]$. Clearly, $\S=\S_1 \dot\cup \S_2$.
  It is easy to see that $[R,\S] \simeq [R_{|L'},\S]$ 
  for any $\S\subseteq L'$, since none of the graphs contain
  vertex $x$. 	Hence, $[R,\S]$ is always disconnected 
  for any $\S\subseteq L'$.	
  Therefore, it remains to show that, for all $\S\cup\{x\}$ with
  $\S\subseteq L'$ holds: if for any $R'\subseteq R$ with $|R'|= 2$ holds
  $\cl(R')\subseteq R$, then $[R, \S\cup \{x\}]$ is disconnected and
  hence, $R$ is consistent. 
  
  To proof this statement we consider the different possibilities for $\S$
  separately. We will frequently use that $[R_{|L'},\S]$ is a subgraph of
  $[R,\S]$ for every $\S\subseteq L$ (Lemma \ref{A:lem:subgraph}).
  
  \emph{Case 1.} If $|\S|=1$, then $\S\cup \{x\}$ implies that $[R, \S\cup
  \{x\}]$ has exactly two vertices and clearly, no edge.  Thus, $[R, \S\cup
  \{x\}]$ is disconnected.
  
  \emph{Case 2.} Let $|\S|=2$ with $\S_1=\{a\}$ and $\S_2=\{b\}$.  Since
  $R$ is strictly dense, exactly one of the triples $\rt{(ab|x)}$,
  $\rt{(ax|b)}$, or $\rt{(xb|a)}$ is contained in $R$.  Hence, $[R, \S\cup
  \{x\}]$ has exactly three vertices where two of them are linked by an
  edge.  Thus, $[R, \S\cup \{x\}]$ is disconnected.
  
  \emph{Case 3.}  Let $|\S|\geq 3$ with $\S_1=\{a_1,\ldots, a_n\}$ and
  $\S_2=\{b_1,\ldots, b_m\}$.  Since $R_{|L'}$ is consistent and strict
  dense and by construction of $\S_1$ and $\S_2$ it holds $\forall a_i,a_j
  \in \S_1, b_k \in \S_2, i \neq j : \rt{(a_ia_j|b_k)} \in R_{|L'}
  \subseteq R$ and $\forall a_i \in \S_1, b_k, b_l \in \S_2, k \neq l :
  \rt{(b_kb_l|a_i)} \in R_{|L'} \subseteq R$. Therefore, since $R$ is
  strictly dense, there cannot be any triple of the form $\rt{(a_ib_k|a_j)}$
  or $\rt{(a_ib_k|b_l)}$ with $a_i,a_j \in \S_1, b_k,b_l \in \S_2$ that is
  contained $R$.  It remains to show that $R$ is consistent. The following
  three subcases can occur.
  \begin{itemize}
  \item[3.a)] The connected components $C_1$ and $C_2$ of $[R_{|L'},\S]$
    are connected in $[R,\S\cup \{x\}]$. Hence, there must be a triple
    $\rt{(ab|x)}\in R$ with $a\in \S_1$ and $b\in \S_2$.  Hence, in order
    to prove that $R$ is consistent, we need to show that there is no
    triple $\rt{(cx|d)}$ contained $R$ for all $c,d \in \S$, which would
    imply that $[R,\S\cup \{x\}]$ stays disconnected.
  \item[3.b)] The connected component $C_1$ of $[R_{|L'},\S]$ is connected
    to $x$ in $[R,\S\cup \{x\}]$. Hence, there must be a triple
    $\rt{(ax|c)}\in R$ with $a\in \S_1$, $c\in \S$.  Hence, in order to
    prove that $R$ is consistent, we need to show that there are no triples
    $\rt{(b_kx|a_i)}$ and $\rt{(b_kx|b_l)}$ for all $a_i \in \S_1$, $b_k,
    b_l \in \S_2$, which would imply that $[R,\S\cup \{x\}]$ stays
    disconnected.
  \item[3.c)] As in Case $3.b)$, the connected component $C_2$ of
    $[R_{|L'},\S]$ might be connected to $x$ in $[R,\S\cup \{x\}]$ and we
    need to show that there are no triples $\rt{(a_ix|b_k)}$ and
    \rt{(a_ix|a_j)}$\rt{(a_ix|a_j)}$ for all $a_i,a_j \in \S_1$, $b_k \in
    \S_2$ in order to prove that $R$ is consistent.
  \end{itemize} 
  
  \emph{Case 3.a)} Let $\rt{(ab|x)} \in R$, $a \in \S_1$, $b \in
  \S_2$. First we show that for all $a_i \in \S_1$ holds $\rt{(a_ib|x)} \in
  R$. Clearly, if $\S_1=\{a\}$ the statement is trivially true. If
  $|\S_1|>1$ then $\{\rt{(ab|x)}, \rt{(a_ia|b)}\}\vdash \rt{(a_ib|x)}$ for
  all $a_i\in \S_1$. Since the closure of all two element subsets of $R$ is
  contained in $R$ and $\rt{(ab|x)},\rt{(a_ia|b)} \in R$ we can conclude
  that $\rt{(a_ib|x)}\in R$.  Analogously one shows that for all $b_k \in
  \S_2$ holds $\rt{(ab_k|x)} \in R$.
  
  \noindent
  Since $\{\rt{(a_ia|b_k)},\rt{(ab_k|x)}\}\vdash \rt{(a_ib_k|x)}$ and
  $\rt{(a_ia|b_k)},\rt{(ab_k|x)}\in R$ we can conclude that
  $\rt{(a_ib_k|x)} \in R$ for all $a_i\in \S_1$, $b_k\in
  \S_2$. Furthermore, $\{\rt{(a_ia_j|b)},\rt{(a_ib|x)}\}\vdash
  \rt{(a_ia_j|x)}$ for all $a_i,a_j\in \S_1$ and again, $\rt{(a_ia_j|x)}\in
  R$ for all $a_i,a_j\in \S_1$. Analogously, one shows that
  $\rt{(b_kb_l|x)}\in R$ for all $b_k,b_l\in \S_2$.
  
  \noindent
  Thus, we have shown, that for all $c,d\in \S$ holds that $\rt{(cd|x)}\in
  R$.  Since $R$ is strictly dense, there is no triple $\rt{(cx|d)}$
  contained in $R$ for any $c,d\in \S$. Hence, $[R,\S\cup \{x\}]$ is
  disconnected.

  \emph{Case 3.b)} Let $\rt{(ax|c)}\in R$ with $a\in \S_1$, $c\in
  \S$. Assume first that $c\in \S_1$. Then there is triple $\rt{(ac|b)}\in
  R$. Moreover, $ \{\rt{(ax|c)},\rt{(ac|b)}\}\vdash \rt{(ax|b)}$ and thus,
  $\rt{(ax|b)}\in R$. This implies that there is always some $c'=b\in \S_2$
  with $\rt{(ax|c')}\in R$. In other words, w.l.o.g. we can assume that for
  $\rt{(ax|c)}\in R$, $a\in \S_1$ holds $c\in \S_2$.
  
  \noindent
  Since $\{\rt{(ax|b)},\rt{(aa_i|b)}\}\vdash \rt{(a_ix|b)}$ and
  $\rt{(ax|b)},\rt{(aa_i|b)}\in R$ we can conclude that $\rt{(a_ix|b)}\in
  R$ for all $a_i\in \S_1$.  Moreover,
  $\{\rt{(a_ix|b)},\rt{(bb_k|a_i)}\}\vdash \rt{(a_ix|b_k)}$ and by similar
  arguments, $\rt{(a_ix|b_k)}\in R$ for all $a_i\in \S_1, b_k\in \S_2$.
  Finally, $\{\rt{(a_ix|b_k)},\rt{(b_lb_k|a_i)}\}\vdash \rt{(b_kb_l|x)}$,
  and therefore, $\rt{(b_kb_l|x)}\in R$ for all $b_k,b_l\in \S_2$. To
  summarize, for all $a_i \in \S_1, b_k,b_l \in \S_2$ we have
  $\rt{(a_ix|b_k)} \in R$ and $\rt{(b_kb_l|x)} \in R$. Since $R$ is strict
  dense there cannot be triples $\rt{(b_kx|a_i)}$ and $\rt{(b_kx|b_l)}$ for
  any $a_i \in \S_1$, $b_k, b_l \in \S_2$, and hence, $[R,\S\cup \{x\}]$ is
  disconnected.
  
  \emph{Case 3.c)} By similar arguments as in Case $3.b)$ and interchanging
  the role of $\S_1$ and $\S_2$, one shows that $[R,\S\cup \{x\}]$ is
  disconnected.
  
  In summary, we have shown that $[R,\S\cup\{x\}]$ is disconnected in all 
  cases. Therefore  $R$ is consistent.
\end{proof}

\begin{cpro}{\ref{pro:BinaryClDense}}
  Let $R$ be a consistent triple set on $L$. If the tree obtained with
  \texttt{BUILD} is binary, then the closure $\cl(R)$ is strictly dense.
  Moreover, this tree $T$ is unique and therefore, a least resolved tree
  for $R$.
\end{cpro}
\begin{proof}
  Note, the algorithm \texttt{BUILD} relies on the Aho graph $[R,\S]$ for
  particular subsets $\S\subseteq L$.  This means, that if the tree
  obtained with \texttt{BUILD} is binary, then for each of the particular
  subsets $\S\subseteq L$ the Aho graph $[R,\S]$ must have exactly two
  components.  Moreover, $R$ is consistent, since \texttt{BUILD} constructs
  a tree.

  Now consider arbitrary three distinct leaves $x,y,z\in L$.  Since $T$ is
  binary, there is a subset $\S\subseteq L$ with $x,y,z\in \S$ in some
  stage of \texttt{BUILD} such that two of the three leaves, say $x$ and $y$
  are in a different connected component than the leaf $z$. This implies
  that $R\cup \rt{(xy|z)}$ is consistent, since even if $\{x,y\}\not \in
  E([R,\S])$, the vertices $x$ and $y$ remain in the same connected
  component different from the one containing $z$ when adding the edge
  $\{x,y\}$ to $[R,\S]$.  Moreover, by the latter argument, both $R\cup
  \rt{(xz|y)}$ and $R\cup \rt{(yz|x)}$ are not consistent.  Thus, for any
  three distinct leaves $x,y,z\in L$  exactly one of the sets
  $R\cup\{\rt{(xy|z)}\}$, $R\cup\{\rt{(xz|y)}\}$, $R\cup\{\rt{(zy|x)}\}$
  is consistent, and thus, contained in the closure $\cl(R)$. Hence, $\cl(R)$
  is strictly dense.

  Since a tree $T$ that displays $R$ also displays $\cl(R)$ and because
  $\cl(R)$ is strictly dense and consistent, we can conclude that $\cl(R) =
  \mathfrak{R}(T)$ whenever $T$ displays $R$. Hence, $T$ must be unique and
  therefore, the least resolved tree for $R$. 
\end{proof}

\begin{lem}\label{A:lem:binstrictdense}
  Let $R$ be a consistent set of triples on $L$. Then there is a strictly dense 
  consistent triple set $R'$ on $L$ that contains $R$. 
\end{lem}
\begin{proof}
  Let $\mathrm{Aho}(R)$ be the tree constructed by \texttt{BUILD} from a
  consistent triple set $R$. It is in general not a binary tree. Let $T'$
  be a binary tree obtained from $\mathrm{Aho}(R)$ by substituting a binary
  tree with $k$ leaves for every internal vertex with $k>2$ children. Any
  triple $(ab|c)\in \mathfrak{R}(\mathrm{Aho}(R))$ is also displayed by
  $T'$ since unique disjoint paths $a-b$ and $c-\rho$ in $\mathrm{Aho}(R)$
  translate directly to unique paths in $T'$, which obviously are again
  disjoint.  Furthermore, a binary tree $T'$ with leaf set $L$ displays
  exactly one triple for each $\{a,b,c\}\in \binom{L}{3}$; hence $R'$ is
  strictly dense.
\end{proof}

\begin{rem}
  Let $T$ be a binary tree. Then $\mathfrak{R}(T)$ is strictly dense and
  hence, $\mathfrak{R}(T)\cup\{r\}$ is inconsistent for any triple $r\notin
  \mathfrak{R}(T)$.  Since $\mathfrak{R}(T)\subseteq
  \mathfrak{R}(\mathrm{Aho}(\mathfrak{R}(T))$ by definition of the action
  of $\texttt{BUILD}$ and there is no consistent triple set that strictly
  contains $\mathfrak{R}(T)$, we have $\mathfrak{R}(T)=
  \mathfrak{R}(\mathrm{Aho}(\mathfrak{R}(T))$. Thus
  $\mathrm{Aho}(\mathfrak{R}(T))=T$.
\end{rem}\bigskip

In order to discuss the relationship of alternative choices of least
  resolved trees we need a few additional definitions.  Let
  $\mathcal{C}(T)=\bigcup_{v\in V(T)}\{L(v)\}$ be the hierarchy defined by
  $T$. We say that a phylogenetic tree $S$ \emph{refines} a tree $T$, if
  $\mathcal{C}(T)\subseteq \mathcal{C}(S)$.  A collection of rooted triples
  $R$ \emph{identifies} a phylogenetic tree $T$ if $T$ displays $R$ and
  every other tree that displays $R$ is a refinement of $T$.  

\bigskip

  \begin{lem}
    Let $R$ be a consistent set of triples that identifies a phylogenetic
    tree $T$. Suppose the trees $T_1$ and $T_2$ display all triples of $R$
    so that $T_1$ has the minimum number of vertices among all trees in
    $\langle R \rangle$ and $T_2$ minimizes the cardinality
    $|\mathfrak{R}(T_2)|$. Then,
    $$T\simeq\mathrm{Aho}(R)\simeq T_1 \simeq T_2.$$
    \label{lem:allEqual}
  \end{lem}
  \begin{proof}
    Lemma 2.1 and 2.2 in \cite{GSS:07} state 
    that $R$ identifies $T$ iff $\mathfrak{R}(T) = \cl(R)$ and 
    that $T\simeq\mathrm{Aho}(R)$ in this case.
    Since $R$ identifies $T$,  any other tree that displays $R$ refines
    $T$ and thus, must have more vertices. Hence,  $T\simeq T_1$. 
    \newline
    Since the closure $\cl(R)$ must be displayed by all trees 
    that display $R$ it follows that $T$ is one of the trees that
    have a minimum cardinality set $\mathfrak{R}(T)$ and thus, 
    $|\mathfrak{R}(T_2)| = |\mathfrak{R}(T)|$ and hence, 
    $\mathfrak{R}(T_2) = \mathfrak{R}(T)=\cl(R)$. 
    Lemma 2.1 in \cite{GSS:07} implies that 
    $R$ identifies $T_2$. Lemma 2.2 in \cite{GSS:07} 
    implies that, therefore, $T_2\simeq \mathrm{Aho}(R)$. 
  \end{proof}

\subsection{Orthology Relations, Symbolic Representations, and Cographs}
\label{ss:cograph}

For a gene tree $T=(V,E)$ on $\Gen$ we define $t:V^0\to M$ as a map that
assigns to each inner vertex an arbitrary symbol $m\in M$.  Such a map $t$
is called a \emph{symbolic dating map} or \emph{event-labeling} for $T$; it
is \emph{discriminating} if $t(u) \neq t(v)$, for all inner edges
$\{u,v\}$, see \citep{Boeckner:98}.

In the rest of this paper we are interested only in event-labelings $t$
that map inner vertices into the set $M=\{\bullet, \square\}$, where the
symbol ``$\bullet$'' denotes a speciation event and ``$\square$'' a
duplication event. We denote with $(T,t)$ a gene tree $T$ with
corresponding event labeling $t$. If in addition the map $\sigma$ is given,
we write this as $(T,t;\sigma)$.

An orthology relation $\Theta \subset \Gen\times\Gen$ is a symmetric
relation that contains all pairs $(x,y)$ of orthologous genes. Note, this
implies that $(x,x)\notin \Theta$ for all $x\in \Gen$. Hence, its
complement $\overline \Theta $ contains all leaf pairs $(x,x)$ and pairs
$(x,y)$ of non-orthologous genes and thus, in this context all paralogous
genes.

For a given orthology relation $\Theta$ we want to find an event-labeled
phylogenetic tree $T$ on $\Gen$, with $t:V^0\to\{\bullet, \square\}$ such
that
\begin{enumerate}
\item $t(\lca_{T}(x,y))=\bullet$ for all $(x,y)\in \Theta$
\item $t(\lca_{T}(x,y))=\square$ for all $(x,y)\in \overline \Theta
  \setminus\{(x,x)\mid x\in \Gen\}$.
\end{enumerate}

In other words, we want to find an event-labeled tree $T$ on $\Gen$ such
that the event on the most recent common ancestor of the orthologous genes
is a speciation event and of paralogous genes a duplication event.  If such
a tree $T$ with (discriminating) event-labeling $t$ exists for $\Theta$, we
call the pair $(T,t)$ a \emph{(discriminating) symbolic representation} of
$\Theta$.

\subsubsection{Symbolic Representations and Cographs} 
\label{sect:cograph}

Empirical orthology estimations will in general contain false-positives.
In addition orthologous pairs of genes may have been missed due to the
scoring function and the selected threshold. Hence, not for all estimated
orthology relations there is such a tree.  In order to characterize
orthology relations we define for an arbitrary symmetric relation $R
\subseteq \Gen\times \Gen$ the underlying graph $G_R=(\Gen, E_R)$ with edge
set $E_R=\left\{\{x,y\}\in \binom{\Gen}{2}\mid (x,y)\in R\right\}$.

As we shall see, orthology relations $\Theta$ and cographs are closely
related. A cograph is a $P_4$-free graph (i.e.\ a graph such that no four
vertices induce a subgraph that is a path on $4$ vertices), although there
are a number of equivalent characterizations of such graphs (see e.g.\
\citep{Brandstaedt:99} for a survey).

It is well-known in the literature concerning cographs that, to any cograph
$G=(V,E)$, one can associate a canonical \emph{cotree} $\CoT(G)=(W\cup
V,F)$ with leaf set $V$ together with a labeling map $\lambda_G:W\to
\{0,1\}$ defined on the inner vertices of $\CoT(G)$. The key observation is
that, given a cograph $G=(V,E)$, a pair $\{x,y\} \in \binom{V}{2}$ is an
edge in $G$ if and only if $\lambda_G(\lca_{\CoT(G)}(x,y))=1$
(cf. \citep[p. 166]{Corneil:81}). The next theorem summarizes the results,
that rely on the theory of so-called symbolic ultrametrics developed in
\citep{Boeckner:98} and have been established in a more general context in
\citep{Hellmuth:13d}.\bigskip 

\begin{thm}[\citealp{Hellmuth:13d}] \label{A:thm:ortho-cograph}
Suppose that $\Theta$ is an (estimated) orthology relation and
denote by $\overline{\Theta}^{\neq}:=\overline{\Theta} \setminus\{(x,x)\mid x\in \Gen\}$
the complement of $\Theta$ without pairs $(x,x)$. 
Then the following statements are equivalent:
\begin{itemize}
\item[(i)] $\Theta$ has a symbolic representation.
\item[(ii)] $\Theta$ has a discriminating symbolic representation.
\item[(iii)] $G_\Theta= \overline{G}_{\overline{\Theta}^{\neq}}$ 
is a cograph.
\end{itemize}
\end{thm}

This result enables us to find the corresponding discriminating
symbolic representation $(T,t)$ for $\Theta$ (if one exists) by
identifying $T$ with the respective cotree $\CoT(G_\Theta)$ of the
cograph $G_\Theta$ and setting $t(v)=\bullet$ if $\{x,y\}\in
E(G_\Theta)$ and thus, $\lambda_{G_\Theta}(v)=1$ and $t(v)=\square$ if
$\{x,y\}\not\in E(G_\Theta)$ and thus $\lambda_{G_\Theta}(v)=0$

We identify the discriminating symbolic representation $(T,t)$ for 
$\Theta$ (if one exists) with the  cotree $\CoT(G_\Theta)$ as explained 
above.

\subsubsection{Cograph Editing} 

It is well-known that cographs can be recognized in linear time
\citep{Corneil:85, habib2005simple}. However, the cograph editing problem,
that is given a graph $G=(V,E)$ one aims to convert $G$ into a cograph
$G^*=(V,E^*)$ such that the number $|E\symdiff E^*|$ of inserted or deleted
edges is minimized is an NP-complete problem \citep{Liu:11, Liu:12}.
In view of the above results, this implies the following:\bigskip 

\begin{thm}
Let $\Theta\subset \Gen \times \Gen$ be an (estimated) orthology relation. 
It can be recognized in linear time whether 
$\Theta$ has a (discriminating) symbolic representation. 

For a given	 positive integer $K$
the problem of deciding if there is an orthology relation
$\Theta^*$  that
has a (discriminating) symbolic representation
s.t. $|\Theta\symdiff \Theta^*|\leq K$ is NP-complete.
\end{thm}
\bigskip 

As the next result shows, it suffices to solve the cograph editing problem separately for the 
connected components of $G$.

\begin{lem}
  For any graph $G(V,E)$ let $F\in\binom{V}{2}$ be a minimal set of edges
  so that $G'=(V,E\symdiff F)$ is a cograph. Then $(x,y)\in F\setminus E$
  implies that $x$ and $y$ are located in the same connected component of
  $G$.
  \label{A:lem:disconnected}
\end{lem}
\begin{proof}
  Suppose, for contradiction, that there is a minimal set $F$ connecting
  two distinct connected components of $G$, resulting in a cograph $G'$.
  W.l.o.g., we may assume that $G$ has only two connected components
  $C_1,C_2$.  Denote by $G''$ the graph obtained from $G'$ by removing all
  edges $\{x,y\}$ with $x\in V(C_1)$ and $y\in V(C_2)$. If $G''$ is
  not a cograph, then there is an induced $P_4$, which must be contained 
  in one of the connected components of $G''$. By construction this
  induced $P_4$ is also contained in $G'$. Since $G'$ is a cograph no such 
  $P_4$ exists and hence $G''$ is also a cograph, contradicting the
  minimality of $F$. 
\end{proof}

\subsection{From Gene Triples to Species Triples and Reconciliation Maps}

A gene tree $T$ on $\Gen$ arises in evolution by means of a series of
events along a species tree $S$ on $\Spe$. In our setting these may be
duplications of genes within a single species and speciation events, in
which the parent's gene content is transmitted to both offsprings. The
connection between gene and species tree is encoded in the reconciliation
map, which associates speciation vertices in the gene tree with the
interior vertex in the species tree representing the same speciation
event. We consider the problem of finding a species tree for a given gene
tree. In this subsection We follow the presentation of
\citet{hernandez2012event}.

\subsubsection{Reconciliation Maps} 

We start with a formal definition of reconciliation maps. \smallskip

\begin{defi}[\citealp{hernandez2012event}] \label{A:def:mu} Let $S=(W,F)$
  be a species tree on $\Spe$, let $T=(V,E)$ be a gene tree on $\Gen$ with
  corresponding event labeling $t:V^0\to \{\bullet,\square\}$ and suppose
  there is a surjective map $\sigma$ that assigns to each gene the
  respective species it is contained in.  Then we say that $S$ is a
  \emph{species tree for $(T,t;\sigma)$} if there is a map $\mu:V\to W\cup
  F$ such that, for all $x\in V$:
\begin{itemize}
\item[(i)]   If $x\in \Gen$ then $\mu(x)=\sigma(x)$.
\item[(ii)]  If $t(x)=\bullet$ then $\mu(x)\in W\setminus \Spe$.
\item[(iii)] If $t(x)=\square$ then $\mu(x)\in F$. 
\item[(iv)]  Let $x,y\in V$ with $x\prec_T y$. We distinguish two cases:
  \begin{enumerate}
  \item  If $t(x)=t(y)=\square$ then $\mu(x)\preceq_S \mu(y)$ in $S$.
  \item  If $t(x)=t(y)=\bullet$ or $t(x)\neq t(y)$ then 
			    $\mu(x)\prec_S\mu(y)$ in $S$.
  \end{enumerate}
\item[(v)]   If $t(x)=\bullet$ then 
     			   $\mu(x)=\lca_S( \sigma(L(x)) )$ 
\end{itemize}
We call $\mu$ the reconciliation map from $(T,t,\sigma)$ to $S$. 
\end{defi}                \smallskip 

A reconciliation map $\mu$ maps leaves $x\in \Gen$ to leaves
$\mu(x):=\sigma(x)$ in S and inner vertices $x\in V^0$ to inner vertices
$w\in W\setminus \Spe$ in $S$ if $t(x)=\bullet$ and to edges $f\in F$ in
$S$ if $t(x)=\square$, such that the ancestor relation $\preceq_S$ is
implied by the ancestor relation $\preceq_T$.  Definition \ref{A:def:mu} is
consistent with the definition of reconciliation maps for the case when the
event labeling $t$ on $T$ is not known, see \citep{Doyon:09}.

\subsubsection{Existence of a Reconciliation Map} 

The reconciliation of gene and species trees is usually studied in the
  situation that only $S$, $T$, and $\sigma$ are known and both $\mu$ and
  $t$ and must be determined
  \cite{Guigo1996,Page1997,Arvestad2003,Bonizzoni2005,Gorecki2006,%
    Hahn2007,Bansal2008,Chauve2008,Burleigh2009,Larget2010}. In this form,
  there is always a solution $(\mu,t)$, which however is not unique in
  general. A variety of different optimality criteria have been used in the
  literature to obtain biologically plausible reconciliations.  The
situation changes when not just the gene tree $T$ but a symbolic
representation $(T,t)$ is given. Then a species tree need not exists.
\citet{hernandez2012event} derived necessary and sufficient conditions for
the existence of a species tree $S$ so that there exists a reconciliation
map from $(T,t)$ to $S$. We briefly summarize the key results.

For $(T,t;\sigma)$ we define the triple set
\begin{align*}
  \begin{split}
    \mathbb{G}=   \left\{ r \in \mathfrak{R}(T)\big\vert 
      t(\lca_T(L_r))=\bullet \, \textrm{ and } \,
      \sigma(x)\not=\sigma(y),\,
    \right. \\
    \left.  \textrm{for all }\,
      x,y\in L_r\,\textrm{pairwise distinct}\right\}
  \end{split}
\end{align*}
In other words, the set $\mathbb G$ contains all triples $r=\rt{(ab|c)}$ of
$\mathfrak{R}(T)$ where all three genes in $a,b,c\in L_r$ are contained in
different species and the event at the most recent common ancestor of $L_r$
is a speciation event, i.e., $t(\lca_T(a,b,c))=\bullet$. It is easy to see
that in this case $S$ must display $\rt{(\sigma(a)\sigma(b)|\sigma(c))}$,
i.e., it is a necessary condition that the triple set
\begin{align*}
  \mathbb{S}= 
  \left\{ \rt{(\al\be|\ga)}   |\, \exists \rt{(ab|c)}\in\mathbb{G} 
    \textrm{\ with\ } 
    \sigma(a)=\al,\sigma(b)=\be,\sigma(c)=\ga \right\}
\end{align*}
is consistent. This condition is also sufficient:\bigskip 

\begin{thm}[\citealp{hernandez2012event}] 
There is a species tree on $\sigma(\Gen)$ for $(T,t,\sigma)$ if
and only if the triple set $\mathbb{S}$ is consistent.
A reconciliation map can then be found in polynomial time. 
\label{A:thm:gen-spec-tree}
\end{thm}

\subsubsection{Maximal Consistent Triple Sets}

In general, however, $\mathbb{S}$ may not be consistent.  In this case it
is impossible to find a valid reconciliation map.  However, for each
consistent subset $\mathbb S^*\subset \mathbb S$, its corresponding species
tree $S^*$, and a suitably chosen homeomorphic image of $T$ one can find the
reconciliation.  For a phylogenetic tree $T$ on $L$, the \emph{restriction}
$T|_{L'}$ of $T$ to $L'\subseteq L$ is the phylogenetic tree with leaf set
$L'$ obtained from $T$ by first forming the minimal spanning tree in $T$
with leaf set $L'$ and then by suppressing all vertices of degree two with
the exception of $\rho_T$ if $\rho_T$ is a vertex of that tree, see
\citep{sem-ste-03a}. For a consistent subset $\mathbb{S}^*\subset
\mathbb{S}$ let $L'=\{x\in \Gen\mid \exists r\in \mathbb{S}^*$ with
$\sigma(x)\in L_r\}$ be the set of genes (leaves of $T|_{L'}$) for which a
species $\sigma(x)$ exits that is also contained in some triple $r\in
\mathbb{S}^*$. Clearly, the reconciliation map of $T|_{L'}$ and the species
tree $S^*$ that displays $\mathbb{S}^*$ can then be found in polynomial
time by means of Theorem \ref{A:thm:gen-spec-tree}.

\section{ILP Formulation}

The workflow outline in the main text consists of three stages, each of
which requires the solution of hard combinatorial optimization problem.
Our input data consist of an $\Theta$ or of a weighted version thereof.  In
the weighted case we assume the edge weights $w(x,y)$ have values in the
unit interval that measures the confidence in the statement
``$(x,y)\in\Theta$''.  Because of measurement errors, our first task is to
correct $\Theta$ to an irreflexive, symmetric relation $\Theta^*$ that is a
valid orthology relation. As outlined in section~\ref{sect:cograph},
$G_{\Theta^*}$ must be cograph so that $(x,y)\in\Theta^*$ implies
$\sigma(x)\ne\sigma(y)$. By Lemma~\ref{A:lem:disconnected} this problem has
to be solved independently for every connected component of $G_{\Theta}$.
The resulting relation $\Theta^*$ has the symbolic representation $(T,t)$.

In the second step we identify the best approximation of the species tree
induced by $(T,t)$. To this end, we determine the maximum consistent subset
$\mathbb{S}^*$ in the set $\mathbb{S}$ of species triples induced by those
triples of $(T,t)$ that have a speciation vertex as their root. The hard
part in the ILP formulation for this problem is to enforce consistency of a
set of triples \cite{chang2011ilp}. This step can be simplified
considerably using the fact that for every consistent triple set
$\mathbb{S}^*$ there is a strictly dense consistent triple set $\mathbb{S}'$
that contains $\mathbb{S}^*$ (Lemma~\ref{A:lem:binstrictdense}). This allows
us to write $\mathbb{S}^*=\mathbb{S}'\cap \mathbb{S}$.  The gain in
efficiency in the corresponding ILP formulation comes from the fact that a
strictly dense set of triples is consistent if and only if all its
two-element subsets are consistent
(Theorem~\ref{thm:consistIFFpairwise}), allowing for a much faster check
of consistency.

In the third step we determine the least resolved species tree $S$ from the
triple set $\mathbb{S}^*$ since this tree makes least assumptions of the
topology and thus, of the evolutionary history. In particular, it displays
only those triples that are either directly derived from the data or that
are logically implied by them. Thus $S$ is the tree with the minimal number
of (inner) vertices that displays $\mathbb{S}^*$. Our ILP formulation uses
ideas from the work of \cite{chang2011ilp} to construct $S$ in the form of
an equivalent partial hierarchy.

\subsection{Cograph Editing}

Given the edge set of an input graph, in our case the pairs
$(x,y)\in\Theta$, our task is to determine a modified edge set so that the
resulting graph is a cograph. The input is conveniently represented by
binary constants $\Theta_{ab}=1$ iff $(a,b)\in \Theta$.  The edges of the
adjusted cograph $G_{\Theta^*}$ are represented by binary variables
$E_{xy}=E_{yx}=1$ if and only if $\{x,y\}\in E(G_{\Theta^*})$.  Since
$E_{xy}\equiv E_{yx}$ we use these variables interchangeably, without
distinguishing the indices.  Since genes residing in the same organism
cannot be orthologs, we exclude edges $\{x,y\}$ whenever
$\sigma(x)=\sigma(y)$ (which also forbids loops $x=y$.  This is expressed
by setting
\begin{align}\tag{\ref{ilp:forbid_E}} 
  E_{xy}=0 \text{ for all } \{x,y\}\in \binom{\Gen}{2} \text{ with }
  \sigma(x)=\sigma(y).
\end{align}
To constrain the edge set of $G_{\Theta^*}$ to cographs, we use the fact
that cographs are characterized by $P_4$ as forbidden subgraph. This can be 
expressed as follows. For every ordered four-tuple $(w,x,y,z)\in\Gen^4$ 
with pairwise distinct $w,x,y,z$ we require  
\begin{align}
  E_{wx} + E_{xy}+ E_{yz} - E_{xz} - E_{wy} - E_{wz} \leq 2 \tag{\ref{ilp:cog}}
\end{align}
Constraint \eqref{ilp:cog} ensures that for each ordered tuple $(w,x,y,z)$
it is not the case that there are edges $\{w,x\}$, $\{x,y\}$, $\{y,z\}$ and
at the same time no edges $\{x,z\}$, $\{w,y\}$, $\{w,z\}$ that is, $w,x,y$
and $z$ induce the path $w-x-y-z$ on four vertices. Enforcing this
constraint for all orderings of $w,x,y,z$ ensures that the subgraph induced
by $\{w,x,y,z\}$ is $P_4$-free. 

In order to find the closest orthology cograph $G_{\Theta^*}$ we 
minimize the symmetric difference of the estimated and adjusted 
orthology relation. Thus the objective function is 
\begin{align} \tag{\ref{ilp:minDiff}}
\min & \sum_{(x,y)\in\Gen \times \Gen} (1-\Theta_{xy}) E_{xy} + 
       \sum_{(x,y)\in \Gen \times \Gen } \Theta_{xy} (1-E_{xy}) 
\end{align}

\begin{rem}\label{rem:real}
  We have defined $\Theta$ above as a binary relation. The problem can be
  generalized to a weighted version in which the input $\Theta$ 
  is a real valued function $\Theta:  \Gen \times \Gen \to [0,1]$
  measuring the confidence with which a pair $(x,y)$ is orthologous. 
  The ILP formulation remains unchanged. 
\end{rem}
\bigskip 

The latter ILP formulation makes use of $O(|\Gen|^2)$ variables and 
Equations \eqref{ilp:forbid_E} and \eqref{ilp:cog} impose 
$O(|\Gen|^4)$ constraints. 

\subsection{Extraction of All Species Triples}

Let $\Theta$ be an orthology relation with symbolic representation
$(T,t;\sigma)$ so that $\sigma(x)=\sigma(y)$ implies $(x,y)\notin\Theta$.
By Theorem~\ref{A:thm:gen-spec-tree}, the species tree $S$ displays all
triples $\rt{(\al\be|\ga)}$ with a corresponding gene triple
$\rt{(xy|z)}\in \mathbb G\subseteq \mathfrak R(T)$, i.e., a triple
$\rt{(xy|z)}$ with speciation event at the root of
$t(\lca_T(x,y,z)=\bullet$ and $\sigma(x)=\al$, $\sigma(y)=\be$,
$\sigma(z)=\ga$ are pairwise distinct species. We denote the set of these
triples by $\mathbb{S}$. Although all species triples can be extracted 
in polynomial 
time, e.g. by using the $\texttt{BUILD}$ algorithm, we give here an ILP
formulation to complete the entire ILP pipeline. It will also be useful as
a starting point for the final step, which consists in finding a minimally
resolved trees that displays $\mathbb{S}$. Instead of using the symbolic
representation $(T,t;\sigma)$ we will directly make use of the information
stored in $\Theta$ using the following simple observation.\bigskip 

\begin{lem}
  Let $\Theta$ be an orthology relation with discriminating symbolic
  representation $(T,t;\sigma)$ that is identified with the cotree of the
  corresponding cograph $G_\Theta=(\Gen, E_\Theta)$. Assume that
  $\rt{(xy|z)}\in \mathfrak{R}(T)$ is a triple where all genes $x,y,z$ are
  contained in pairwise different species.  Then it holds:
  $t(\lca(x,y))=\square$ if and only if $\{x,y\} \notin E_\Theta$ and
  $t(\lca(x,y,z))=\bullet$ if and only if $\{x,z\},\{y,z\}\in E_\Theta$
  \label{A:lem:bulletifftheta}
\end{lem}
\begin{proof}
  Assume there is a triple $\rt{(xy|z)}\in \mathfrak R(T)$ where all genes
  $x,y,z$ are contained in pairwise different species.  Clearly,
  $t(\lca(x,y))=\square$ iff $(x,y)\notin \Theta$ iff $\{x,y\} \notin
  E_\Theta$.  Since, $\lca(x,y)\neq \lca(x,z)=\lca(y,z)=\lca(x,y,z)$ we
  have $t(\lca(x,z))=t(\lca(y,z)) = \bullet$, which is iff $(x,z), (y,z)
  \in \Theta$ and thus, iff $\{x,z\},\{y,z\}\in E_\Theta$.
\end{proof}

The set $\mathbb{S}$ of species triples is encoded by the binary variables
$T_{\rt{(\al\be|\ga)}}=1$ iff $\rt{(\al\be|\ga)}\in\mathbb{S}$. Note that
$\rt{(\be\al|\ga)}\equiv\rt{(\al\be|\ga)}$. In order to avoid superfluous
variables and symmetry conditions connecting them we assume that the first
two indices in triple variables are ordered.  Thus there are three triple
variables $T_{\rt{(\al\be|\ga)}}$, $T_{\rt{(\al\ga|\be)}}$, and
$T_{\rt{(\be\ga|\al)}}$ for any three distinct $\al,\be,\ga\in\Spe$.

Assume that $\rt{(xy|z)}\in \mathfrak{R}(T)$ is an arbitrary triple
displayed by $T$. In the remainder of this section, we assume 
that these genes $x,y$ and $z$ are from pairwise different species
$\sigma(x)=\al$, $\sigma(y)=\be$ and $\sigma(z)=\ga$. Given that in addition
$t(\lca(x,y,z))=\bullet$, we need to ensure that $T_{\rt{(\al\be|\ga)}}=1$. If
$t(\lca(x,y,z))=\bullet$ then there are two cases: \emph{(1)}
$t(\lca(x,y))=\square$ or \emph{(2)} $t(\lca(x,y))=\bullet$. These two
cases needs to be considered separately for the ILP formulation.

\emph{Case (1) $t(\lca(x,y))=\square\neq t(\lca(x,y,z))$:}
Lemma~\ref{A:lem:bulletifftheta} implies that $E_{xy}=0$ and 
$E_{xz}=E_{yz}=1$. This yields, $(1-E_{xy})+E_{xz}+E_{yz}=3$.
To infer that in this case $T_{\rt{(\al\be|\ga)}}=1$ we add the next 
constraint.

\begin{ILP}
  (1-E_{xy})+E_{xz}+E_{yz}-T_{\rt{(\al\be|\ga)}}   & \leq 2 
  \label{ilp:inferOne}\\[-0.1cm]
\end{ILP}

These constraints need, by symmetry, also be added for the possible triples
$\rt{(xz|y)}$, resp., $\rt{(yz|x)}$ and the corresponding species triples
$\rt{(\al\ga|\be)}$, resp., $\rt{(\be\ga|\al)}$:
\begin{align} 
  E_{xy}+(1-E_{xz})+E_{yz}-T_{\rt{(\al\ga|\be)}} \leq 2 \tag{\ref{ilp:inferOne}}\\ 
  E_{xy}+E_{xz}+(1-E_{yz})-T_{\rt{(\be\ga|\al)}} \leq 2 \notag 
\end{align}

\emph{Case (2) $t(\lca(x,y))=\bullet=t(\lca(x,y,z))$:} Lemma
\ref{A:lem:bulletifftheta} implies that $E_{xy}=E_{xz}=E_{yz}=1$.  Since
$\lca(x,y) \neq \lca(x,y,z)$ and the gene tree we obtained the triple from
is a discriminating representation, that is consecutive event labels are
different, there must be an inner vertex $v\not\in \{\lca(x,y),
\lca(x,y,z)\}$ on the path from $\lca(x,y)$ to $\lca(x,y,z)$ with
$t(v)=\square$. Since $T$ is a phylogenetic tree, there must be a leaf
$w\in L(v)$ with $w\neq x,y$ and $\lca(x,y,w)=v$ which implies
$t(\lca(x,y,w))=t(v)=\square$. For this vertex $w$ we derive that $(xw|z),
(yw|z) \in \mathfrak R(T)$ and in particular, $\lca(y,w,z)=\lca(x,y,z) =
\lca(w,z)$. Therefore, $t(\lca(y,w,z)) =t(\lca(w,z))=\bullet$.

Now we have to distinguish two subcases; either \emph{Case (2a)}
$\sigma(x)=\al=\sigma(w)$ (analogously one treats the case
$\sigma(y)=\be=\sigma(w)$ by interchanging the role of $x$ and $y$) or
\emph{Case (2b)} $\sigma(x)=\al\neq\sigma(w)=\de\notin\{\al,\be,\ga\}$.
Note, the case $\sigma(w)=\sigma(z)=\ga$ cannot occur, since we obtained
$(T,t)$ from the cotree of $G_\Theta$ and in particular, we have
$t(\lca(w,z))=\bullet$.  Therefore, $E_{wz}=1$ and hence, by Constraint
\ref{ilp:forbid_E} it must hold $\sigma(w)\neq \sigma(z)$.

\begin{itemize}
\item[(2a)] Since $t(\lca(y,w,z))=\bullet$ and $v=\lca(y,w)$ with
  $t(v)=\square$ it follows that the triple $(yw|z)$ fulfills the
  conditions of \emph{Case 1}, and hence $T_{\rt{(\al\be|\ga)}}=1$ and we
  are done.
\item[(2b)] Analogously as in Case (2a), the triples $(xw|z)$ and $(yw|z)$
  fulfill the conditions of Case (1), and hence we get
  $T_{\rt{\rt{(\al\de|\ga)}}}=1$ and $T_{\rt{(\be\de|\ga)}}=1$. However, we
  must ensure that also the triple $\rt{(\al\be|\ga)}$ will be determined
  as observed species triple.  Thus we add the constraint:
  \begin{align}
    T_{\rt{\rt{(\al\de|\ga)}}}+T_{\rt{(\be\de|\ga)}}-T_{\rt{(\al\be|\ga)}}
    \leq 1 \tag{\ref{ilp:inferOne}} 
  \end{align}
  which ensures that $T_{\rt{(\al\be|\ga)}}=1$ whenever 
  $T_{\rt{\rt{(\al\de|\ga)}}}=T_{\rt{(\be\de|\ga)}}=1$.  
\end{itemize}
The first three constraints in Eq. \eqref{ilp:inferOne} are added for all
$\{x,y,z\}\in \binom{\Gen}{3}$ and where all three genes are contained in
pairwise different species $\sigma(x)=\al$, $\sigma(y)=\be$ and
$\sigma(z)=\ga$ and the fourth constraint in Eq. \eqref{ilp:inferOne}
is added for all $\{\al,\be,\ga,\de\}\in \binom{\Spe}{4}$.

In particular, these constraints ensure, that for each triple
$\rt{(xy|z)}\in \mathbb G$ with speciation event on top and corresponding
species triple $\rt{(\al\be|\ga)}$ the variable $T_{\rt{(\al\be|\ga)}}$ is
set to $1$.

However, the latter ILP constraints allow some degree of freedom for the
choice of the binary value $T_{\rt{(\al\be|\ga)}}$, where for all
respective triples $\rt{(xy|z)}\in \mathfrak R(T)$ holds
$t(\lca(x,y,z))=\square$.  To ensure, that only those variables
$T_{\rt{(\al\be|\ga)}}$ are set to $1$, where at least one triple
$\rt{(xy|z)}\in \mathfrak R(T)$ with $t(\lca(x,y,z))=\bullet$ and
$\sigma(x)=\al$, $\sigma(y)=\be$, $\sigma(z)=\ga$ exists, we add the
following objective function that minimizes the number of variables
$T_{\rt{(\al\be|\ga)}}$ that are set to $1$:

\begin{ILP}
  \min \sum_{\{\al,\be,\ga\}\in \binom{S}{3}} 
  T_{\rt{(\al\be|\ga)}}+T_{\rt{(\al\ga|\be)}}+T_{\rt{(\be\ga|\al)}}
  \label{ilp:min1triple}
\end{ILP}

For the latter ILP formulation $O(|\Spe|^3)$ variables 
and $O(|\Gen|^3+|\Spe|^4)$ constraints are required.

\subsection{Find Maximal Consistent Triple Set}
Given the set of species triple $\mathbb{S}$ the next step is to extract a
maximal subset $\mathbb{S}^* \subseteq \mathbb{S}$ that is consistent.
This combinatorial optimization problem is known to be NP-complete
\cite{Jansson2001,Wu2004}. In an earlier ILP approach, \citet{chang2011ilp}
explicitly constructed a tree that displays $\mathbb{S}^*$.  In order to
improve the running time of the ILP we focus here instead on constructing a
consistent, strictly dense triple set $\mathbb{S}$' containing the desired
solution $\mathbb{S^*}$ because the consistency check involves two-element
subsets in this case (Theorem \ref{thm:consistIFFpairwise}). From
$\mathbb{S}'$ obtain the desired solution as
$\mathbb{S^*}=\mathbb{S}'\cap\mathbb{S}$. We therefore introduce binary
variables $T'_{\rt{(\al\be|\ga)}}=1$ iff $\rt{(\al\be|\ga)}\in\mathbb{S}'$.

To ensure, that $\mathbb{S}'$ is strictly dense we add for all
$\{\al,\be,\ga\}\in \binom{\mathcal{S}}{3}$ the constraints:
\begin{align} \tag{\ref{ilp:sd}}
  &T'_{\rt{(\al\be|\ga)}}+T'_{\rt{(\al\ga|\be)}}+T'_{\rt{(\be\ga|\al)}} = 1.
\end{align}
We can now apply the inference rules in
Eq. \eqref{eq:infRule2} and the results of Theorem
\ref{thm:consistIFFpairwise} and Lemma \ref{lem:suffRule}.  
Therefore, we add the following constraint for all ordered tuples $(\al,\be,\ga,\de)$
for all $\{\al,\be,\ga,\de\}\in \binom{\Spe}{4}$:

\begin{align} 
  2T'_{\rt{(\al\be|\ga)}} + 2&T'_{\rt{(\al\de|\be)}}- 
  T'_{\rt{(\be\de|\ga)}} - T'_{\rt{\rt{(\al\de|\ga)}}} \leq 2  
  \tag{\ref{ilp:eq:infRule2}}
\end{align}

The constraint in Eq. \eqref{ilp:eq:infRule2} is a 
direct translation of the inference rule in Eqn. \eqref{eq:infRule2}.
Moreover, by Theorem \ref{thm:consistIFFpairwise} and Lemma \ref{lem:suffRule},
we know that testing pairs of triples with Eq. \eqref{eq:infRule2} is
sufficient for verifying consistency.

To ensure maximal cardinality of $\mathbb S^* = \mathbb S' \cap \mathbb S$
we use the  objective function
\begin{align} 
 \max \sum_{\rt{(\al\be|\ga)}\in \mathbb S} T'_{\rt{(\al\be|\ga)}} 
 \tag{\ref{ilp:maxdense}}
\end{align}

This ILP formulation can easily be adapted to solve a ``weighted''
maximum consistent subset problem: With $w\rt{(\al\be|\ga)}$ we denote for
every species triple $\rt{(\al\be|\ga)}\in \mathbb S$ the number of
connected components in $G_{\Theta^*}$ that contains three
vertices $a,b,c\in \Gen$ with $\rt{(ab|c)}\in \mathbb G$ and
$\sigma(a)=\al,\sigma(b)=\be, \sigma(c)=\ga$. In this way, we increase
the significance of species triples in $\mathbb S$ that have been
observed more times, when applying the following objective function.
\begin{align}
  \max \sum_{\rt{(\al\be|\ga)}\in \mathbb S}
   T'_{\rt{(\al\be|\ga)}}*w\rt{(\al\be|\ga)}. \tag{\ref{ilp:wmax}}
\end{align}

Finally, we define binary variables $T^*_{\rt{(\al\be|\ga)}}$ that indicate
whether a triple $\rt{(\al\be|\ga)}\in \mathbb{S}$ is contained in a
maximal consistent triples set $\mathbb{S}^*\subseteq \mathbb{S}$, i.e.,
$T^*_{\rt{(\al\be|\ga)}}=1$ iff $\rt{(\al\be|\ga)}\in\mathbb S^*$ and thus, 
iff $T_{\rt{(\al\be|\ga)}}=1$ and $T'_{\rt{(\al\be|\ga)}}=1$.
Therefore, we add for all $\{\al,\be,\ga\}\in \binom {\mathcal S}{3}$ 
the binary variables 
$T^*_{\rt{(\al\be|\ga)}}$ and add the constraints
\begin{equation}
  0 \leq T'_{\rt{(\al\be|\ga)}} + T_{\rt{(\al\be|\ga)}} - 2T^*_{\rt{(\al\be|\ga)}} \leq 1
  \tag{\ref{eq:tstar}}
\end{equation}

It is easy to verify, that in the latter ILP formulation 
$O(|\Spe|^3)$ variables and $O(|\Spe|^4)$ constraints are required.

\subsection{Least Resolved Species Tree}

The final step consists in finding a minimally resolved tree that displays
all triples of $\mathbb S^*$ and, in addition, has a minimum number
of inner vertices. The variables $T^*_{\rt{(\al\be|\ga)}}$
defined in the previous step take on the role of constants here. 

There is an ILP approach by \cite{chang2011ilp}, for determining a maximal
consistent triple sets. However, this approach relies on determining
consistency by checking and building up a binary tree, a very time
consuming task.  As we showed, this can be improved and simplified by the
latter ILP formulation. However, we will adapt now some of the ideas
established by \cite{chang2011ilp}, to solve the NP-hard problem
\cite{Jansson:12} of finding a least resolved tree.

To build an arbitrary tree for the consistent triple set $\mathbb{S}^*$,
one can use the fast algorithm \texttt{BUILD}
\citep{sem-ste-03a}. Moreover, if the tree obtained by \texttt{BUILD} for
$\mathbb{S}^*$ is a binary tree, then Proposition \ref{pro:BinaryClDense}
implies that the closure $\cl(\mathbb{S}^*)$ is strictly dense and that this
tree is a unique and least resolved tree for $\mathbb{S}^*$.  Hence, as a
preprocessing step one could use \texttt{BUILD} first, to test whether the
tree for $\mathbb{S}^*$ is already binary and if not, proceed with the
following ILP approach.

A phylogenetic tree $S$ is uniquely determined by hierarchy $\mathcal{C} =
\{L(v)\mid v\in V(S)\}$ according to Theorem \ref{A:thm:hierarchy}. Thus it
is possible to construct $S$ by building the clusters induced by the
triples of $\mathbb{S}^*$. Thus we need to translate the condition for
$\mathcal{C}$ to be a hierarchy into the language of ILPs. 

Following \cite{chang2011ilp} we use a binary $|\Spe|\times N$ matrix $M$,
with entries $M_{\al p}=1$ iff species $\al$ is contained in cluster
$p$. By Lemma \ref{A:lem:nrC}, it is clear that we need at most $2|\Spe|-1$
columns. As we shall see later, we exclude (implicitly) the trivial
singleton clusters $\{x\}\in \Spe$ and the cluster $\Spe$. Hence, it
suffices to use $N=2|\Spe|-1-|\Spe|-1=|\Spe|-2$ clusters. Each cluster $p$,
which is represented by the $p$-th column of $M$, corresponds to an inner
vertex $v_p$ in the species tree $S$ so that $p=(L(v_p))$.

Since we are interested in finding a least resolved tree rather than a
fully resolved one, we allow that number of clusters is smaller than $N-2$,
i.e., we allow that some columns of $M$ have no non-zero entries. Here, we
deviate from the approach of \cite{chang2011ilp}. Columns $p$ with
$\sum_{\al\in \Spe} M_{\al p}=0$ containing only $0$ entries and thus,
clusters $L(v_p)=\emptyset$, are called \emph{trivial}, all other columns
and clusters are called \emph{non-trivial}. Clearly, the non-trivial
clusters correspond to the internal vertices of $S$, hence we have to
maximize the number of trivial columns of $M$. This condition also
suffices to remove redundancy, i.e., non-trivial columns with the same
entries.

We first give the ILP formulation that captures that all triples
$\rt{(\al\be|\ga)}$ contained in $\mathbb{S}^* \subseteq \mathbb{S}$ are
displayed by a tree. A triple $\rt{(\al\be|\ga)}$ is displayed by a tree if
and only if there is an inner vertex $v_p$ such that $\al,\be\in L(v_p)$
and $\ga\notin L(v_p)$ and hence, iff $M_{\al p}=M_{\be p}=1\neq M_{\ga
p}=0$ for this cluster $p$. 

To this end, we define binary variables $N_{\al\be, p}$ so that
  $N_{\al\be, p}=1$ iff $\al,\be\in L(v_p)$ for all $\{\al,\be\}\in
  \binom{\Spe}{2}$ and $p=1,\dots, |\Spe|-2$. This condition is captured by
  the constraint:
\begin{align}\label{ref{ilp:Nclus}}\tag{\ref{ilp:Nclus}}
  0\leq & M_{\al p} + M_{\be p} - 2 N_{\al\be, p}  \leq 1.
\end{align}
We still need to ensure that for each triple $\rt{(\al\be|\ga)}\in \mathbb
S^*$ there is at least one cluster $p$ that contains $\alpha$ and $\beta$
but not $\gamma$, i.e., $N_{\al\be, p}=1$ and $N_{\al\ga, p}=0$ and
$N_{\be\ga, p}=0$. For each possible triple $\rt{(\al\be|\ga)}$ we
therefore add the constraint
\begin{align}\label{ref{ilp:rep}}
  1 - |\Spe|(1- T^*_{\rt{(\al\be|\ga)}}) \leq 
  \sum_p N_{\al\be,p} - \frac{1}{2}  N_{\al\ga,p} -\frac{1}{2} N_{\be\ga,p}. 
  \tag{\ref{ilp:rep}}
\end{align}
To see that \eqref{ilp:rep} ensures $\al,\be\in L(v_p)$ and $\ga\notin
L(v_p)$ for each $\rt{(\al\be|\ga)}\in \mathbb S^*$ and some $p$, assume
first that $\rt{(\al\be|\ga)}\not\in \mathbb S^*$ and hence,
$T^*_{\rt{(\al\be|\ga)}}=0$. Then, $1 - |\Spe|(1- T^*_{\rt{(\al\be|\ga)}})
= 1 - |\Spe|$ and we are free in the choice of the variables
$N_{\al\be,p}$, $N_{\al\ga,p}$, and $N_{\be\ga,p}$.  Now assume that
$\rt{(\al\be|\ga)}\in \mathbb S^*$ and hence, $T^*_{\rt{(\al\be|\ga)}}=1$.
Then, $1 - |\Spe|(1- T^*_{\rt{(\al\be|\ga)}}) = 1$. This implies that at
least one variable $N_{\al\be,p}$ must be set to $1$ for some $p$. If
$N_{\al\be,p}=1$ and $N_{\al\ga,p}=1$, then constraint
\eqref{ref{ilp:Nclus}} implies that $M_{\al p} = M_{\be p} = M_{\ga p} =1$
and thus $N_{\be\ga,p}=1$. Analogously, if $N_{\al\be,p}=1$ and
$N_{\be\ga,p}=1$, then $N_{\al\ga,p}=1$. It remains to show that there is
some cluster $p$ with $N_{\al\be,p}=1$ and
$N_{\al\ga,p}=N_{\be\ga,p}=0$. Assume, for contradiction, that for none of
the clusters $p$ with $N_{\al\be,p}=1$ holds that
$N_{\al\ga,p}=N_{\be\ga,p}=0$. Then, by the latter arguments all of these
clusters $p$ satisfy: $N_{\al\ga,p}=N_{\be\ga,p}=1$. However, this implies
that $ N_{\al\be,p} - \frac{1}{2} N_{\al\ga,p} -\frac{1}{2} N_{\be\ga,p} =
0$ for all $p$, which contradicts the constraint
\eqref{ilp:rep}. Therefore, if $T^*_{\rt{(\al\be|\ga)}}=1$, there must be
at least one cluster $p$ with $N_{\al\be,p}=1$ and
$N_{\al\ga,p}=N_{\be\ga,p}=0$ and hence, $M_{\al p} = M_{\be p}=1$ and
$M_{\ga p} =0$.

In summary the constraints above ensure that for the
maximal consistent triple set $\mathbb S^*$ of $\mathbb S$ and for each
triple $\rt{(\al\be|\ga)}\in \mathbb S^*$ exists at least one column $p$ in
the matrix $M$ that contains $\al$ and $\be$, but not $\ga$.  Note that for
a triple $\rt{(\al\be|\ga)}$ we do not insist on having a cluster $q$ that
contains $\ga$ but not $\al$ and $\be$ and therefore, we do not insist on
constructing singleton clusters. Moreover, there is no constraint that
claims that the set $\Spe$ is decoded by $M$. In particular, since we later
maximize the number of trivial columns in $M$ and since we do not gave ILP
constraints that insist on finding clusters $\Spe$ and $\{x\},\ x\in \Spe$,
these clusters will not be defined by $M$. However, these latter clusters
are clearly known, and thus, to decode the desired tree, we only require
that $M$ is a ``partial'' hierarchy, that is for every pair of clusters $p$
and $q$ holds $p\cap q\in \{p,q, \emptyset\}$.  In such case the clusters
$p$ and $q$ are said to be compatible.  Two clusters $p$ and $q$ are
incompatible if there are (not necessarily distinct) species
$\al,\be,\ga\in \Spe$ with $\al\in p\setminus q $ and $\be\in q\setminus
p$, and $\ga\in p\cap q$. In the latter case we would have $(M_{\al
  p},M_{\al q})=(1,0)$, $(M_{\be p},M_{\be q})=(0,1)$, $(M_{\ga p},M_{\ga
  q})=(1,1)$.  Here we follow the idea of \citet{chang2011ilp}, and use the
so-called three-gamete condition.  For each gamete $(\Gamma,\Lambda)\in
\{(0,1),(1,0)(1,1)\}$ and each column $p$ and $q$ we define a set of binary
variables $C_{p,q,\Gamma\Lambda}$. For all $\al\in \Spe$ and $p,q=1,\dots,
|\Spe|-2$ with $p\neq q$ we add
\begin{align}\tag{\ref{ilp:CM}}
C_{p,q,01}\geq &-M_{\al p}+M_{\al q}\\
C_{p,q,10}\geq &\ \ \ \ M_{\al p}-M_{\al q} \notag\\
C_{p,q,11}\geq &\ \ \ \ M_{\al p}+M_{\al q}-1 \notag
\end{align}
These constraints capture that $C_{p,q,\Gamma\Lambda}=1$ as long as
if $M_{\al p}=\Gamma$ and $M_{\al q}=\Lambda$ for some $\al\in \Spe$. 
To ensure that only compatible clusters are contained, we add
for each of the latter defined variable
\begin{align}\label{ref{ilp:comp}}
C_{p,q,01} + C_{p,q,10} + C_{p,q,11} \leq 2. \tag{\ref{ilp:comp}}
\end{align}
Hence the latter Equations \eqref{ilp:Nclus}-\eqref{ilp:comp} ensure,
that we get a ``partial'' hierarchy $M$, where only the singleton clusters
and the set $\Spe$ is missing, 

Finally we want to have for the maximal consistent triple sets
$\mathbb{S}^*$ of $\mathbb{S}$ the one that determines the least resolved
tree, i.e, a tree that displays all triples of $\mathbb{S}^*$ and has a
minimal number of inner vertices and makes therefore, the fewest
assumptions on the tree topology.  Since the number of leaves $|\Spe|$ in
the species tree $S$ is fixed and therefore the number of clusters is
determined by the number of inner vertices, as shown in the proof of Lemma
\ref{A:lem:nrC}, we can conclude that a minimal number of clusters results
in tree with a minimal number of inner vertices.  In other words, to find a
least resolved tree determined by the hierarchy matrix $M$, we need to
maximize the number of trivial columns in $M$, i.e., the number of columns
$p$ with $\sum_{\al\in \Spe} M_{\al p}=0$.

For this, we require in addition to the constraints
\eqref{ilp:Nclus}-\eqref{ilp:comp} for each $p=1,\dots,|\Spe|-2$ a binary
variable $Y_p$ that indicates whether there are entries in column $p$ equal
to $1$ or not. To infer that $Y_p=1$ whenever column $p$ is non-trivial we
add for each $p=1,\dots,|\Spe|-2$ the constraint
\begin{align}
  0&\leq Y_p|\Spe|  - \sum_{\al\in\Spe} M_{\al p}\leq |\Spe|-1 
  \tag{\ref{ilp:yp}}
\end{align}
If there is a ``1'' entry in column $p$ and $Y_p=0$ then, $Y_p|\Spe| -
\sum_{\al\in\Spe} M_{\al p}<0$, a contradiction.  If column $p$ is trivial
and $Y_p=1$ then, $Y_p|\Spe| - \sum_{\al\in\Spe} M_{\al p}=|\Spe|$, again a
contradiction.  Finally, in order to minimize the number of non-trivial
columns in $M$ and thus, to obtain a least resolved tree for $\mathbb S^*$
with a minimum number of inner vertices, 
we add the objective function
\begin{align}	\tag{\ref{ilp:minY}}
 \min\ &\sum_p Y_p
\end{align}
Therefore, we obtain for the maximal consistent subset
$\mathbb{S}^*\subseteq \mathbb{S}$ of species triples a ``partial''
hierarchy defined by $M$, that is, for all clusters $L(v_p)$ and $L(v_q)$
defined by columns $p$ and $q$ in $M$ holds $L(v_p) \cap L(v_q) \in
\{L(v_p), L(v_q), \emptyset\}$. The clusters $\Spe$ and $\{x\},\ x\in \Spe$
will not be defined by $M$. However, from these clusters and the clusters
determined by the columns of $M$ it is easily build the corresponding tree,
which by construction displays all triples in $\mathbb{S}^*$, see
\cite{sem-ste-03a,Dress:book}.
 
The latter ILP formulation requires $O(|\Spe|^3)$ variables and
constraints.

A least resolved tree with a minimum number of inner vertices \cite{Jansson:12} is not the
  only possible way to construct a species tree without spurious resolution.
  As an alternative approach one might consider trees $T$ that display
  all triples of $\mathbb{S}^*$ and at the same time minimize the number of
  additional triples $r\in \mathfrak{R}(T)\setminus\mathbb{S}^*$. Since the
  closure $\cl(\mathbb{S}^*)$ is displayed by all trees that display also
  $\mathbb{S}^*$, this task is equivalent to finding a tree with
  $\widehat{\mathbb{S}}:=\mathfrak{R}(T)$ that displays $\cl(\mathbb{S}^*)$
  and minimizes number of triples in $\widehat{\mathbb{S}}\setminus
  \cl(\mathbb{S}^*)$. Thus, $\widehat{\mathbb{S}}$ must be of minimum
  cardinality. To this end, we modify the ILP formulation
  \eqref{ilp:minY}-\eqref{ilp:comp}, and remove the objective function
  \eqref{ilp:minY}, constraint \eqref{ilp:yp} and omit the variables
  $Y_p$. Instead, we introduce the binary variables
  $\widehat{T}_{\rt{(\al\be|\ga)}}=1$ iff
  $\rt{(\al\be|\ga)}\in\widehat{\mathbb{S}}$. For each $p=1,\dots,|\Spe|-2$
  and all $\{\al,\be,\ga\}\in \binom {\mathcal S}{3}$ we add the
  constraints
  \begin{ILP} \label{ilp:tHat} 
    M_{\al p} + M_{\be p} + (1-M_{\ga p}) -
    \widehat{T}_{\rt{(\al\be|\ga)}} \leq 2.
  \end{ILP}
  Constraint \eqref{ilp:tHat} enforces that
  $\widehat{T}_{\rt{(\al\be|\ga)}}=1$ whenever there exists a cluster $p$
  with $\al,\be \in L(v_p)$ and $\ga \notin L(v_p)$, and hence
  $\rt{(\al\be|\ga)}\in\widehat{\mathbb{S}}$.  Finally, in order to
  minimize the number of triples in $\widehat{\mathbb{S}}$ we change the
  objective function \eqref{ilp:minY} to
  \begin{ILP} \label{ilp:minTHat} 
    \min\ &\sum_{\rt{(\al\be|\ga)}} \widehat{T}_{\rt{(\al\be|\ga)}}.
  \end{ILP}
  Constraints \eqref{ilp:Nclus} to \eqref{ilp:comp} remain unchanged.  This
  alternative ILP formulation requires $O(|\Spe|^3)$ variables and
  $O(|\Spe|^4)$ constraints.

If $\mathbb{S}^*$ is strictly dense, then both ILP formulations will
  result in the same binary tree as constructed using the \texttt{BUILD}
  algorithm.

\section{Implementation and Data Sets} 

\subsection{ILP Solver}

The ILP approach has been implemented using \textsc{IBM ILOG
  CPLEX{\texttrademark}} Optimizer 12.6 in the weighted version of the
maximum consistent triple set problem.
For each component of $G_\Theta$ we
check in advance if it is already a cograph. If this is not the case then
an ILP instance is executed, finding the closest cograph. In a similar
manner, we check for each resulting cograph whether it contains any paralogous
genes at all. If not, then the cograph is a complete graph and the
resulting gene tree would be a star, not containing any species triple information.
Hence, extracting the species triples is skipped. Triple extraction is done using an
polynomial time algorithm instead of the ILP formulation.
Although the connected components
of $G_\Theta$ are treated separately, some instances of the cograph editing
problem have exceptionally long computation times. We therefore exclude
components of $G_{\Theta}$ with more than 50 genes. In addition, we limit
the running time for finding the closest cograph for one disconnected
component to 30 minutes.  If an optimal solution for this component is not
found within this time limit, we use the best solution found so far.  The
other ILP computations are not restricted by a time limit.\smallskip

\subsection{Simulated Data}
To evaluate the ILP approach we use simulated and real-life data sets.
  Artificial data is created with the the method described in
  \citep{HHW+14} as well as the \texttt{Artificial Life Framework} (\ALF)
  \cite{Dalquen:12}. The first method generates explicit species/gene tree
  histories, from which the orthology relation is directly accessible.  All
  simulations are performed with parameters $1.0$ for gene duplication,
  $0.5$ for gene loss and $0.1$ for the loss rate, respectively increasing
  loss rate, after gene duplication. We do not consider cluster or genome
  duplications. \ALF\ simulates the evolution of sequences along a branch
  length-annotated species tree, explicitly taking into account gene
  duplication, gene loss, and horizontal transfer events.  To obtain
  bacteria-like data sets we adopted the procedure from \cite{DAAGD2013}: a
  tree of \emph{$\gamma$-proteobacteria} from the OMA project
  \cite{Altenhoff:11} was randomly pruned to obtain trees of moderate size,
  while conserving the original branch lengths.  All simulations are
  performed with parameters 0.005 for gene duplication/loss rate.  We do
  not consider cluster duplications/loss.

 The presented method heavily depends on the amount of duplicated
  genes, which, in turn, is depending on the number of analyzed genes per
  species.  Naturally, the question arose, how many genes, respectively
  gene families, are needed, to provide enough information to reconstruct
  accurate species trees, assuming a certain gene duplication rate.
  Therefore, we evaluate the precision of reconstructed trees with respect
  to the number of species and gene families. 100 species trees of size
  5, 10, 15, and 20 (\ALF\ only) leaves are generated.  For each tree,
  the evolution of ten to 100 (first simulation method) and 100 to 500
  (\ALF) gene families is simulated.  This corresponds for the first
  simulation method to $32.6\%$ (five species), $19.0\%$ (ten species), and
  $13.5\%$ (15 species) and for \ALF\ simulations $11.2\%$ (five species),
  $8.1\%$ (ten species), and $7.5\%$ (15 and 20 species) of all homologous
  pairs being paralogs (values determined from the simulations).
  Horizontal gene transfer and cluster duplication/loss were not
  considered.

The reconstructed trees are compared with the generated (binary)
  species trees.  Therefore, we use the software \texttt{TreeCmp}
  \cite{BGW-treeCmp:12} to compute distances for rooted trees based on
  Matching Cluster (MC), Robinson-Foulds (RC), Nodal Splitted (NS) and
  Triple metric (TT). The distances are normalized by the average distance
  between random Yule trees \cite{Yule:25}.  

In order to estimate the effects of noise in the empirical orthology
  relation we consider several forms of perturbations (i) insertion and
  deletion of edges in the orthology graph (homologous noise), (ii)
  insertion of edges (orthologous noise), (iii) deletion of edges
  (paralogous noise), and (iv) modification of gene/species assignments
  (xenologous noise). In the first three models each possible edge is
  modified with probability $p$.  Model (ii) simulates overprediction of
  orthology, while model (iii) simulates underprediction. Model (iv)
  retains the original orthology information but changes the associations
  between genes and their respective species with probability $p$. This
  simulates noise as expected in case of horizontal gene transfer. For each
  model we reconstruct the species trees of 100 simulated data sets with
  ten species and 100 gene families (first simulation method), respectively
  1000 gene families (\ALF).  As before, no horizontal gene transfer or
  cluster duplications/losses were simulated.  Noise is added with a
  probability $p \in \{0.05, 0.10, 0.15, 0.20, 0.25\}$.  

 Horizontal transfer is an abundant process in particular in
  procaryotes that may lead to particular types of errors in each step of
  our approach, see the theoretical discussion below. We therefore
  investigated the robustness of our approach against HGT as a specific
  type of perturbation in some detail. To this end, we simulate data sets
  of 1000 gene families, using \ALF, with a duplication/loss rate of
  $0.005$ and evolutionary rates $r \in \{0.0, 0.0025, 0.005, 0.0075\}$ for
  horizontal transfer.  Cluster duplications/losses, or horizontal
  transfers of groups of genes are not considered.  The simulation is
  repeated $100$ times for each combination of parameters.  From the
  simulated sequences, orthologous pairs of genes are predicted with
  \texttt{Proteinortho} \cite{Lechner:11a}, using an $E$-value threshold of
  $1e-10$ and similarity parameter of $0.9$. From this estimate of the
  orthology relation species trees are reconstructed.  

 The authors of \cite{DAAGD2013} observed that increasing HGT rates
  have only a minor impact on the recall of orthology prediction, while the
  precision drops significantly, i.e., orthology prediction tools tend to
  mis-predict xenology as orthology.  To evaluate the impact of noise
  solely coming from mis-predicting xenology as orthology, a second
  orthology relation is constructed from the same simulations.  This
  orthology relation only differs from the simulated orthology relation by
  all simulated xenologs being predicted as orthologs, i.e., all paralogs
  are correctly detected (\emph{perfect paralogy knowledge}), see \ref{fig:simXenology} (B).  
  Analogously, we evaluated the impact of noise
  solely coming from mis-predicting xenology as paralogy, i.e., all orthologs
  are correctly detected (\emph{perfect orthology knowledge}), see \ref{fig:simXenology} (C).  
  From these orthology relations, species trees
  are reconstructed with the ILP approach, and compared with the generated
  species trees, used for the simulation.
  
All simulations so far have been performed by computing a least
    resolved tree that minimized the cardinality of the vertex set, i.e.,
    the definition of \cite{Jansson:12}. Given a consistent set of triples
    $\mathbb{S}^*$ it is thus of interest to evaluate the influence of
    different choices of how a tree is inferred from the triple set. To
    this end we compare (i) the \texttt{BUILD} algorithm, (ii) least
    resolved trees with minimum number of vertices, and (iii) least resolved 
		trees that
    minimize the number of additional triples $r \notin cl(\mathbb{S}^*)$.
    As a consequence of Proposition \ref{pro:BinaryClDense} the three
    methods will produce the same tree whenever the tree constructed with
    \texttt{BUILD} is binary. This is nearly always the case when the
    target tree is binary. Therefore, we use \ALF\ to generate a
    duplication/loss history along a non-binary species tree.  As before,
    parameter values of $0.005$ are used for gene duplication and loss.
    Horizontal gene transfer and cluster duplication/loss were not
    considered here. The resulting orthology relation is perturbed with
    ``orthologous noise'' (insertion of edges) with probability $0.05$.
    Each data set was analyzed with the ILP pipeline with the three
    different tree building methods. The resulting trees are compared with
    each other as well as the input tree used for the simulation. The
    procedure is repeated 100 times using the same ternary species tree
	  that is given here in Newick tree format: 
			\begin{small}
     \texttt{(((SE001,SE002,SE003),SE004,SE005)}\texttt{,(SE006,}\texttt{(SE007,}\\\texttt{SE008,}
      \texttt{SE009),SE010),(SE011,SE012,(SE013,SE014,SE015)))}.  
			\end{small}

\subsection{Real-Life Data Sets} 

As real-life applications we consider two sets of eubacterial genomes.  The
set of eleven \emph{Aquificales} species studied in \cite{Lechner:14b}
covers the three families \emph{Aquificaceae}, \emph{Hydrogenothermaceae},
and \emph{Desulfurobacteriaceae}.  The species considered are the
\emph{Aquificaceae}: \emph{Aquifex aeolicus} VF5
(NC\textunderscore000918.1, NC\textunderscore001880.1),
\emph{Hydrogeni\-virga sp.} 128-5-R1-1 (ABHJ00000000.1),
\emph{Hydrogenobacter thermophilus} TK-6 (NC\textunderscore013799.1),
\emph{Hydrogenobaculum sp.} \linebreak Y04AAS1 (NC\textunderscore
011126.1), \emph{Thermocrinis albus} DSM 14484 (NC\textunderscore013894.1),
\emph{Thermocrinis ruber} DSM 12173 (CP007028.1), the
\emph{Hydrogenothermaceae}: \emph{Persephonella marina} EX-H1 \linebreak
(NC\textunderscore012439.1, NC\textunderscore012440.1),
\emph{Sulfuri\-hydrogenibium sp.} \linebreak YO3AOP1 (NC\textunderscore010730.1)
\emph{Sulfurihydrogenibium azorense} Az-Fu1 
(NC\textunderscore012438.1), and the \emph{Desulfurobacteriaceae}:
\emph{Desulfobacterium thermolithotrophum} DSM 11699
(NC\textunderscore015185.1), and \emph{Thermovibrio ammonificans} HB-1
(NC\textunderscore014917.1, \linebreak NC\textunderscore014926.1).

A larger set of 19 \emph{Enterobacteriales} was taken from RefSeq:
\emph{Enterobacteriaceae} family: \emph{Cronobacter sakazakii} ATCC BAA-894
(NC\textunderscore009778.1, NC\textunderscore009779.1,
NC\textunderscore009780.1), \emph{Enterobacter aerogenes} KCTC 2190
(NC\textunderscore015663.1), \emph{Enterobacter cloacae} ATCC 13047
(NC\textunderscore014107.1, NC\textunderscore014108.1,
NC\textunderscore014121.1), \emph{Erwinia amylovora} ATCC 49946
(NC\textunderscore013971.1, NC\textunderscore013972.1,
NC\textunderscore013973.1), \emph{Escherichia coli} K-12 substr DH10B
(NC\textunderscore010473.1), \emph{Escherichia fergusonii} ATCC 35469
(NC\textunderscore011740.1, NC\textunderscore011743.1), \emph{Klebsiella
  oxytoca} KCTC 1686 (NC\textunderscore016612.1), \emph{Klebsiella
  pneu\-moniae} 1084 (NC\textunderscore018522.1),
\emph{Proteus mirabilis} BB2000 (NC\textunderscore022000.1),
\emph{Salmonella bongori} Sbon 167 (NC\textunderscore021870.1,
NC\textunderscore021871.1), \emph{Salmonella enterica} serovar Agona SL483
(NC\textunderscore011148.1, NC\textunderscore011149.1),
\emph{Salmonella typhimurium} DT104 (NC\textunderscore022569.1,
NC\textunderscore022570.1), \emph{Serratia marcescens} FGI94
(NC\textunderscore020064.1), \emph{Shigella boydii} Sb227
(NC\textunderscore007608.1, NC\textunderscore007613. 1), \emph{Shigella
  dysenteriae} Sd197 (NC\textunderscore007606.1, NC\textunderscore007607.1,
NC\textunderscore009344.1), \emph{Shigella flexneri} 5 str 8401
(NC\textunderscore008258.1), \emph{Shigella sonnei} Ss046
(NC\textunderscore007384.1, NC\textunderscore007385.1,
NC\textunderscore009345.1, NC\textunderscore009346.1,
NC\textunderscore009347.1), \emph{Yersinia pestis} Angola
(NC\textunderscore010157. 1, NC\textunderscore010158.1,
NC\textunderscore010159.1), and \emph{Yersinia pseudotuberculosis} IP 32953
(NC\textunderscore006153.2, NC\textunderscore006154.1,
NC\textunderscore006155.1).

\subsection{Estimation of the Input Orthology Relation}

An initial estimate of the orthology relation is
computed with \texttt{Proteinortho} \cite{Lechner:11a} from all the
annotated proteins using an $E$-value threshold of $1e-10$ and similarity
parameter of $0.9$.
Additionally, the genomes of all species were re-\texttt{blast}ed to detect
homologous genes not annotated in the \texttt{RefSeq}.
In brief, \texttt{Proteinortho} implements a modified pair-wise best hit
strategy starting from \texttt{blast} comparisons. It first creates a graph
consisting of all genes as nodes and an edge for every blast hit with an
$E$-value above a certain threshold. In a second step edges between two genes
$a$ and $b$ from different species are removed if a much better blast hit is
found between $a$ and a duplicated gene $b'$ from the same species as $b$.
Finally, the graph is filtered with spectral partitioning to result in
disconnected components with a certain minimum algebraic connectivity.

The resulting orthology graph usually consists of several pairwise
disconnected components, which can be interpreted as individual gene
families. Within these components there may exist pairs of genes having
\texttt{blast} $E$-values worse than the threshold so that these nodes are
not connected in the initial estimate of $\Theta$. Thus, the input data have
a tendency towards underprediction of orthology in particular for distant
species. Our simulation results suggest that the ILP approach handles
overprediction of orthology much better. We therefore re-add edges that
were excluded because of the $E$-value cut-off only within connected
components of the raw $\Theta$ relation.

\subsection{Evaluation of Phylogenies} 

For the analysis of simulated data we compare the reconstructed trees with
the trees generated by the simulation.  To this end we computed the
four commonly used distances measures for rooted trees, Matching Cluster
(MC), Robinson-Foulds (RC), Nodal Splitted (NS) and Triple metric (TT), as
described by \citet{BGW-treeCmp:12}.

The MC metric asks for a minimum-weight one-to-one matching between the
internal nodes of both trees, i.e., the clusters $C_1$ from tree $T_1$ with
the clusters $C_2$ from tree $T_2$.  For a given one-to-one matching the MC
tree distance $d_{MC}$ is defined as the sum of all weights
$h_C(p_1,p_2)=|L(p_1) \setminus L(p_2) \cup L(p_2) \setminus L(p_1)|$ with
$p_1 \in C_1$ and $p_2 \in C_2$.  For all unmatched clusters $p$ a weight
$|L(p)|$ is added.
The RC tree distance $d_{RC}$ is equal to the number of different clusters
in both trees divided by 2.
The NS metric computes for each tree $T_i$ a matrix $l(T_i)=(l_{xy})$ with
$x,y \in L(T_i)$ and $l_{xy}$ the length of the path from $\lca(x,y)$ to
$x$.  The NS tree distance $d_{NS}$ is defined as the $L^2$ norm of these
matrices, i.e., $d_{NS} = \Vert l(T_1)-l(T_2)\Vert_2$.
The TT metric is based on the set of triples $\mathfrak{R}(T_i)$ displayed
by tree $T_i$.  For two trees $T_1$ and $T_2$ the TT tree distance is equal
to the number of different triples in respective sets $\mathfrak{R}(T_1)$
and $\mathfrak{R}(T_2)$.

The four types of tree distances are implemented in the software
\texttt{TreeCmp} \cite{BGW-treeCmp:12}, together with an option to compute
normalized distances. Therefore, average distances between
random Yule trees \cite{Yule:25} are provided for each metric
and each tree size from 4 to 1000 leaves.
These average distances are used for normalization, resulting
in a value of 0 for identical trees and a value of approximately
1 for two random trees. Note, however, distances greater 1 are also possible.

For the trees reconstructed from the real-life data sets we compute a
support value $s \in [0,1]$, utilizing the triple weights
$w\rt{(\al\be|\ga)}$ from Eq.~\eqref{ilp:wmax}. Precisely,
\setcounter{equation}{1}
\begin{equation}
s = \frac{\sum_{\rt{(\al\be|\ga)} \in \mathbb S^*} 
  w\rt{(\al\be|\ga)}}{\sum_{\rt{(\al\be|\ga)} \in \mathbb S^*} 
  w\rt{(\al\be|\ga)}+w\rt{(\al\ga|\be)}+w\rt{(\be\ga|\al)}}
\end{equation}
The support value of a reconstructed tree indicates how often the triples
from the computed maximal consistent subset $\mathbb S^*$ were obtained
from the data in relation to the frequency of all obtained triples.
It is equal to 1 if there was no ambiguity in the data.
Values around 0.33 indicate randomness.

In a similar way, we define support values for each subtree $T(v)$ of the
resulting species tree $T$.  Therefore, let $S^v = \{\rt{(\al\be|\ga)} \in
\mathfrak{R}(T) | \al, \be \in L(v), \ga \notin L(v)\}$ be the subset of
the triples displayed by $T$ with the two closer related species being
leaves in the subtree $T(v)$ and the third species not from this
subtree. Then, the subtree support is defined as:
\begin{equation}
  s_v = \frac{\sum_{\rt{(\al\be|\ga)} \in \mathbb S^v} 
    w\rt{(\al\be|\ga)}}{\sum_{\rt{(\al\be|\ga)} \in \mathbb S^v} 
    w\rt{(\al\be|\ga)}+w\rt{(\al\ga|\be)}+w\rt{(\be\ga|\al)}}
\end{equation}
Note that $S^v$ only contains triples that support a subtree with leaf set
$L(v)$.  Therefore, the subtree support indicates how often triples are
obtained supporting this subtree in relation to the frequency of all
triples supporting the existence or non-existence of this subtree.

In addition, bootstrap trees are constructed for each data set, using two
different bootstrapping approaches.  (i) bootstrapping based on components,
and (ii) bootstrapping based on triples.  Let $m$ be the number of pairwise
disconnected components from the orthology graph $G_{\Theta^*}$, $n_i$ the
number of species triples extracted from component $i$, and
$n=\sum_{i=1}^{m} n_i$.  In the first approach we randomly select $m$
components with repetition from $G_{\Theta^*}$.  Then we extract the
respective species triples and compute the maximal consistent subset and
least resolved tree.  In the second approach we randomly select $n$ triples
with repetition from $\mathbb S$.  Each triple $\rt{(\al\be|\ga)}$ is
chosen with a probability according to its relative frequency
$w\rt{(\al\be|\ga)} / n$. From this set the maximal consistent subset and
least resolved tree is computed.  Bootstrapping is repeated $100$
times. Majority-rule consensus trees are computed with the software
\texttt{CONSENSE} from the \texttt{PHYLIP} package.

\section{Additional Results}

\subsection{Robustness against Noise from Horizontal Gene Transfer}
\label{sec:insens}

 Horizontal gene transfer (HGT) is by far the most common deviation
  from vertical inheritance, e.g.\ \cite{Grilli:14}. The key problem with
  HGT in the context of orthology prediction is that pairs of genes that
  derive from a speciation rather than a duplication event may end up in
  the same genome. Since pairs of genes in the same genome are classified
  as paralogs by the initial orthology detection heuristics and
  subsequently by ILP constraint \ref{ilp:forbid_E} during cograph
  editing. Such \emph{pseudo-orthologous} pairs can lead to a misplaced
  node with an incorrect event label in the cotree.  This may, under
  some circumstances, lead to the inference of false species triples, see
  Figure \ref{fig:HGT-wrongTriples}. Note, the latter problem still remains
  even if we would have detected all events on the gene tree correctly but
  use the triple sets $\mathbb{G}$ and $\mathbb{S}$ without any additional
  restrictions. Again Fig.~\ref{fig:HGT-wrongTriples} serves as an example.
  Therefore, it is of central interest  to understand in more detail the relation 
	between symbolic representations, reconciliation maps 
	and triple sets that take also HGT into account, which might solve
	this problem.

  When all true paralogs are known, we obtain surprisingly accurate species
  tree, see Figure \ref{fig:simXenology} (B). 
  The species trees reconstructed from a perfect orthology relation are somewhat 
  less accurate, see Figure \ref{fig:simXenology} (C). 
  The most pressing practical
  issue, therefore, is to identify true paralogs, i.e., pairs of genes that
  originate from a duplication event and subsequently are inherited only
  vertically. In addition, phylogeny-free
  methods to identify xenologs e.g.\ based on sequence composition
  \cite{Jaron:14,Tsirigos:05} are a promising avenue for future work to
  improve the initial estimates of orthology and paralogy.

\subsection{Simulated Data} 
The results for simulated data sets with a varying number of independent
gene families suggest, that a few hundred gene families are sufficient to
contain enough information for reconstructing proper phylogenetic species
trees. The reconstructions for data sets generated with \ALF\ need
  much more gene families to obtain a similar accuracy, as compared to
  simulations with the first simulation method.  This can be explained by
  the fact that the simulations of the first method resulted in a higher
  amount of paralogs, ranging from $13.5\%$ to $32.6\%$, compared to the
  \ALF\ simulations ($7.5\%$ to $11.2\%$).  Another reason is that due to
  the construction of the gene trees, used for \ALF\ simulations, the
  distribution of branch lengths, and hence, the distribution of
  duplications among the species tree, is very heterogeneous.  The average
  percentage of short branches (for which less than 1 duplication is
  expected, using a duplication rate of 0.005 and $n$ gene families) is
  ranging from $11.3\%$ ($5$ species, $500$ gene families) to $33.6\%$
  ($20$ species, $100$ gene families).  Note, that the lack of duplications
  leads to species trees that are not fully resolved, and hence have a
  larger distance to the generated trees used for the simulation.  Figures
\ref{fig:simFamilyFull} (first simulation method) and
\ref{fig:simFamilyFullALF} (\ALF\ simulations) show boxplots for the four
tree distances as a function of the number of independent gene families.

The complete results for the 2000 simulated data sets of 10 species and
100, resp. 1000 gene families with a varying amount of noise are depicted
in Figures \ref{fig:simNoiseFull} (first simulation method) and
\ref{fig:simNoiseFullALF} (\ALF).

The results for simulated data sets with horizontal gene transfer show
  that our method is very robust against noise introduced by horizontal
  gene transfer, which appears as mis-predicted orthology. Even xenologous
  noise of up to 39.5\% of the homologous pairs had only a minor impact on
  the obtained tree distances.  The triple support values $s$ for the
  reconstructed species trees, which ranges between 0.978 (HGT rate 0.0025)
  and 0.943 (HGT rate 0.0075).  This shows that only very few false species
  triples have been inferred.  However, these triples could be excluded
  during the computation of the maximal consistent subset, as they are
  usually dominated by the amount of correctly identified species triples.
  The small differences between generated and reconstructed species trees
  can be explained by the fact that the method forces homologous genes
  within the same species to be paralogous, although, due to horizontal
  gene transfer their lowest common ancestor can be a speciation event.
  This leads to the estimated orthology not being a cograph, introducing
  errors during the cograph editing step.  Figure \ref{fig:simXenology}
  shows boxplots for the tree distance as a function of the percentage of
  xenologous noise.

Computations with the different tree building methods (i.e.,
    (i) the \texttt{BUILD} algorithm, (ii) least
    resolved trees with minimum number of vertices, and (iii) least
    resolved trees that
    minimize the number of additional triples $r \notin cl(\mathbb{S}^*)$) 
		showed no influence of the used method on the
    resulting species trees. For none of the 100 data sets, the correct
    tree was reconstructed since -- expectedly -- the added noise
    introduces a few spurious triples that incorrectly resolve non-binary
    vertices. Indeed, all trees were (not fully resolved) refinements of the target tree.
    Furthermore, the triple distance between the reconstructed trees and the
    target tree was minute, with an average of $0.04$.
    Interestingly, for each of the 100 data sets, the respective three
    trees reconstructed with the three different tree building methods,
    were always identical. The consistent sets of triples obtained
    from the 100 data sets always identified a certain tree $T$. As
    demonstrated by Lemma \ref{lem:allEqual}, under this condition all
    three methods necessarily yield the same result. This finding suggests
    that our choice of the least resolved tree could be replaced by 
    \texttt{BUILD} as an efficient heuristic. For the analyzed data sets \texttt{BUILD}
    used less than 3 milliseconds of computation time and was approximately $10^5$ times faster
    than the other two methods.
    A head-to-head comparison of the two ILP methods shows that the method which minimizes
    the number of additional tripples (153 seconds on average) was approximately three times faster
    compared to the method which minimizes the number of vertices (496 seconds on average). 

\subsection{Real-life Data}
Figure \ref{fig:aquificales} depicts the phylogenetic tree of
\emph{Aquificales} species obtained from paralogy data in comparison to the
tree suggested by \citet{Lechner:14b}. The trees obtained from
bootstrapping experiments are given in Figure
\ref{fig:aquificales}.  
The majority-rule consensus trees for both bootstrapping approaches are
identical to the previously computed tree.  However, the bootstrap support
appears to be smaller next to the leaves.  This is in particular the case
for closely related species with only a few duplicated genes exclusively
found in one of the species.

Figure \ref{fig:enterobacteriales} depicts the phylogenetic tree of
\emph{Enterobacteriales} species obtained from paralogy data in comparison
to the tree from \texttt{PATRIC} database \citep{Wattam:13}. The trees
obtained from bootstrapping experiments are given in Figure
\ref{fig:enterobacteriales}.   
When assuming the \texttt{PATRIC} to be
correct, then the subtree support values appear to be a much more reliable
indicator, compared to the bootstrap values.

\subsection{Additional Comments on Running Time}

The CPLEX Optimizer is capable of solving instances with approximately a
few thousand variables. As the ILP formulation for cograph editing requires
$O(|\Gen|^2)$ many variables, we can solve instances with up to 100 genes
per connected component in $G_\Theta$.  However, for our computations we
limit the size of each component to 50 genes.  Furthermore, the ILP
formulations for finding the maximal consistent triple set and least
resolved species tree requires $O(|\Spe|^3)$ many variables. Hence, problem
instances of up to about 20 species can be processed.

Table \ref{tab:runtimeExtended} shows the running times for simulated and
real-life data sets for each individual sub-task. Note that the time used
for cograph editing is quite high, compared to the other sub-tasks. This is
due to the fact, that cograph editing if performed for each connected component
in $G_\Theta$ individually, and initializing the ILP solver is a relevant
factor. In the implementation we first perform a check, if for a
given gene family cograph editing has to be performed.
Triple extraction is performed with a polynomial time algorithm.
Another oddity is the extraordinary short running time for the computation of the 
maximal consistent subset of species triples in the
\emph{Enterobacteriales} data set. During the bootstrapping experiments for
this set much longer times were observed.

\clearpage
\begin{figure*}[tp]
  \centering
  \includegraphics[width=0.4\textwidth]{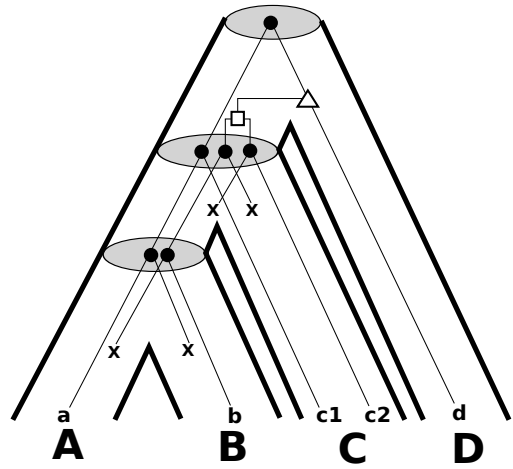}\\
\caption{Shown is a gene tree $T$ on $\Gen = \{a,b,c1,c2,d\}$ evolving along species tree $S$
on $\Spe=\{A,B,C,D\}$. 
In this scenario false gene triples in  $\mathbb{G}$ and thus, false species triples in  $\mathbb{S}$
are introduced, due to the HGT-event ($\triangle$) followed by a duplication event ($\Box$) and 
certain losses (\textbf{x}). 
Here, we obtain that $\rt{(bc2|a)}\in \mathbb{G}$ and thus $\rt{(BC|A)}\in \mathbb{S}$, contradicting
that $\rt{(AB|C)} \in \mathfrak{R}(S)$.}
\label{fig:HGT-wrongTriples}
\end{figure*}

\begin{figure*}[tp]
\begin{center}
\begin{minipage}{0.98\linewidth}
  \centering
  \includegraphics[width=0.95\textwidth]{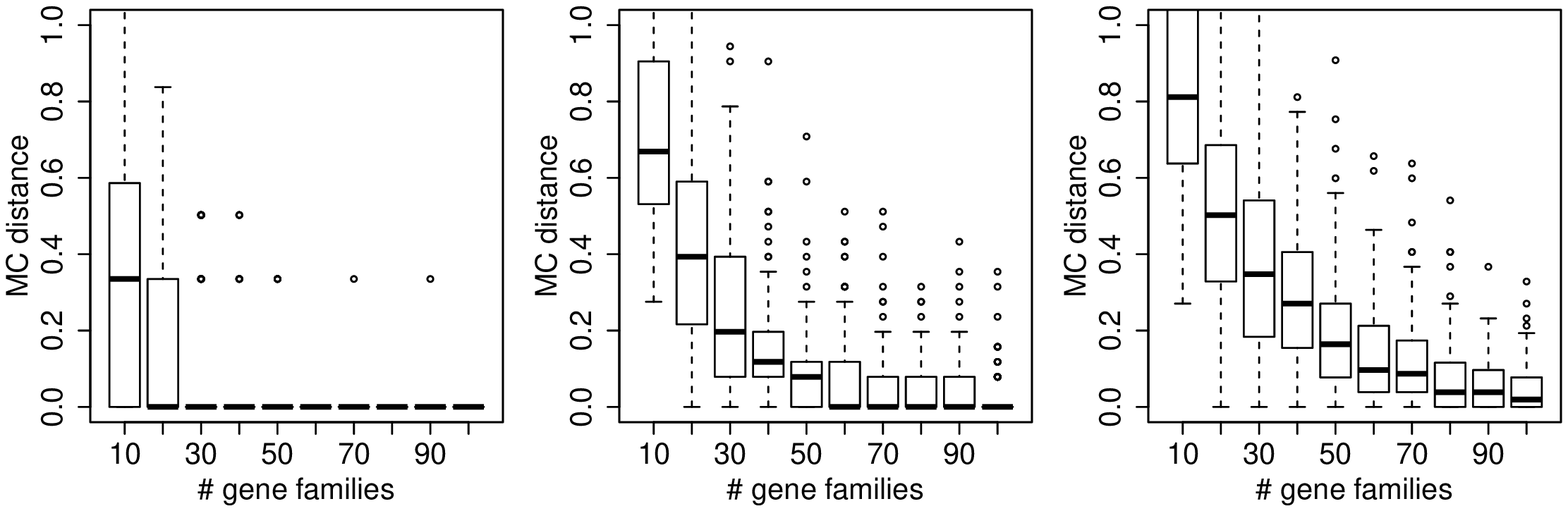}\\
\end{minipage}
\begin{minipage}{0.98\linewidth}
  \centering
  \includegraphics[width=0.95\textwidth]{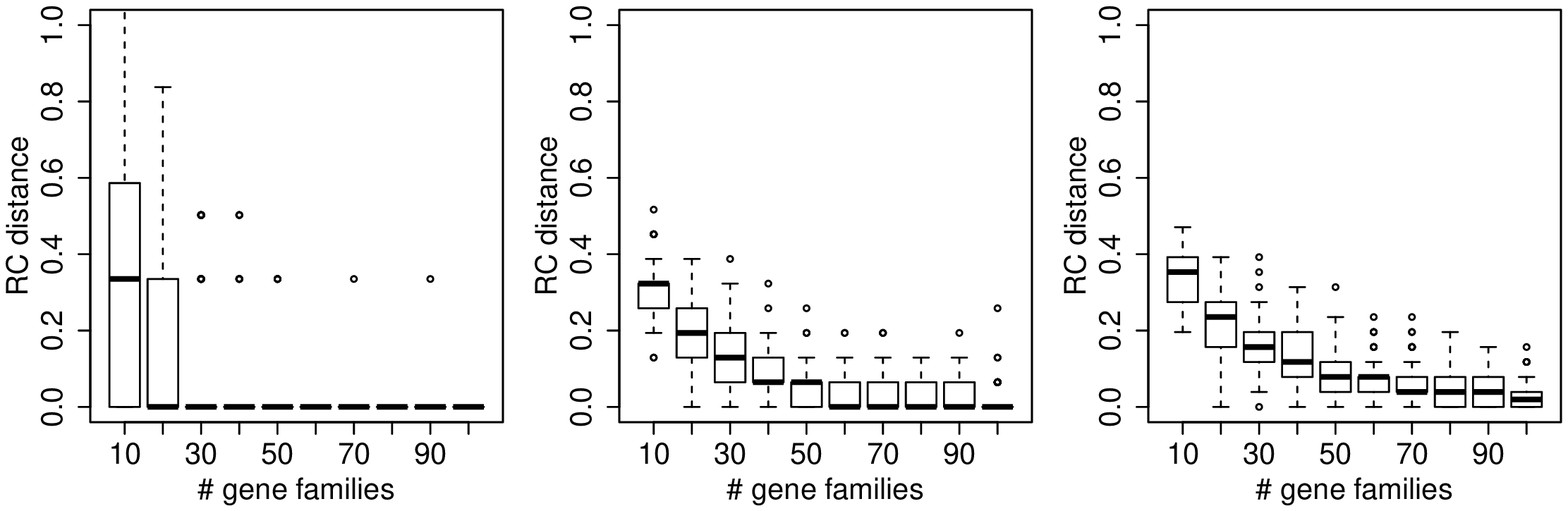}\\
\end{minipage}
\begin{minipage}{0.98\linewidth}
  \centering
  \includegraphics[width=0.95\textwidth]{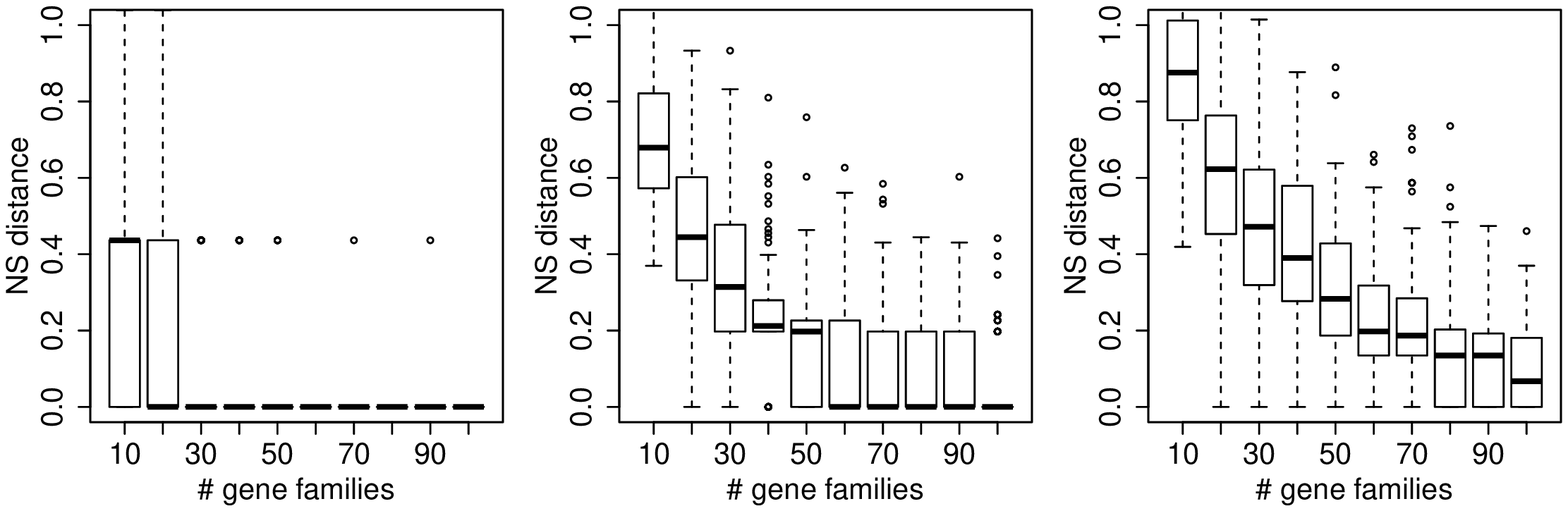}\\
\end{minipage}
\begin{minipage}{0.98\linewidth}
  \centering
  \includegraphics[width=0.95\textwidth]{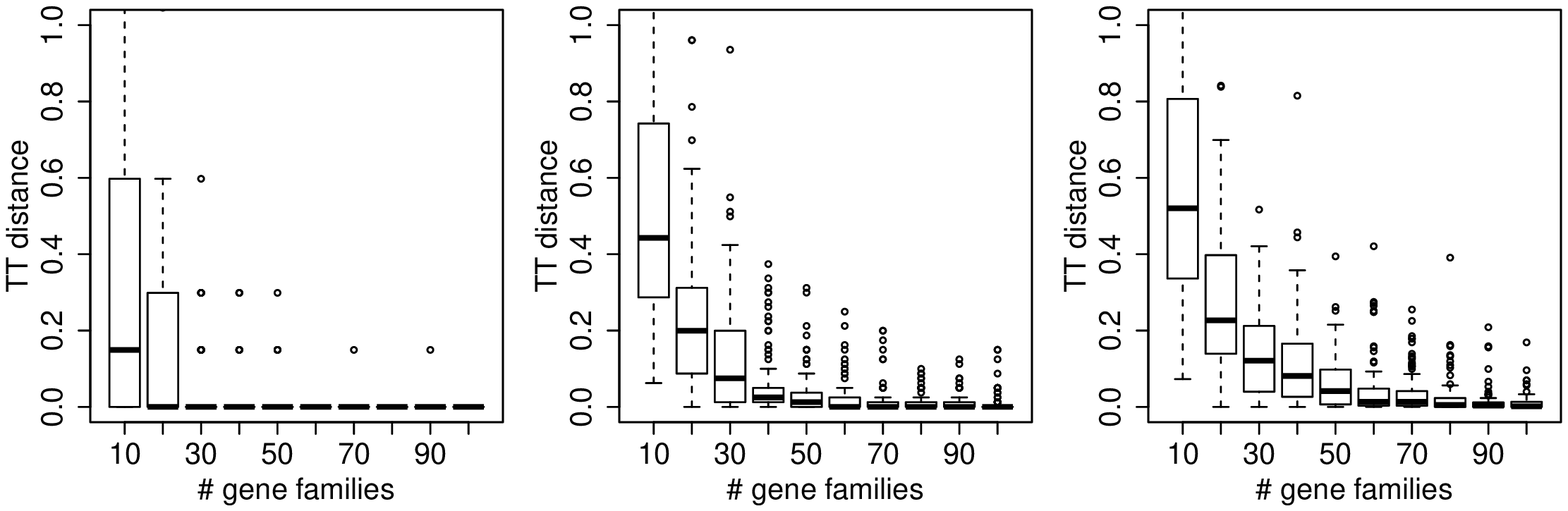}\\
\end{minipage}
\end{center}
\caption{Matching Cluster (MC), Robinson-Foulds (RC), Nodal Splitted (NS)
  and Triple metric (TT) tree distances of 100 reconstructed phylogenetic
  trees with (from left to right) five, ten, and 15 species and 10 to 100 gene
  families, each. Simulations are generated with first simulation method.}
\label{fig:simFamilyFull}
\end{figure*}

\begin{figure*}[tp]
\begin{center}
\begin{minipage}{0.98\linewidth}
  \centering
  \includegraphics[width=0.95\textwidth]{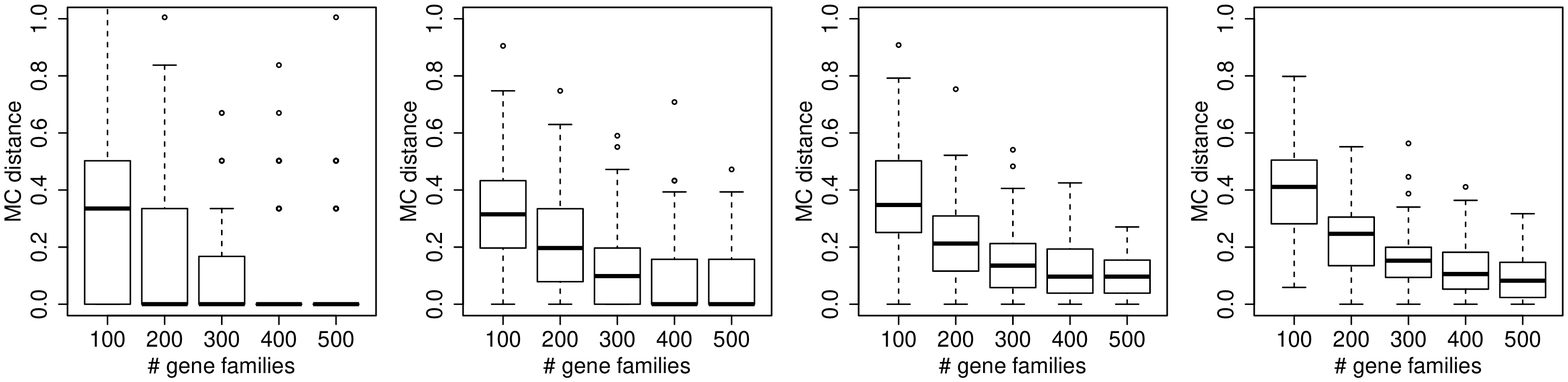}\\
\end{minipage}
\begin{minipage}{0.98\linewidth}
  \centering
  \includegraphics[width=0.95\textwidth]{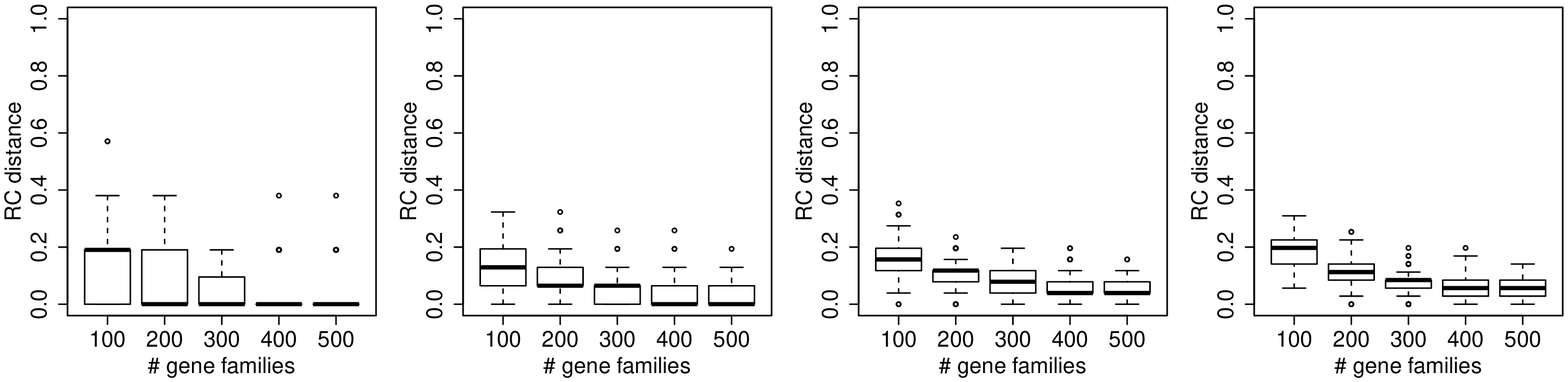}\\
\end{minipage}
\begin{minipage}{0.98\linewidth}
  \centering
  \includegraphics[width=0.95\textwidth]{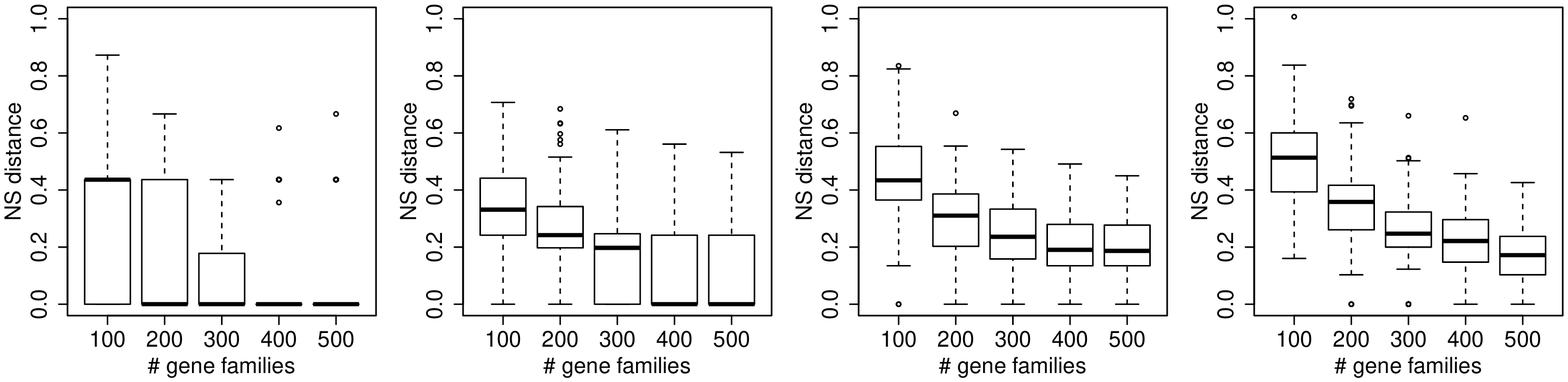}\\
\end{minipage}
\begin{minipage}{0.98\linewidth}
  \centering
  \includegraphics[width=0.95\textwidth]{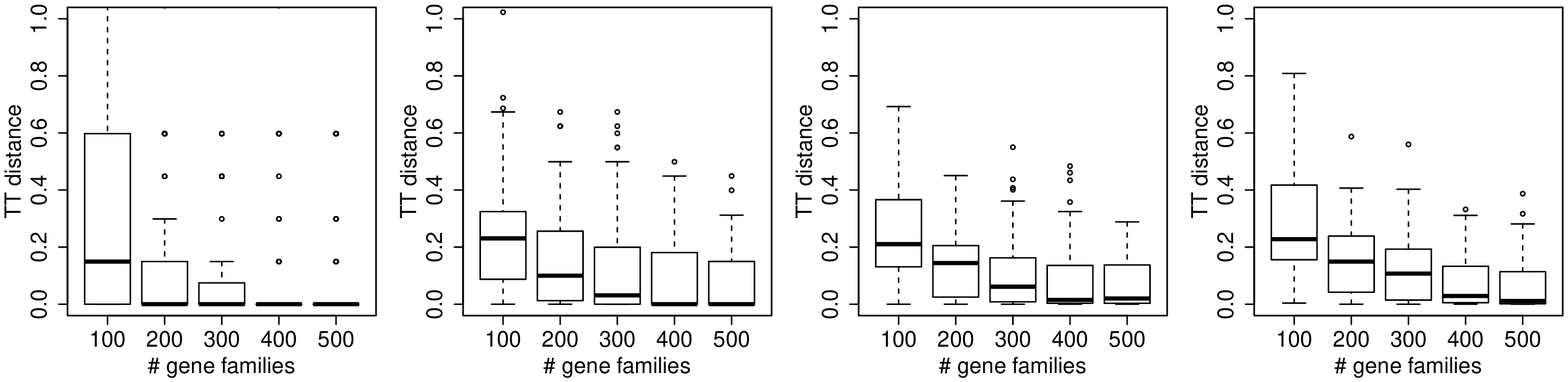}\\
\end{minipage}
\end{center}
\caption{Matching Cluster (MC), Robinson-Foulds (RC), Nodal Splitted (NS)
  and Triple metric (TT) tree distances of 100 reconstructed phylogenetic
  trees with (from left to right) five, ten, 15, and 20 species and 100 to 500 gene
  families, each. Simulations are generated with \ALF.}
\label{fig:simFamilyFullALF}
\end{figure*}

\begin{figure*}[htb]
\begin{center}
\begin{minipage}{0.98\linewidth}
  \centering
  \includegraphics[width=0.95\textwidth]{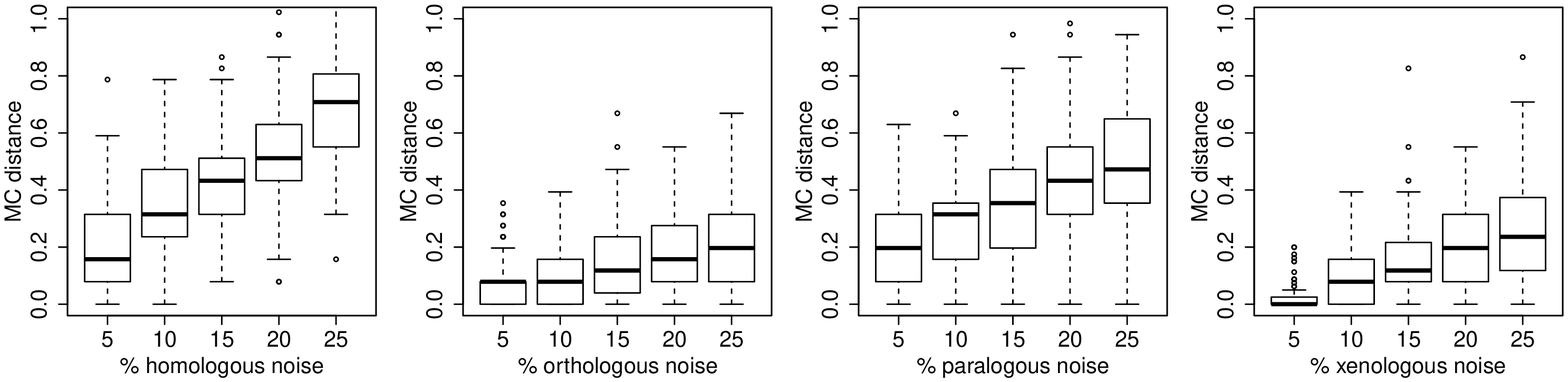}\\
\end{minipage}
\begin{minipage}{0.98\linewidth}
  \centering
  \includegraphics[width=0.95\textwidth]{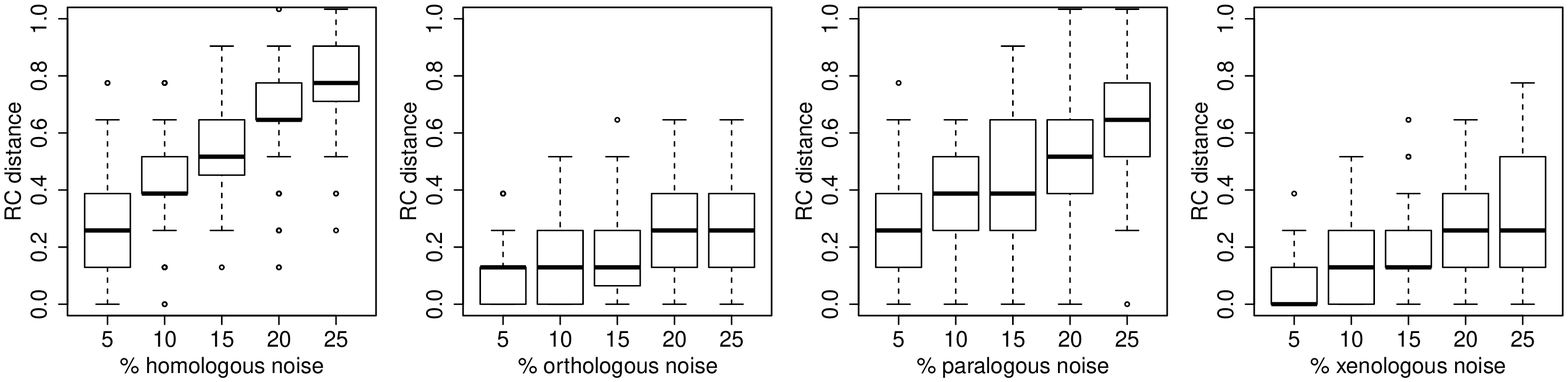}\\
\end{minipage}
\begin{minipage}{0.98\linewidth}
  \centering
  \includegraphics[width=0.95\textwidth]{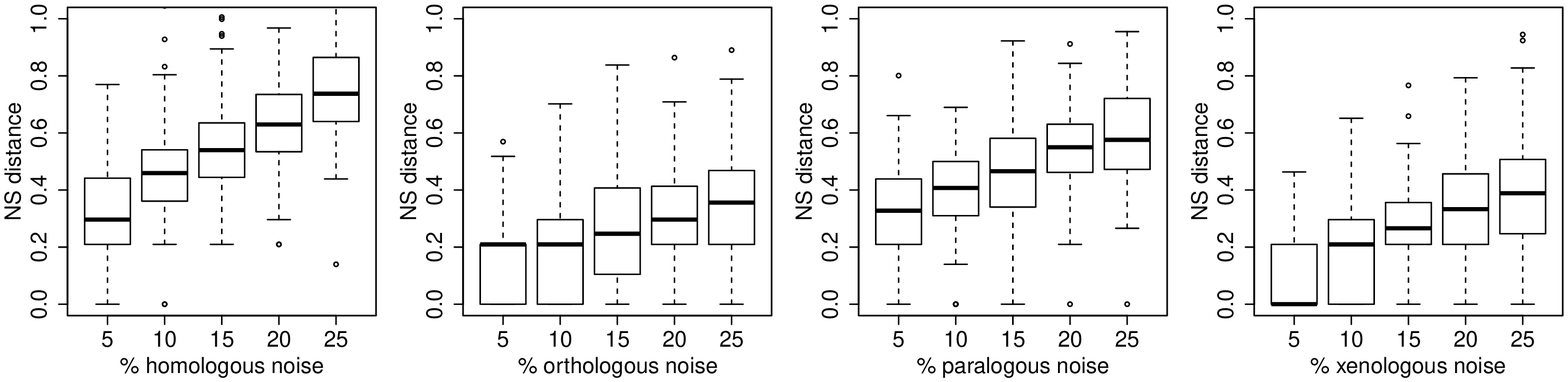}\\
\end{minipage}
\begin{minipage}{0.98\linewidth}
  \centering
  \includegraphics[width=0.95\textwidth]{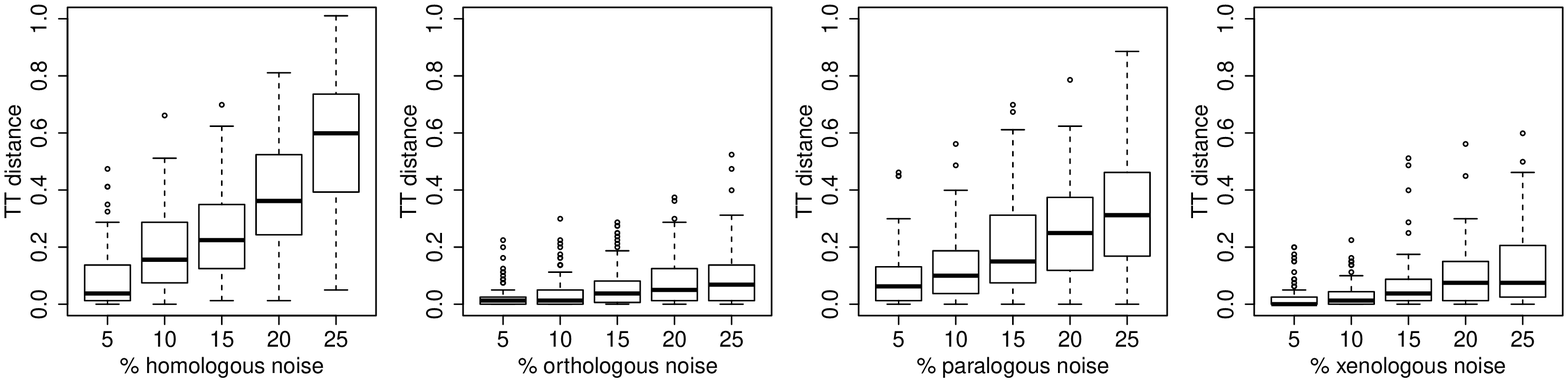}\\
\end{minipage}
\end{center}
\caption{Matching Cluster (MC), Robinson-Foulds (RC), Nodal Splitted (NS)
  and Triple metric (TT) tree distances of 100 reconstructed phylogenetic
    trees with ten species and 100 gene families generated with first simulation method.
    For each model noise was added with a probability of 0.05 to 0.25.}
\label{fig:simNoiseFull}
\end{figure*}

\begin{figure*}[htb]
\begin{center}
\begin{minipage}{0.98\linewidth}
  \centering
  \includegraphics[width=0.95\textwidth]{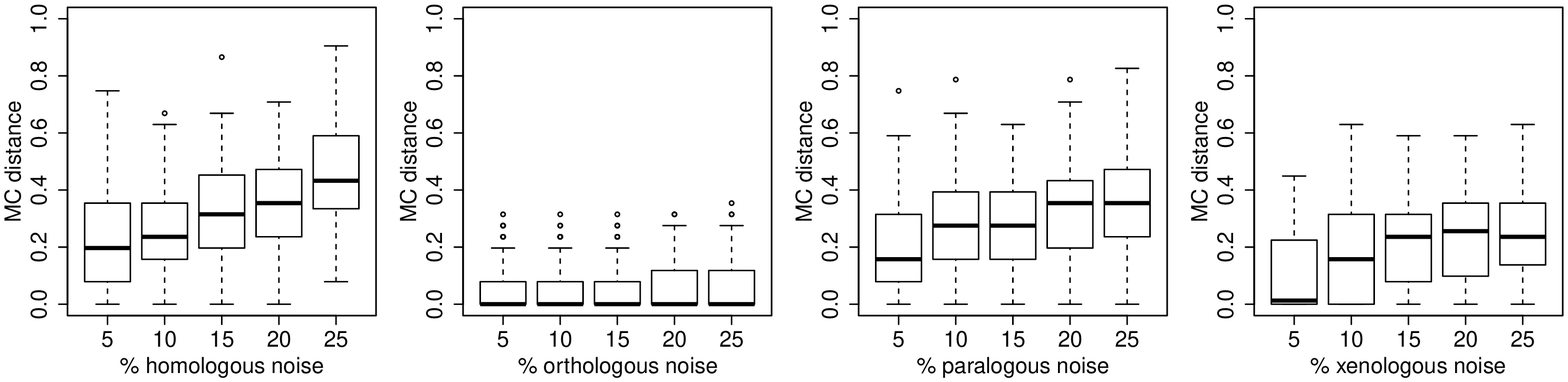}\\
\end{minipage}
\begin{minipage}{0.98\linewidth}
  \centering
  \includegraphics[width=0.95\textwidth]{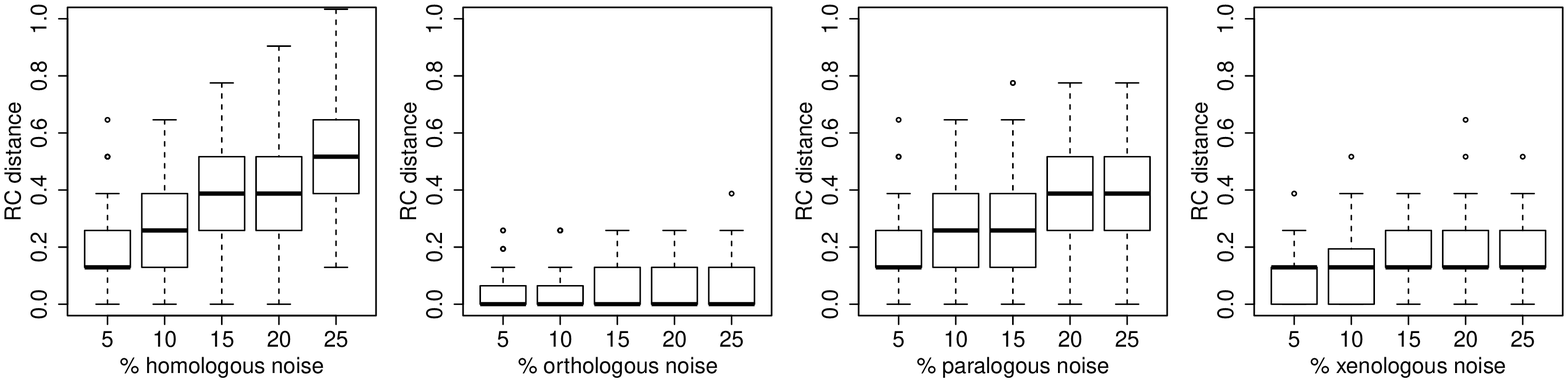}\\
\end{minipage}
\begin{minipage}{0.98\linewidth}
  \centering
  \includegraphics[width=0.95\textwidth]{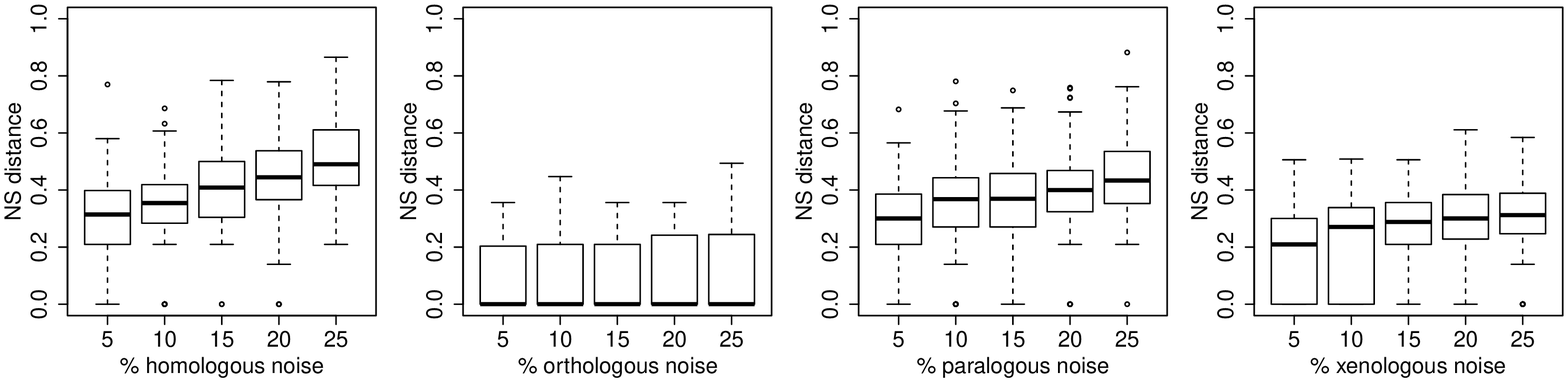}\\
\end{minipage}
\begin{minipage}{0.98\linewidth}
  \centering
  \includegraphics[width=0.95\textwidth]{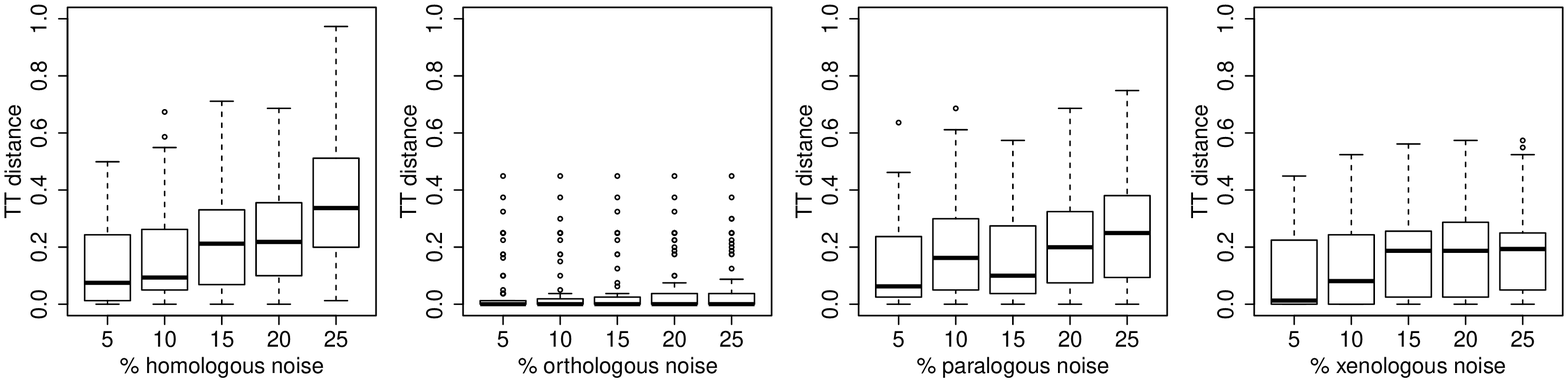}\\
\end{minipage}
\end{center}
\caption{Matching Cluster (MC), Robinson-Foulds (RC), Nodal Splitted (NS)
  and Triple metric (TT) tree distances of 100 reconstructed phylogenetic
    trees with ten species and 1000 gene families generated with \ALF.
    For each model noise was added with a probability of 0.05 to 0.25.}
\label{fig:simNoiseFullALF}
\end{figure*}

\begin{figure*}[tp]
\begin{center}
\begin{minipage}{0.98\linewidth}
  \centering
  \includegraphics[width=0.95\textwidth]{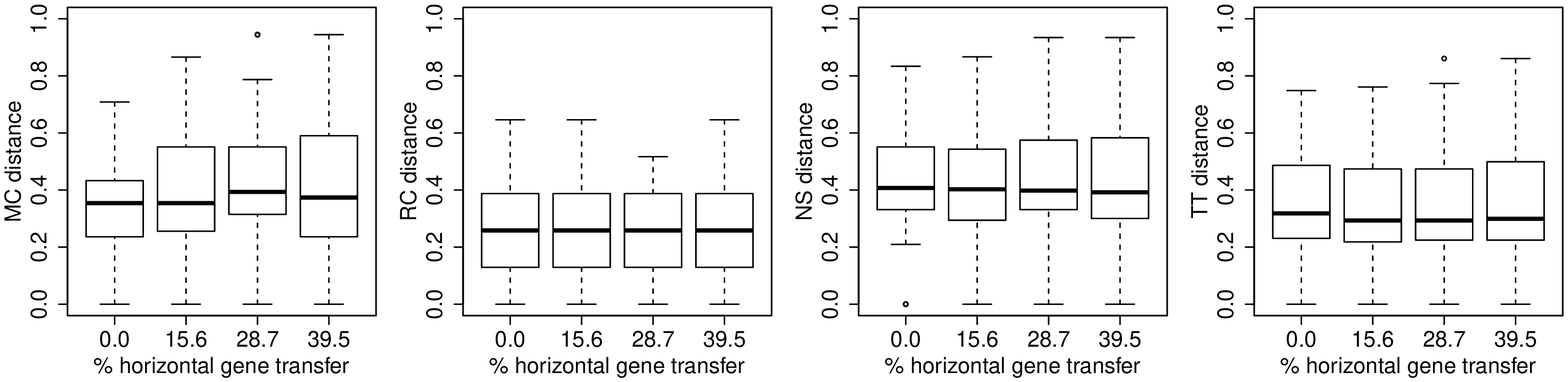}\\
  (A)
\end{minipage}
\vspace{2em}\\
\begin{minipage}{0.98\linewidth}
  \centering
  \includegraphics[width=0.95\textwidth]{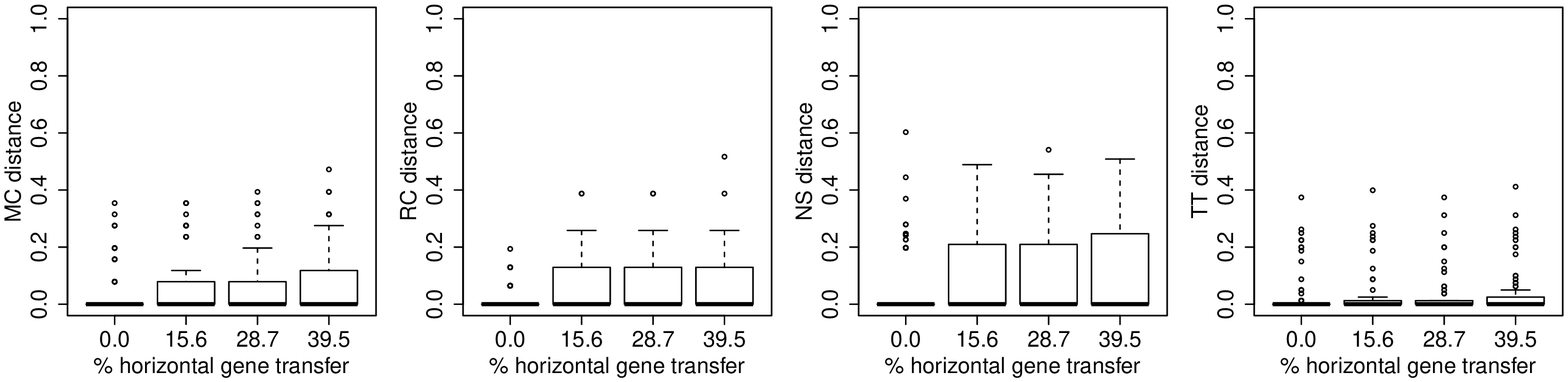}\\
  (B)
\end{minipage}
\vspace{2em}\\
\begin{minipage}{0.98\linewidth}
  \centering
  \includegraphics[width=0.95\textwidth]{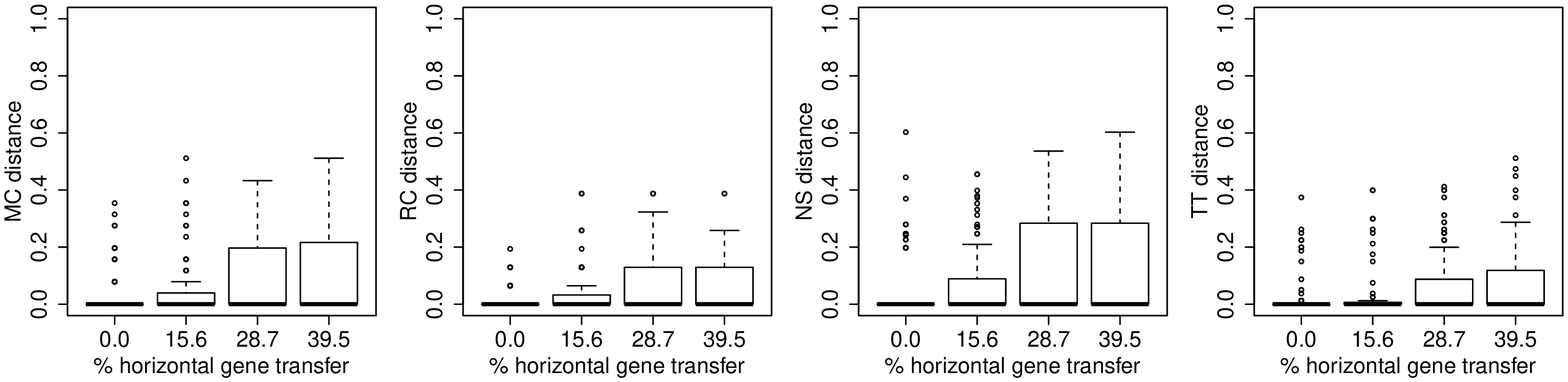}\\
  (C)
\end{minipage}
\end{center}
\caption{Matching Cluster (MC), Robinson-Foulds (RC), Nodal Splitted (NS)
  and Triple metric (TT) tree distances of 100 reconstructed phylogenetic
  trees with ten species. \ALF\ simulations are performed with duplication/loss rates of $0.005 \cong 6.1\%$
  and hgt rates of $0.0025$ to $0.0075$, resulting in xenologous noise between $0.0\%$ to $39.5\%$.
  Reconstructions are based on (A) \texttt{Proteinortho} orthology estimation, (B) perfect paralogy knowledge, and (C) perfect orthology knowledge.}
\label{fig:simXenology}
\end{figure*}

\begin{table}[htb]
  \caption{Running time in seconds on 2 Six-Core AMD Opteron\textsuperscript{\texttrademark} Processors with 2.6GHz for individual sub-tasks:
    \textbf{CE} cograph editing, 
    \textbf{MCS} maximal consistent subset of triples,
    \textbf{LRT} least resolved tree.}{
    \label{tab:runtimeExtended}
    \begin{tabular}{lrrlr}
      \hline \\
      \textbf{Data}   &
      \quad\textbf{CE}\quad     &  
      \quad\textbf{MCS}\quad    &
      \quad\textbf{LRT}\quad    & 
      \quad\textbf{Total\footnote{Total time includes triple extraction, parsing input, and writhing output files.}} \\ 
      Simulations\footnote{Average of 2000 simulations generated with \ALF, 10 species,
        1000 gene families.}
      & $125$\footnote{2,000,000 cographs, 41 not optimally solved within time limit of 30 min.}
      & $<1$ 
      & $<1$\footnote{In $95.95\%$ of the simulations the least resolved tree could be found using \texttt{BUILD}.} 
      & $126$\\
      \emph{Aquificales}  & $34$ 
      & $<1$ 
      & $<1\ (6)$\footnote{A unique tree was obtained using \texttt{BUILD}. Second value indicates running time with ILP solving enforced. \label{footnoteBUILDsupp}}
      & $34$\\
      \emph{Enterobacteriales} & $2673$ 
      & $2$\footnote{Note that the bootstrap computations had a much longer running time ($125$ sec. on average).}
      & $<1\ (1749)^\P$
      & $2676$ \\
\end{tabular}}
\end{table}

\begin{figure*}[htb]
\begin{center}
\includegraphics[width=0.4\textwidth]{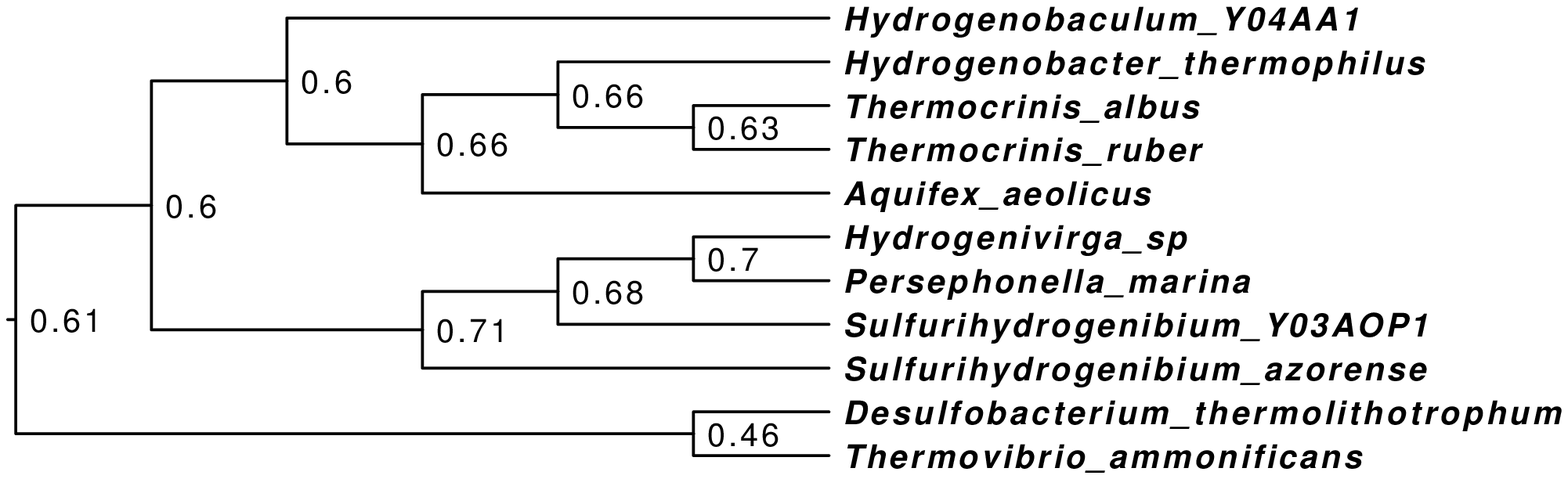}
\hspace{2em}
\includegraphics[width=0.4\textwidth]{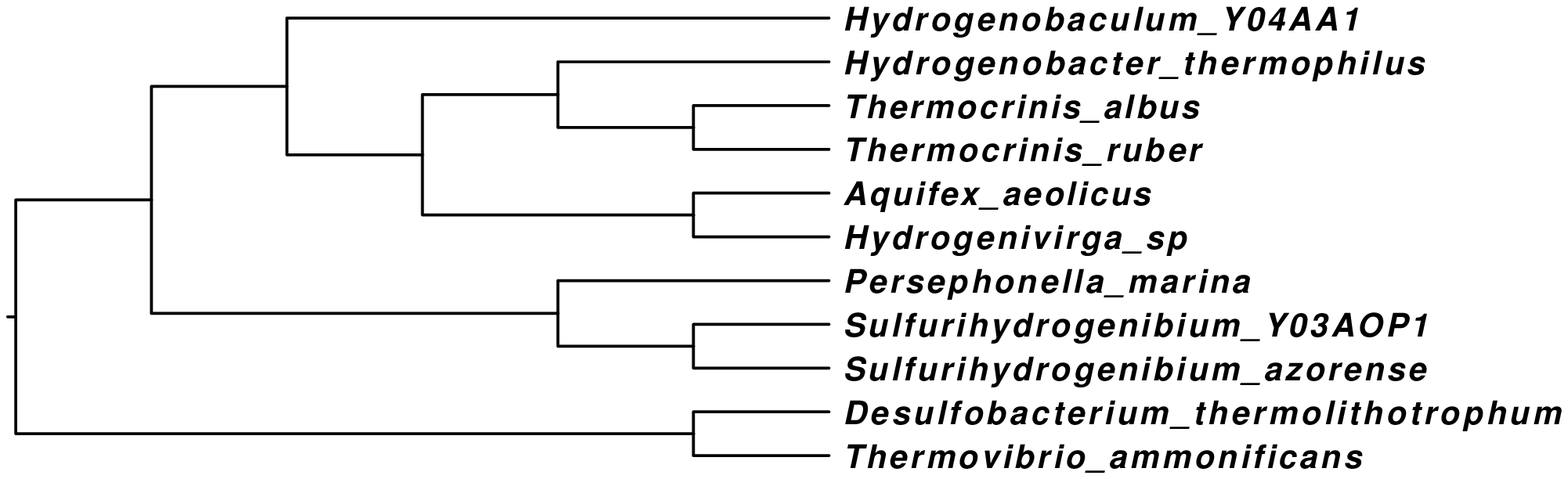}\\
\includegraphics[width=0.4\textwidth]{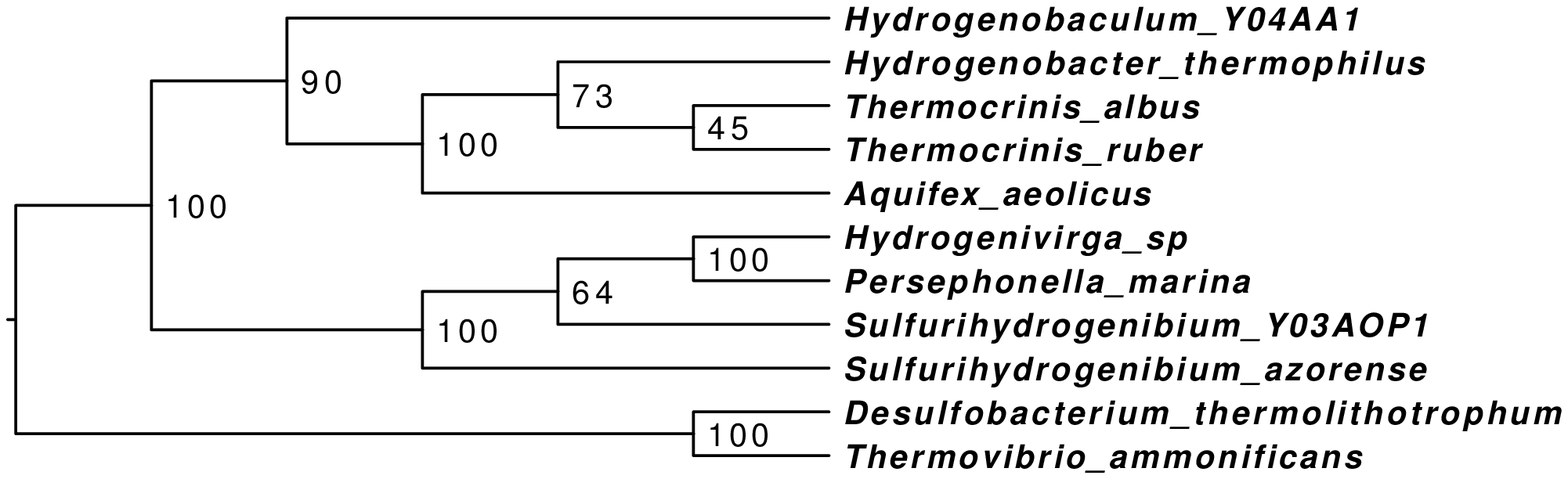}
\hspace{2em}
\includegraphics[width=0.4\textwidth]{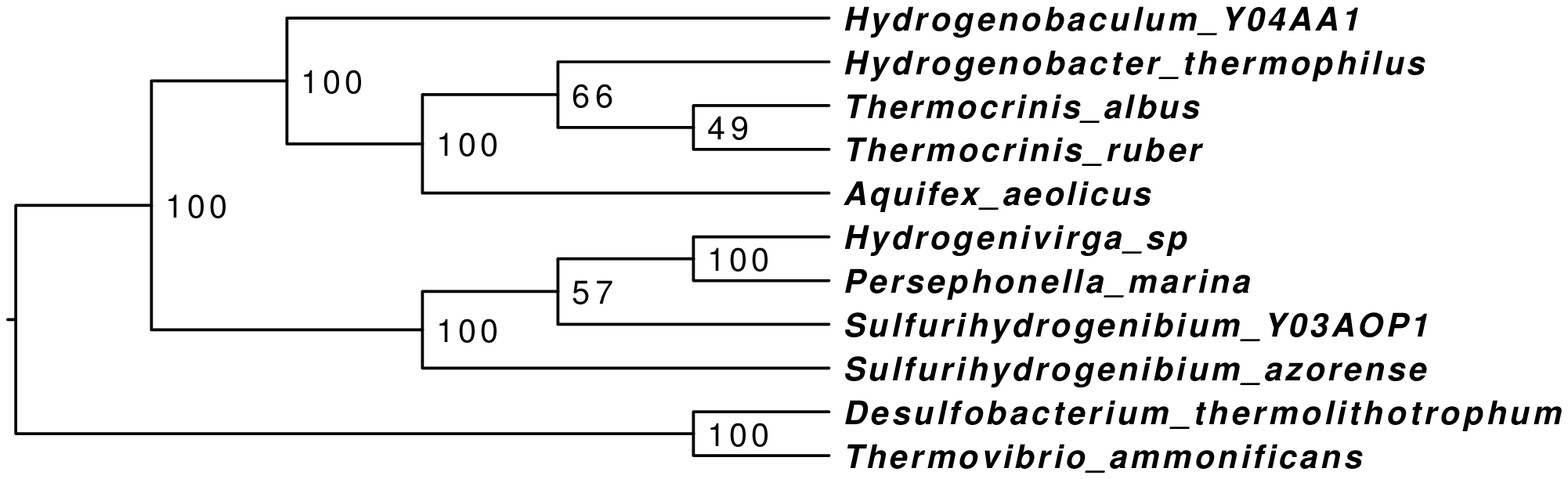}\\
\end{center}
\caption{Phylogenetic tree of eleven \emph{Aquificales} species. Top, L.h.s.:
  tree computed from paralogy data. Internal node labels indicate support
  of subtrees. Top, R.h.s.: reference tree from \citet{Lechner:14b}.
  Cograph-based (bottom, l.h.s.) and triple-based (bottom, r.h.s.) bootstrapping
  trees of the eleven \emph{Aquificales} species.
}
\label{fig:aquificales}
\end{figure*}

\begin{figure*}[htb]
\begin{center}
\includegraphics[width=0.4\textwidth]{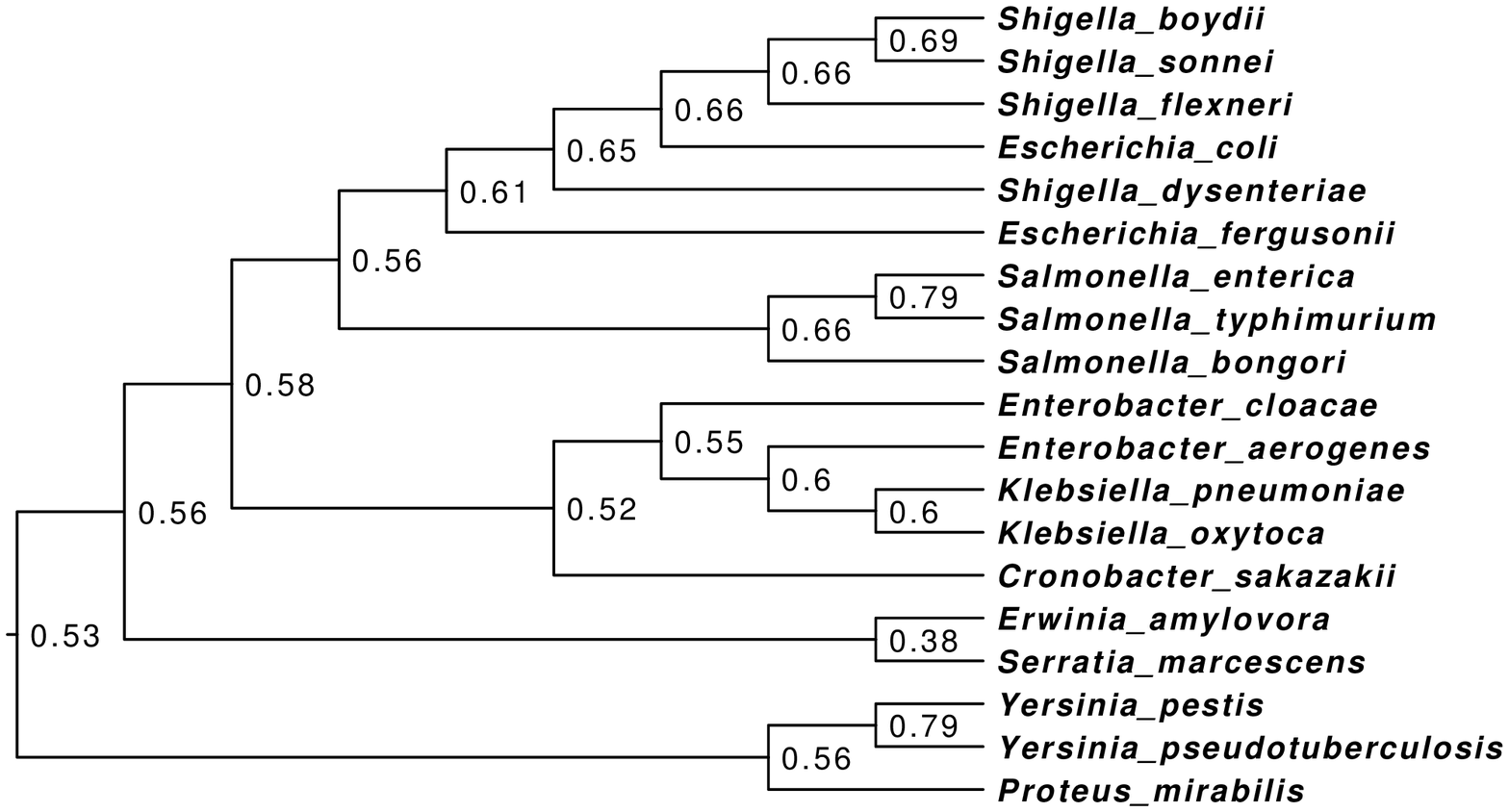}
\hspace{2em}
\includegraphics[width=0.4\textwidth]{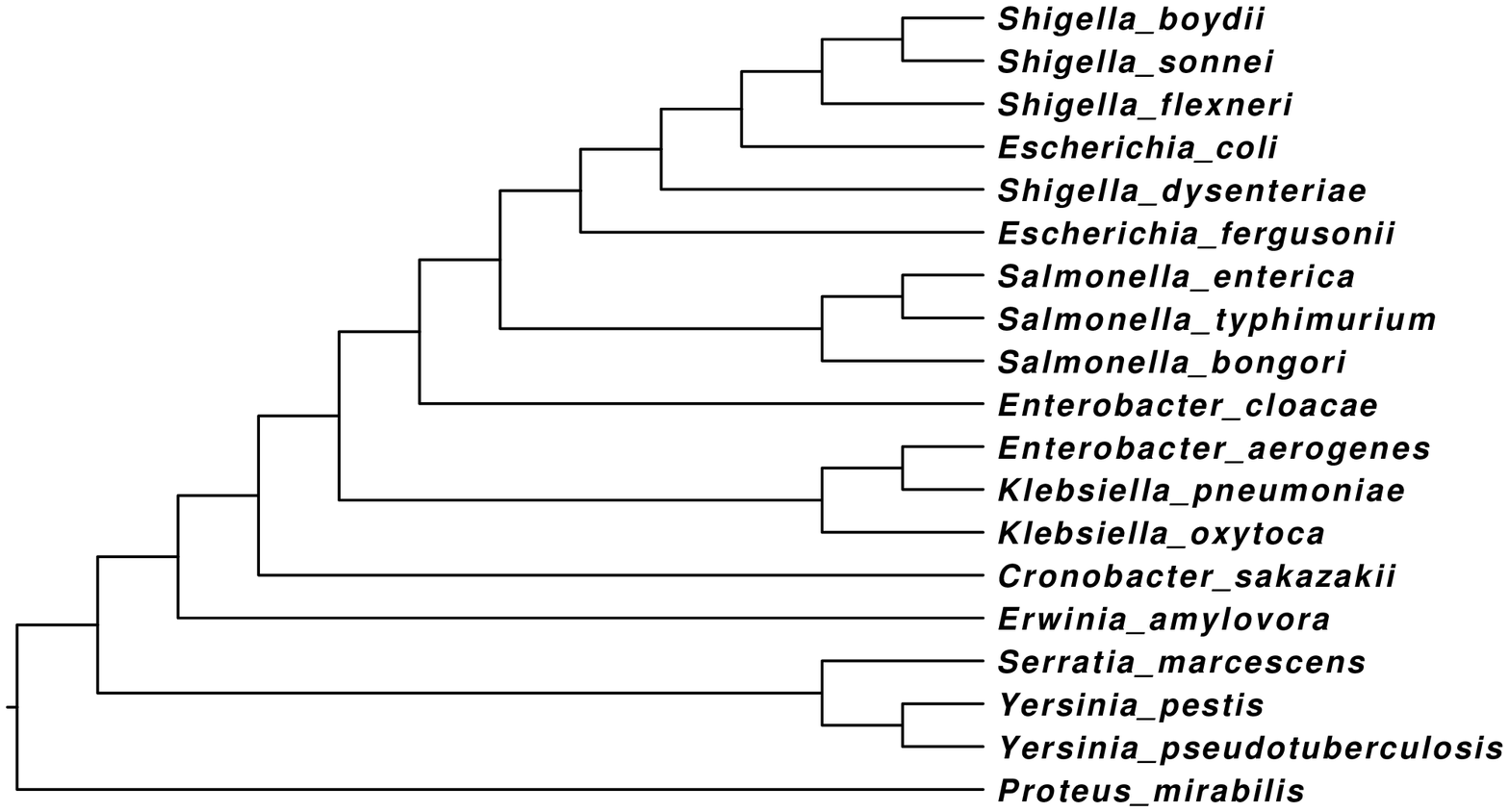}\\
\includegraphics[width=0.4\textwidth]{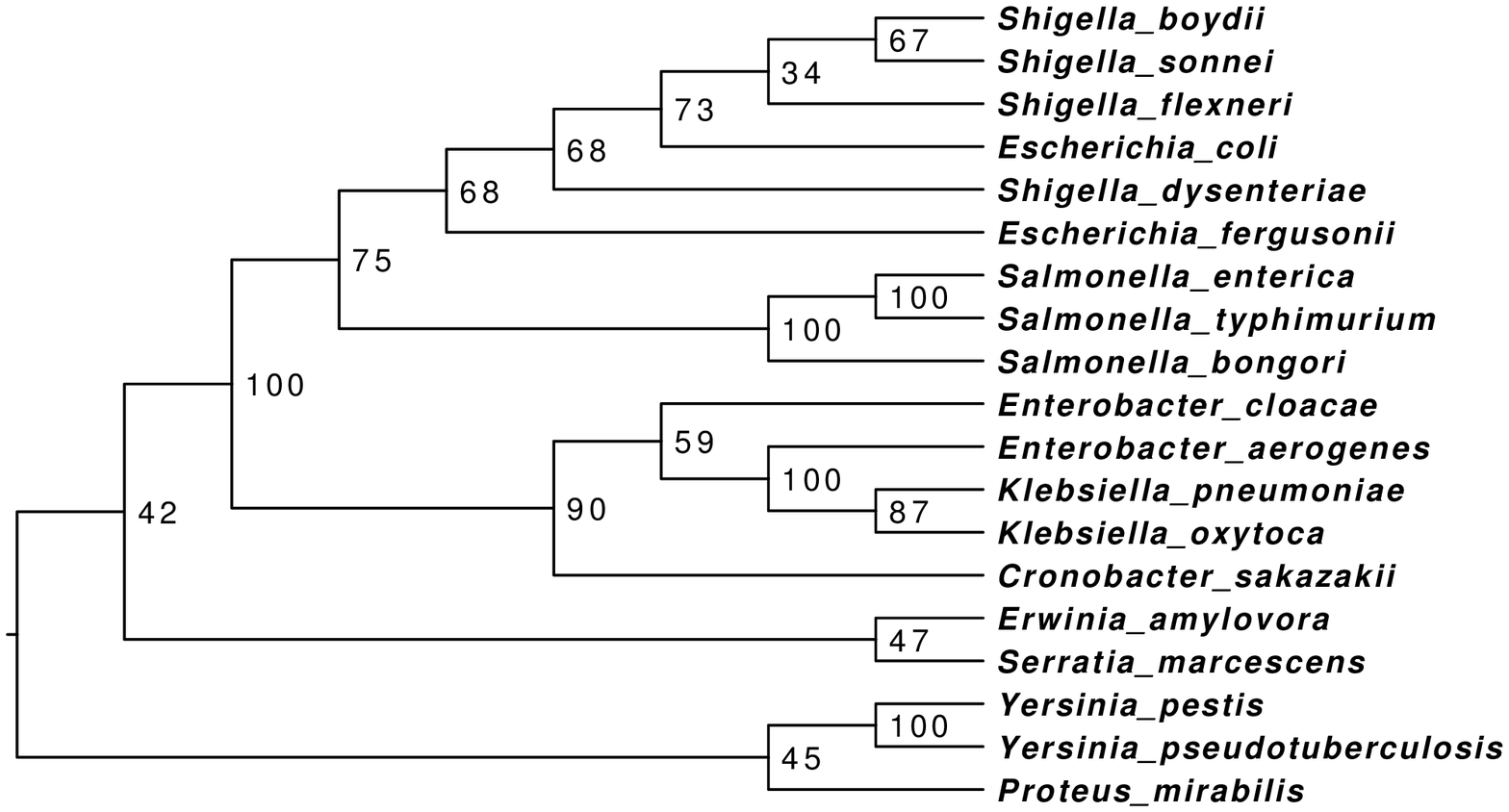}
\hspace{2em}
\includegraphics[width=0.4\textwidth]{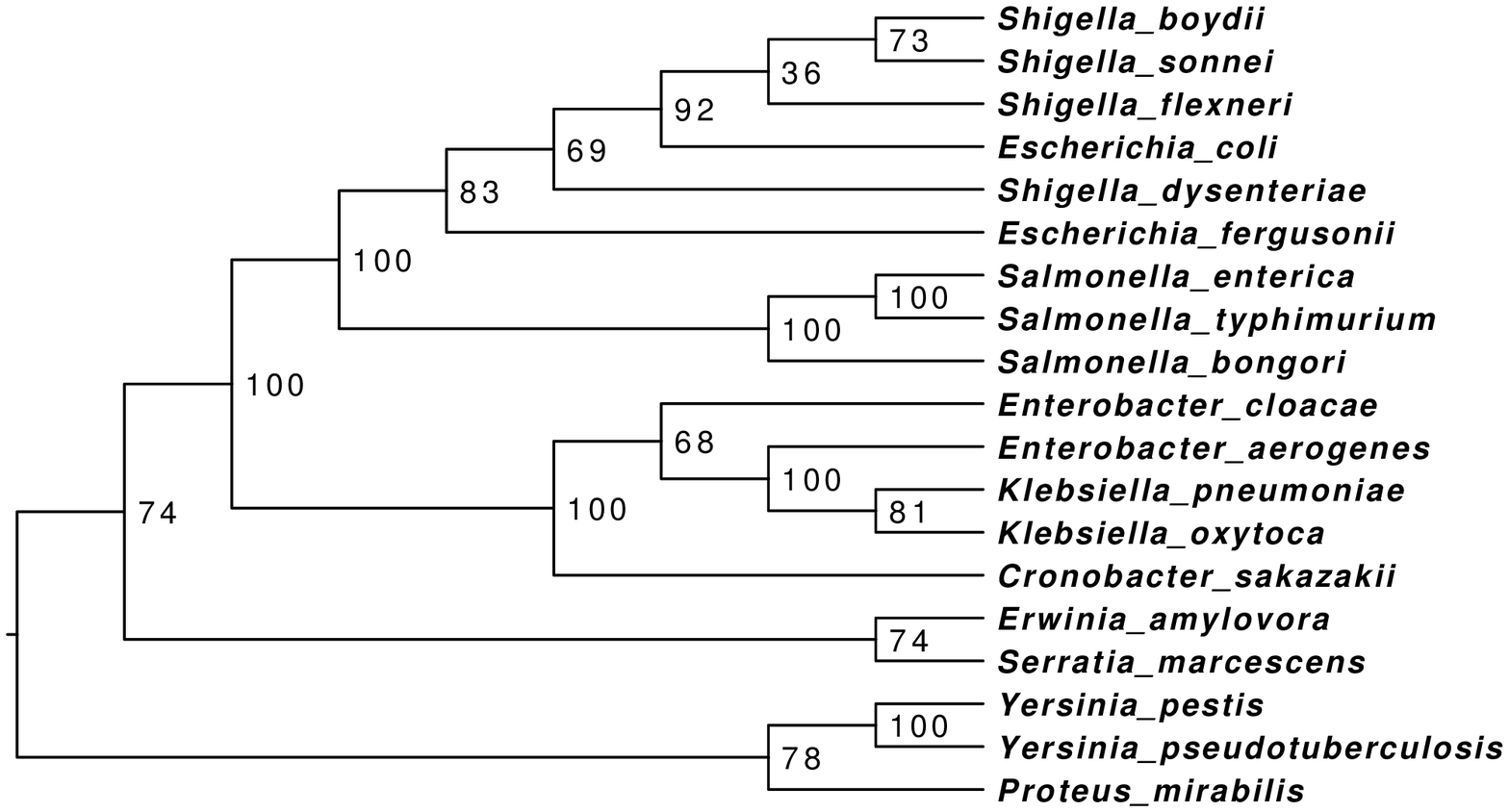}\\
\end{center}
\caption{Phylogenetic trees of 19 \emph{Enterobacteriales} species. Top, L.h.s.:
  tree computed from paralogy data. Internal node labels indicate support
  of subtrees. Top, R.h.s.: reference tree from \texttt{PATRIC} database,
  projected to the 19 considered species. \emph{Salmonella typhimurium} is
  missing in \texttt{PATRIC} tree.
  Cograph-based (bottom, l.h.s.) and triple-based (bottom, r.h.s.) bootstrapping
  trees of the 19 \emph{Enterobacteriales} species.
}
\label{fig:enterobacteriales}
\end{figure*}


\clearpage
\bibliographystyle{unsrt}
\bibliography{biblio}

\end{document}